\documentclass[12pt]{article}
\usepackage{afterpage}
\usepackage{amsfonts}
\usepackage{amsmath}
\usepackage{amssymb}
\usepackage{bbm}
\usepackage[usenames,dvipsnames]{color}
\usepackage{float}
\usepackage{graphicx}
\usepackage{epstopdf}
\usepackage{lscape}
\usepackage{natbib}
\usepackage{ntheorem}
\usepackage{rotating}
\usepackage{tabularx}
\usepackage{times}

\makeatletter
\renewcommand{\@seccntformat}[1]{{\normalfont
\csname the#1\endcsname .}\hspace{.3cm}}
\renewcommand{\section}{\@startsection
{section}{1} {0mm} {-\baselineskip}{0.2\baselineskip} {\bf}}
\renewcommand{\subsection}{\@startsection
{subsection}{1} {0mm} {-\baselineskip}{0.2\baselineskip} {\bf}}
\makeatother

\newcommand{\Perp}{\perp\!\!\!\perp}
\newcommand{\qed}{\hfill $\blacksquare$}
\numberwithin{equation}{section}
\renewcommand{\baselinestretch}{1.0}
\allowdisplaybreaks

\theoremstyle{plain}
\newtheorem{theorem}{Theorem}[section]
\newenvironment{proof}[1][Proof]{\begin{trivlist}
\item[\hskip \labelsep {\bfseries #1.}]}{\end{trivlist}}
\newtheorem{proposition}[theorem]{Proposition}
\newtheorem{lemma}[theorem]{Lemma}
\newtheorem*{remark}{Remark}

\begin{document}
\title{Inference from high-frequency data:\\ A subsampling approach\thanks{
We thank Yacine A\"{i}t-Sahalia (co-managing editor), Ilze Kalnina, Per Mykland, two anonymous referees, participants at the ``Market Microstructure and High-Frequency Data'' meeting in Chicago, USA; the ``High-Frequency Financial Econometrics'' conference in Barcelona, Spain; the AHOI workshop at Imperial College, London, UK; the Dynstoch workshop in Lund, Sweden; the ``Econometrics of High-Dimensional Risk Networks'' conference at the Stevanovich Center in Chicago, USA; the ``9th International Conference on Computational and Financial Econometrics'' in London, UK; the ``9th Annual SoFiE Conference'' in Hong Kong, China and in seminars at U. of Padova, Venice and Verona for helpful comments and suggestions. Christensen and Podolskij appreciate research funding from the Danish Council for Independent Research (DFF -- 4182-00050). In addition, Podolskij thanks the Villum Foundation for supporting his research on ``Ambit fields: Probabilistic properties and statistical inference.'' Thamrongrat gratefully acknowledges money from Deutsche Forschungsgemeinschaft through the Research Training Group (RTG -- 1953). This work was also supported by CREATES, which is funded by the Danish National Research Foundation (DNRF78). Please address correspondence to: kim@econ.au.dk.}}
\author{K. Christensen\thanks{CREATES, Department of Economics and Business Economics, Aarhus University, Fuglesangs All\'{e} 4, 8210 Aarhus V, Denmark.}
\and M. Podolskij\thanks{Department of Mathematics, Aarhus University, Ny Munkegade 118, 8000 Aarhus C, Denmark.} $^{\text{,}}$ \kern-0.15cm \footnotemark[2]
\and N. Thamrongrat\thanks{Institute of Applied Mathematics, Heidelberg University, Im Neuenheimer Feld 294, 69120 Heidelberg, Germany.}
\and B. Veliyev\footnotemark[2]}
\date{August, 2016}
\maketitle

\begin{abstract}
In this paper, we show how to estimate the asymptotic (conditional) covariance matrix, which appears in central limit theorems in high-frequency estimation of asset return volatility. We provide a recipe for the estimation of this matrix by subsampling; an approach that computes rescaled copies of the original statistic based on local stretches of high-frequency data, and then it studies the sampling variation of these. We show that our estimator is consistent both in frictionless markets and models with additive microstructure noise. We derive a rate of convergence for it and are also able to determine an optimal rate for its tuning parameters (e.g., the number of subsamples). Subsampling does not require an extra set of estimators to do inference, which renders it trivial to implement. As a variance-covariance matrix estimator, it has the attractive feature that it is positive semi-definite by construction. Moreover, the subsampler is to some extent automatic, as it does not exploit explicit knowledge about the structure of the asymptotic covariance. It therefore tends to adapt to the problem at hand and be robust against misspecification of the noise process. As such, this paper facilitates assessment of the sampling errors inherent in high-frequency estimation of volatility. We highlight the finite sample properties of the subsampler in a Monte Carlo study, while some initial empirical work demonstrates its use to draw feasible inference about volatility in financial markets.


\end{abstract}

\vfill

\thispagestyle{empty}
\pagebreak

\section{Introduction}

\setcounter{page}{1} \renewcommand{\baselinestretch}{1.6} \normalsize

Volatility is a key ingredient in the assessment and prediction of financial risk, be it in asset- and derivatives pricing \citep*[e.g.,][]{black-scholes:73a, sharpe:64a}, portfolio
selection \citep*[e.g.,][]{markowitz:52a}, or risk management and hedging \citep*[e.g.,][]{jorion:06a}.

Around the turn of the millennium, the advent of financial high-frequency data led to a surge in the nonparametric measurement of volatility \citep*[see, e.g.,][]{andersen-bollerslev-diebold:10a, barndorff-nielsen-shephard:07a}. High-frequency data are recorded at the tick-by-tick level and store information about the time, price (i.e., a bid-ask quote or transaction price), and size of individual orders and executions. In theory, the harnessing of high-frequency information leads to a perfect, error-free measure of ex-post volatility via the realized variance; a sum of squared intraday log-returns \citep*[e.g.,][]{andersen-bollerslev:98a, barndorff-nielsen-shephard:02a}.

After the initial---pioneering---work, the literature turned towards addressing two inherent shortcomings of realized variance. Firstly, realized variance can only estimate quadratic variation, and it does not separate continuous, diffusive volatility from discontinuous jump risk. This motivated the development and application of estimators that can robustly measure very general functionals of volatility, also in the presence of jumps \citep*[e.g.,][]{ait-sahalia-jacod:12a, andersen-dobrev-schaumburg:12a, barndorff-nielsen-shephard:04b, barndorff-nielsen-graversen-jacod-podolskij-shephard:06a, christensen-oomen-podolskij:10a, corsi-pirino-reno:10a, mancini:09a, mykland-shephard-sheppard:12a}. Secondly, when applied to tick-by-tick data realized variance is severely biased by common sources of noise, which form an integral part of any realistic model for securities' prices  \citep*[e.g.,][]{hansen-lunde:06b, zhou:96a}. This paved the way for the next cohort of estimators that were designed to be more resistant to noise, e.g., \citet*{ait-sahalia-mykland-zhang:05a, barndorff-nielsen-hansen-lunde-shephard:08a, jacod-li-mykland-podolskij-vetter:09a, podolskij-vetter:09a, zhang:06a}. As such, much progress has been made and today there is no shortage of estimators that can provide consistent estimates of volatility functionals in various contexts (either plain vanilla, or robustly to jumps or noise---or both).

The large battery of estimators at our disposal also brings with it an increasing demand for assessing estimation errors and drawing inference about volatility---e.g., in the form of confidence intervals or hypothesis tests. This is because whether the sample is small or large, as long as it is finite, there is necessarily some sampling error left in the estimate,
and when confidence intervals are computed in practice, high-frequency estimators of volatility are often found to contain sizable errors \citep*[e.g.,][]{barndorff-nielsen-hansen-lunde-shephard:08a}. The distinction between a realized measure of volatility and its population target is critical, because failing to properly take sampling uncertainty into account can severely distort parameter estimation of stochastic volatility models and be detrimental to the construction and evaluation of forecasts of volatility \citep*[e.g.,][]{andersen-bollerslev:98a, andersen-bollerslev-meddahi:05a,andersen-bollerslev-meddahi:11a, hansen-lunde:06a,hansen-lunde:14a,patton:11a}.

There are several problems associated with drawing inference about volatility functionals in high-frequency data. The first and foremost is of course to figure out the relevant distribution theory. The next hurdle is then to find a good proxy for the asymptotic (conditional) variance of the estimator. This is a formidable challenge in practice, because the asymptotic variance often depends on parameters that are substantially more difficult to back out from the available sample of high-frequency data. The expression for the asymptotic variance typically also rests heavily on the properties of the data and it is bound to change depending on these. This is an unpleasant concern with real high-frequency data, which are contaminated by market microstructure frictions. While the noise is often assumed to be i.i.d. and independent of the efficient price, there is some empirical and theoretical support for a serially correlated, heteroscedastic, and, potentially, endogenous noise process at the tick level \citep*[e.g.,][]{ait-sahalia-mykland-zhang:11b,diebold-strasser:13a, hansen-lunde:06b,kalnina-linton:08a}. An estimator of the asymptotic variance designed for i.i.d. and independent noise can not be expected to give valid inference, if the underlying conditions are violated. In practice, it is not trivial to verify the conditions imposed on the noise \citep*[e.g.,][]{hautsch-podolskij:13a}, which makes it more pressing to find estimators that are robust against modeling criteria. Finally, in multivariate analysis, inference would at some stage require an estimate of the asymptotic covariance matrix. Here, the proposed estimator should ideally be positive semi-definite, while, in contrast, some existing estimators of the asymptotic covariance matrix in the high-frequency setting are not assured to be that. As we show in this paper, this runs smack into problems in practically relevant and realistic settings (see Table \ref{Table:Proportion} in Section \ref{Sec:Simulation}).

In this paper, we propose to use subsampling for assessing the uncertainty embedded in high-frequency estimation of functionals of financial volatility. Subsampling is based on creating several---properly rescaled---estimates of the parameter(s) of interest using local stretches of sample data and then studying the sampling variation of these. It was originally developed in the context of stationary time series in the long-span domain \citep*[e.g.,][]{politis-romano:94a, politis-romano-wolf:99a}. The term appeared in the high-frequency literature in \citet*{zhang-mykland-ait-sahalia:05a}, who proposed a two-scale realized variance based on price subsampling. This is different from traditional subsampling and actually does not work for asymptotic variance estimation, because it leads to an overlapping samples problem in the subsampled returns, causing the subsample estimates to be too strongly correlated in large samples. This was pointed out by \citet*{kalnina-linton:07a} and \citet*{kalnina:11a}, who propose an inference strategy based on various alternative subsampling schemes, which lead to better asymptotic properties. \citet*{kalnina:15a} extends these ideas to inference about a multivariate parameter, while \citet*{ikeda:16a} and \citet*{varneskov:16a} consider subsample estimation of the asymptotic variance of the realized kernel.

As an inferential tool, subsampling has several attractive features from a practical point of view. First, subsampling is intuitive and relatively easy to compute, because it does not require an extra set of estimators; it uses copies of the original statistic. Second, in the multivariate context, it leads to variance-covariance matrix estimates that are positive semi-definite by construction. And, third, subsampling does not explicitly take the structure of the asymptotic variance into account. It is to a large degree automatic and has an innate ability to adapt to the problem at hand, which makes it highly robust against design criteria, as shown by \citet*{kalnina:11a}. This type of analysis, where inference is effectively carried out by bypassing the asymptotic variance, is also emphasized by \cite{mykland-zhang:17a}, who propose a so-called Observed Asymptotic Variance, which, as our approach, is based on the comparison of adjacent estimators.

This paper builds on these ideas. It contributes to extant literature in several directions. First, we propose to subsample bipower variation as a means to
estimate the asymptotic variance-covariance matrix of this statistic. We devise an estimator, which involves fewer tuning parameters compared to \citet*{kalnina:11a}. Second, we derive an asymptotic theory within this framework in both frictionless and noisy markets. We show our estimator is consistent under weak assumptions on the data-generating process, accommodating jumps in the price and volatility, while allowing the noise to be either heteroscedastic or autocorrelated. Third, with stronger conditions, we provide a decomposition of
the leading errors of the subsampler, from which we get insights about how to configure it by optimally choosing its tuning parameters (e.g., the number of subsamples). This yields a rate of convergence for our statistic; a result that has---to the best of our knowledge---not been derived in earlier work. It reveals that the robustness of subsampling is not free of charge, but leads to a loss of efficiency compared to existing estimators in the form of a slower rate of convergence. It implies a trade-off in that if, for example, one is prepared
to use an estimator, which is not positive semi-definite, a better rate can potentially be achieved. Or, if prior knowledge about the asymptotic variance matrix is available or parametric assumptions can be verified from the data, it is typically better to construct estimators which attempt to exploit that information relative to doing subsampling.\footnote{To paraphrase \citet*{politis-romano-wolf:99a}, subsampling is ``a robust starting point toward even more refined procedures.''} Still, in finite samples we show in a realistic setting with microstructure noise the subsampler produces convincing results compared to some available alternatives.

The rest of this paper goes as follows. Section 2 introduces the setting. In Section 3, we derive the theory first without and then with noise. In Section \ref{Sec:Simulation}, we do numerical simulations in order to inspect the finite sample performance of our estimator. In Section \ref{Sec:Empirical}, we confront our framework with some real high-frequency data, while the Appendix contains the proofs of our results.

\section{Theoretical framework}

We consider a scalar process $X = (X_{t})_{t \geq 0}$, which represents the log-price of some financial security. It is defined on a filtered probability space $(\Omega, \mathcal{F}, (\mathcal{F}_{t})_{t\geq 0}, \mathbb{P})$ and adapted to $(\mathcal{F}_{t})_{t \geq 0}$. We assume that $X$ can be described by a continuous It\^{o} semimartingale (or stochastic volatility model), as expressed by the equation:
\begin{equation}
\label{Eqn:X}
X_{t} = X_{0} + \int_{0}^{t} a_{s} \text{d}s + \int_{0}^{t} \sigma_{s}
\text{d} W_{s},
\end{equation}
where $X_{0}$ is the starting price, $a = (a_{t})_{t \geq 0}$
is a predictable and locally bounded drift process, $\sigma = (\sigma_{t})_{t
\geq 0}$ is an adapted, c\`{a}dl\`{a}g volatility process, while $W =
(W_{t})_{t \geq 0}$ is a standard Brownian motion.\footnote{A basic result
in financial economics states that if $X$ is drawn from an arbitrage-free,
frictionless market, it necessarily has a semimartingale structure
\citep*[e.g.,][]{back:91a,delbaen-schachermayer:94a}. If further $X$ has
continuous sample paths, then weak assumptions
ensure that $X$ can be represented as in Eq. \eqref{Eqn:X}. Later on,
in Section \ref{section:truncation}, we
provide further details about the robustness of our results, when $X$
exhibits jumps. We cover the noisy
setting with market frictions in Section \ref{Sec:Noise}.}

In these models, natural measures of ex-post volatility can be written,
for some suitable function $f$,
\begin{equation}
\label{Eqn:IV}
IV(f)_{t} = \int_{0}^{t} f(\sigma_{s}) \text{d}s,
\end{equation}
i.e. integrated functions of the diffusion coefficient.

We also point out that $a$ and $\sigma$ are left unspecified in this paper, and our results are completely nonparametric (within this class of models). While we do not impose assumptions a priori, we sometimes need to add additional, weak regularity conditions on $\sigma$, which nevertheless allow for very complex dynamics in the volatility process.\\[-0.50cm]

\noindent \textbf{Assumption (V)}: $\quad \sigma$ is of the form:
\begin{align}
\begin{split}
\label{Eqn:sigma}
\sigma_{t} &= \sigma_{0} + \int_{0}^{t} \tilde{a}_{s}\text{d}s + \int_{0}^{t}
\tilde{ \sigma}_{s}\text{d}W_{s} + \int_{0}^{t} \tilde{v}_{s}\text{d}B_{s}
\\[0.25cm]
&+ \int_{0}^{t} \int_{E} \tilde{ \delta}(s,x) 1_{ \{| \tilde{ \delta}(s,x)|
\leq 1 \}} (\tilde{ \mu} - \tilde{ \nu}) ( \text{d}s, \text{d}x) + \int_{0}^{t}
\int_{E} \tilde{ \delta}(s,x) 1_{ \{| \tilde{ \delta}(s,x)| > 1 \}}
\tilde{ \mu}( \text{d}s, \text{d}x),
\end{split}
\end{align}
where $\sigma_{0}$ is its initial value, $\tilde{a} = (\tilde{a}_{t})_{t
\geq 0}$, $\tilde{ \sigma} = (\tilde{ \sigma}_{t})_{t \geq 0}$ and
$\tilde{v} = (\tilde{v}_{t})_{t \geq 0}$ are
adapted, c\`{a}dl\`{a}g stochastic processes, while $B = (B_{t})_{t \geq 0}$
is a standard Brownian motion that is independent of $W$.
Furthermore, $(E, \mathcal{E})$ is a Polish space, $\tilde{ \mu}$ is a random
measure on $\mathbb{R}_{+} \times E$, which is independent of $(W,B)$ and has
an intensity measure $\tilde{ \nu}( \text{d}s, \text{d}x) = \text{d}s
\tilde{F}( \text{d}x)$, where $\tilde{F}$ is a $\sigma$-finite measure on
$(E, \mathcal{E})$. Also, $\tilde{ \delta}: \Omega \times \mathbb{R}_{+}
\times E \to \mathbb{R}$ is a predictable function and $(S_{k})_{k \geq 1}$
is a sequence of stopping times increasing to $\infty$ such that $| \tilde{ \delta}(
\omega,s,x)| \wedge 1 \leq \tilde{ \psi}_{k}(x)$ for all $( \omega, s, z)$ with
$s \leq S_{k}( \omega)$ and $\int_{E} \tilde{ \psi}^{2}_{k}(x) \tilde{F}(
\text{d}x) < \infty$ for all $k \geq 1$. \\[-0.50cm]

We are in the high-frequency setting. We suppose that historical data of
$X$ is available in the time frame $[0,1]$, i.e. we set $t = 1$. In this
interval, we assume that $X$ is recorded at equidistant time points $t_{i}
= i / n$, for $i = 0, 1, \ldots n$, so that $n + 1$ is the total number of
log-price observations in the sample. We define the $n$ increments, or
log-returns, of $X$ as:
\begin{equation}
\label{Eqn:return}
\Delta_{i}^{n} X = X_{i/n} - X_{(i-1)/n}, \quad \text{for }
i = 1, \ldots, n.
\end{equation}
The asymptotic theory we derive below is then infill, i.e. we are at some
point going to let $n \to \infty$.

To maintain a streamlined exposition, we only study the
univariate setting in this paper. All of our theoretical results extend
directly and without any changes (apart from additional notation) to
multivariate $X$, if the sampling times are equidistant, as above, and
recorded synchronously
across assets, see, e.g., \citet*{kalnina:15a} for related research
in that direction. If the high-frequency data are randomly spaced and
asynchronous, our results are probably still true using a previous-tick
rule to align prices and under suitable assumptions on the regularity
on the observation grid of individual assets, e.g.,
\citet*{christensen-podolskij-vetter:13a}.

\subsection{Bipower variation}

The econometric challenge is that the objects appearing in Eq.
\eqref{Eqn:IV} are latent, but they can be estimated from the available
sample of high-frequency data in Eq. \eqref{Eqn:return}. A popular
statistic, which is well-suited to do this, is the bipower variation of
\citet*{barndorff-nielsen-shephard:04b}.\footnote{Note that bipower
variation is nested within
a broader class of multipower variations, which adds additional lags in
Eq. \eqref{Eqn:BVn}, see, e.g, \citet*{barndorff-nielsen-shephard:04b}.
Our theoretical results extend to this framework.} Here,
we adopt the more general definition of bipower variation from
\citet*{barndorff-nielsen-graversen-jacod-podolskij-shephard:06a}:\footnote{As
we frequently cite this paper, it will henceforth be referred to as BGJPS6.}
\begin{equation}
\label{Eqn:BVn}
V(f,g)^{n} = \frac{1}{n} \sum_{i = 1}^{n - 1} f\bigl( \sqrt{n}
\Delta_{i}^{n} X \bigr) g\bigl( \sqrt{n} \Delta_{i + 1}^{n} X \bigr),
\end{equation}
where $f = (f_{1}, \ldots, f_{m})'$ and $g = (g_{1}, \ldots, g_{m})'$
are $\mathbb{R}^{m}$-valued functions. Note that in Eq. \eqref{Eqn:BVn}
the multiplication is understood to be done element-by-element.
Moreover, the components of $f = (f_{1}, \dots, f_{m})'$ and
$g = (g_{1}, \ldots, g_{m})'$ are assumed to fulfill
a condition, which we state with a generic function $h$.\\[-0.50cm]

\noindent \textbf{Assumption (K)}: $\quad$ $h: \mathbb{R} \mapsto
\mathbb{R}$ is even and continuously differentiable. Moreover, both
$h$ and its derivative $h'$ are at most of polynomial growth.\\[-0.50cm]

Assumption (\textbf{V}) and (\textbf{K}) are standard conditions for
the validity of central limit theorems for classical high-frequency statistics;
see, e.g., BGJPS6.
The following proposition, which is adapted from that paper, then describes the
limiting properties of bipower variation.

\begin{proposition}
\label{Thm:BvCLT}
Assume that $X$ is a continuous It\^{o} semimartingale as in Eq. \eqref{Eqn:X},
where the volatility process $\sigma$ follows Assumption $(\normalfont
\textbf{V})$ and Assumption $(\normalfont \textbf{K})$ is true for each
component of $f = (f_{1}, \dots, f_{m})'$ and $g = (g_{1}, \ldots, g_{m})'$.
Then, as $n \to \infty$, it holds that
\begin{equation}
\label{Eqn:MN}
\sqrt{n} \Bigl( V(f,g)^{n} - V(f,g) \Bigr) \overset{d_{s}}{\to}
\text{\upshape{MN}}(0, \Sigma),
\end{equation}
where
\begin{equation}
\label{Eqn:BV}
V(f,g) = \int_{0}^{1} \rho_{\sigma_{s}}(f) \rho_{\sigma_{s}}(g)
\text{\upshape{d}}s,
\end{equation}
and ``$\overset{d_{s}}{\to}$'' means convergence in law stably
(as described below). Moreover,
$\rho_{x}(f) = \mathbb{E} \bigl[f(xZ) \bigr]$ for $x \in \mathbb{R}$
and $Z \sim \text{\upshape{N}}(0,1)$. Finally, $\Sigma$ is the $m
\times m$ asymptotic conditional covariance matrix, which has elements
\begin{align}
\begin{split}
\label{Eqn:Sigmaij}
\Sigma_{ij} = \int_{0}^{1} \Bigl[ \rho_{ \sigma_{s}} (f_{i} f_{j}) \rho_{
\sigma_{s}}(g_{i} g_{j}) &+ \rho_{ \sigma_{s}}(f_{i}) \rho_{ \sigma_{s}}
(g_{j}) \rho_{ \sigma_{s}}(f_{j} g_{i}) \\[0.25cm]
&+ \rho_{ \sigma_{s}}(f_{j}) \rho_{ \sigma_{s}}(g_{i}) \rho_{
\sigma_{s}}(f_{i} g_{j}) -3 \rho_{ \sigma_{s}}(f_{i}) \rho_{
\sigma_{s}}(f_{j}) \rho_{ \sigma_{s}}(g_{i}) \rho_{ \sigma_{s}}(g_{j})
\Bigr] \text{\upshape{d}}s.
\end{split}
\end{align}
\end{proposition}
\begin{proof}
See BGJPS6.
\end{proof}

The proposition provides the foundation for making
inference about bipower variation. It shows that $V(f,g)^{n}$ is consistent
for $V(f,g)$. The asymptotic distribution of $V(f,g)^{n}$ is
mixed normal, i.e. it has a random variance-covariance matrix $\Sigma$, which
is independent from the randomness of the normal distribution. To
transform this into a probabilistic statement based on the standard normal
distribution, it would be tempting to look at Eq. \eqref{Eqn:MN} and
deduce that, assuming $\Sigma$ is invertible, $\Sigma^{-1/2} \sqrt{n}
\Bigl( V(f,g)^{n}
- V(f,g) \Bigr) \overset{d}{\to} \text{\upshape{N}}(0, I_{m})$, where
``$\overset{d}{\to}$'' is convergence in law and $I_{m}$
is the $m$-dimensional identity matrix. We cannot readily jump to this
conclusion from the regular definition of convergence in law, however, unless
$\Sigma$ is constant. In our setting, where $\Sigma$ is stochastic, we therefore exploit
an alternative, stronger type of convergence, which recovers the above feature.

More formally, let $Z_{n}$ be a sequence of random variables. Also, let $Z$
be a random variable, which is defined on some appropriate extension $(\Omega',
\mathcal{F}',  \mathbb{P}')$ of the original space $(\Omega, \mathcal{F},
\mathbb{P})$. $Z_{n}$ is then said to converge stably to $Z$, and we
shall write $Z_{n} \overset{d_{s}}{\to} Z$, if for any bounded, continuous
function $h$ and any bounded, $\mathcal{F}$-measurable random variable $Y$,
the convergence $\mathbb{E} \bigl[Yh(Z_{n}) \bigr] \to \mathbb{E}' \bigl[Yh(Z)
\bigr]$ holds, as $n \to \infty$. It follows that stable convergence implies
convergence in law by taking $Y = 1$. We refer to \citet*{jacod-protter:12a}
for further details of this concept.

Proposition \ref{Thm:BvCLT} is based on Assumption (\textbf{K}),
which assumes $f$ and $g$ are differentiable. It can be extended to
non-differentiable functions, given Assumption (\textbf{H'}) and (\textbf{K'})
from BGJPS6. Assumption (\textbf{H'}) says $\sigma$ should be bounded away
from zero, while Assumption (\textbf{K'}) puts restrictions on the
set, where $f$ and $g$ are not differentiable. We refer to BGJPS6 for a
concise, mathematical statement of these conditions. Throughout the paper,
we explain how this weaker setting affects our results.

A prime example that falls in the latter group is the original bipower
variation of \citet*{barndorff-nielsen-shephard:04b}. It is based on the
summation of products of the absolute value of adjacent high-frequency returns
and sets $f_{k}(x) = |x|^{q_{k}}$ and $g_{k}(x) = |x|^{r_{k}}$, for
$1 \leq k \leq m$ and $q_{k}, r_{k} \geq 0$. As it is extensively used in
applied work, we sometimes restrict attention to this class of estimators
below. To distinguish this ``pure'' bipower variation from the general
one,
we write this version as $V(q,r)^{n}$, where $q = (q_{1},
\ldots, q_{m})'$ and $r = (r_{1}, \ldots, r_{m})'$ are $m$-dimensional
vectors, whose coordinates index the powers:
\begin{equation}
\label{Eqn:OBV}
V(q_{k},r_{k})^{n} = \frac{1}{n} \sum_{i=1}^{n-1} | \sqrt{n} \Delta_{i}^{n}
X|^{q_{k}} | \sqrt{n} \Delta_{i+1}^{n} X|^{r_{k}} \quad \overset{p}{ \to}
\quad V(q_{k}, r_{k}) =
\mu_{q_{k}} \mu_{r_{k}} \int_{0}^{1} | \sigma_{s}|^{q_{k}+r_{k}}\text{d}s,
\end{equation}
where $\mu_{q} = \mathbb{E}[|Z|^{q}]$ and $Z \sim \text{N}(0,1)$.
Note that while $f_{k}(x) = |x|^{q_{k}}$ and $g_{k}(x) = |x|^{r_{k}}$ are
not differentiable at $x = 0$, if $q_{k} \leq 1$ or $r_{k} \leq 1$,
Proposition \ref{Thm:BvCLT} is nevertheless still true, if $\sigma > 0$.

Here, $\Sigma$ has the form:
\begin{equation}
\label{SigmaPowers}
\Sigma_{ij} = \left( \mu_{q_{i}+q_{j}} \mu_{r_{i}+r_{j}}+ \mu_{q_{i}}
\mu_{r_{j}} \mu_{q_{j}+r_{i}}+\mu_{q_{j}} \mu_{r_{i}} \mu_{q_{i}+r_{j}}
- 3 \mu_{q_{i}} \mu_{q_{j}} \mu_{r_{i}} \mu_{r_{j}} \right)
\int_{0}^1 |\sigma_{s}|^{q_{i}+q_{j}+r_{i}+r_{j}} \text{\upshape{d}}s,
\end{equation}
and the link between Eq. \eqref{Eqn:IV} and Eq. \eqref{Eqn:return} is
made apparent.

\section{Subsample estimation of $\Sigma$}

In order to do inference about bipower variation based on Proposition
\ref{Thm:BvCLT}, we also need to estimate $\Sigma$. Although there are
some existing estimators out there, they are typically not general enough,
or else they are inherently flawed, as outlined below.

For example, assume the goal is to estimate the asymptotic conditional
variance-covariance matrix of the pure bipower variation estimator
$V(q_{k},r_{k})^{n}$. The structure of $\Sigma$ unveiled by Eq.
\eqref{SigmaPowers} implies that here we can
do estimation componentwise with a second bipower variation:
\begin{equation}
\label{Eqn:Alt1}
c_{ij} V(q_{i}+q_{j},r_{i}+r_{j})^{n} \overset{p}{ \to} \Sigma_{ij},
\end{equation}
where $c_{ij}$ is a constant that can be determined from Eq. \eqref{Eqn:OBV} --
\eqref{SigmaPowers}. This is a direct element-by-element calculation,
which also works for the general case of Eq. \eqref{Eqn:Sigmaij} with some
minor modifications. However, this procedure is difficult to adapt to the
noisy setting.

Another idea, which is more generally applicable, is to write
\begin{equation}
\gamma_{i}^{n}(k) \equiv f_{k} \big( \sqrt{n} \Delta_{i}^{n} X) g_{k}(
\sqrt{n} \Delta_{i+1}^{n} X),
\end{equation}
and set
\begin{equation}
\label{Eqn:Alt2}
\tilde{ \Sigma}_{ij}^{n} = \frac{1}{n} \sum_{l=2}^{n-3} \sum_{m=-1}^{1}
\big( \gamma_{l}^{n}(i) \gamma_{l+m}^{n}(j) - \gamma_{l}^{n}(i)
\gamma_{l+2}^{n}(j) \big).
\end{equation}
Then, $\tilde{ \Sigma}^{n} \overset{p}{ \to} \Sigma$ as given by Eq.
\eqref{Eqn:Sigmaij}. In contrast to Eq. \eqref{Eqn:Alt1}, the latter
estimator is ``automatic,'' as it does not depend on explicit knowledge
of the asymptotic covariance matrix. It works, because the summands
$\gamma_{i}^{n}(k)$ are asymptotically 1-dependent, and the estimator
implicitly incorporates covariances up to lag one (hence the index bounds
in the last summation). Indeed, we compare a
vectorized version of $\tilde{ \Sigma}^{n}$ to our subsampler in Section
\ref{Sec:BipowerVariation}. But, while both of the above estimators are
$\sqrt{n}$-consistent, neither is positive semi-definite in finite samples.
In the noisy setting, this type of construction has been proposed by
\citet*{podolskij-vetter:09a}, which we thoroughly analyze in Section
\ref{Sec:Simulation}.

In this paper, we propose to estimate $\Sigma$ by subsampling of high-frequency data.  This was suggested for the asymptotic variance of realized variance in \citet*{kalnina-linton:07a} and its noise-robust two-scale version in \citet*{kalnina:11a}, following earlier work in the classical time series literature \citep*[e.g.,][]{politis-romano:94a,politis-romano-wolf:99a}. Subsampling is an attractive procedure. First, all it does is to compute rescaled copies of the original statistic based on local stretches of high-frequency data, and then it studies the sampling variation of these. This makes it highly intuitive and trivial to implement. Second, it leads to estimators of $\Sigma$ that are positive semi-definite by construction. This is critical in practice, as shown in Section \ref{Sec:Simulation}. And third, it has an innate robustness against the statistical properties of microstructure noise, as discussed in Section \ref{Sec:Noise}.

\subsection{Subsampling for power variation}

\label{sec:pv} In order to develop a complete theory for positive semi-definite estimation of the asymptotic, conditional covariance matrix $\Sigma$ of bipower variation using subsampling, we will start by analyzing the power variation case, i.e. $f = (f_{1}, \ldots, f_{m})'$. This should not only help to build intuition for the general case, but it turns out the results are interesting in their own right, because the rates of convergence differ according to whether we consider power variation, or multipower variation of order bigger than one.

Thus, we define the power variation estimator:
\begin{equation}
V(f)^{n} = \frac{1}{n} \sum_{i = 1}^{n} f \bigl( \sqrt{n} \Delta_{i}^{n} X \bigr).
\end{equation}
Then, we propose to set:
\begin{equation}
\label{Eqn:SigmaN}
\hat{ \Sigma}_{n} = \frac{1}{L} \sum_{l = 1}^{L} \biggl( \sqrt{
\frac{n}{L}} \Bigl( V_{l}(f)^{n} - V(f)^{n} \Bigr) \biggr) \biggl( \sqrt{
\frac{n}{L}} \Bigl( V_{l}(f)^{n} - V(f)^{n} \Bigr) \biggr)',
\end{equation}
where, assuming $L$ divides $n,$
\begin{equation}
V_{l}(f)^{n} =\frac{1}{n/L} \sum_{i = 1}^{n/L} f \bigl( \sqrt{n}
\Delta_{(i - 1 )L + l}^{n} X \bigr).
\end{equation}

In Figure \ref{Fig:subsample}, we illustrate the construction of the subsamples for power variation. As the figure shows, we start by splitting the full sample of high-frequency data into $L$ smaller subsamples. We assign log-returns successively to each subsample, starting from the first and going back to it after the $L$th subsample has been reached. We continue this process until the entire sample is exhausted. The $l$th subsample therefore consists of the increments $\bigl( \Delta_{(i - 1)L + l}^{n} X \bigr)_{i = 1, \ldots, n/L}$.

\begin{figure}[ht!]
\begin{center}
\caption{Infill return subsampling for power variation.\label{Fig:subsample}}
\begin{tabular}{c}
\includegraphics[height=10cm,width=15cm]{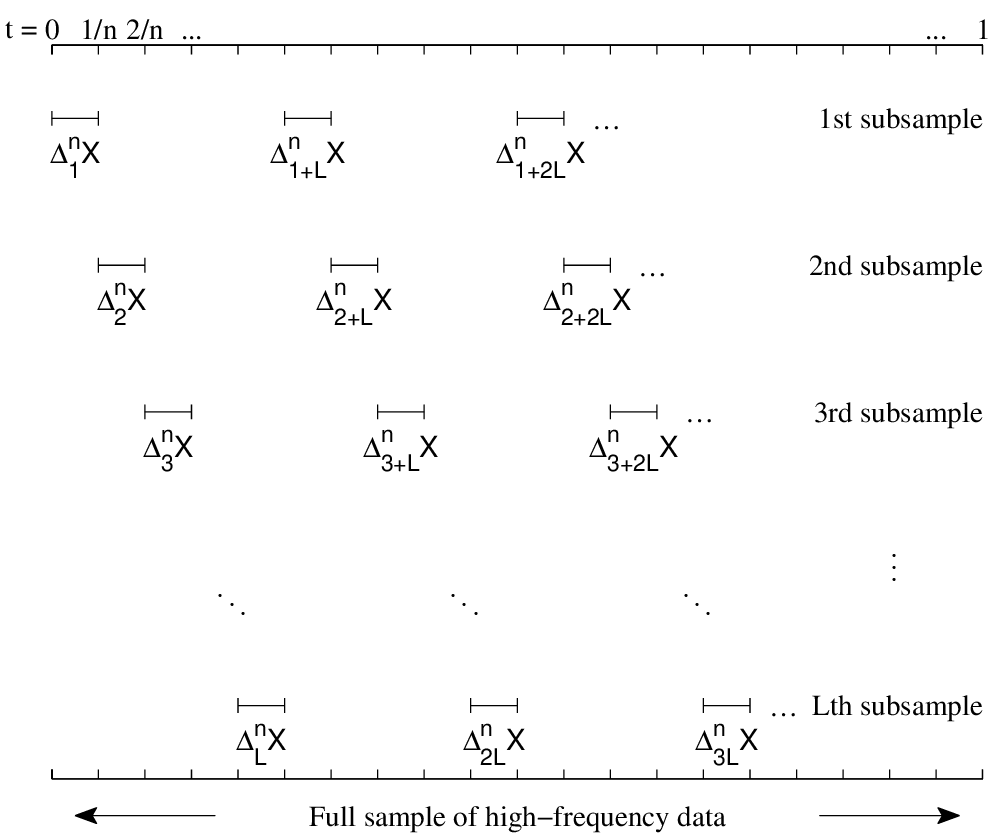}\\
\end{tabular}
\begin{small}
\parbox{0.925\textwidth}{\emph{Note.} The figure shows how the full sample of
available high-frequency data $(\Delta_{i}^{n} X)_{i = 1}^{n}$ is split into
smaller samples of size $n / L$ in order to compute subsampled
estimates of power variation.}
\end{small}
\end{center}
\end{figure}

We then compute the statistics $V_{l}(f)^{n}$, which is the power variation
estimator based on the $l$th subsample.
To estimate $\Sigma$, we look at the sampling variation of the
$V_{l}(f)^{n}$'s, after they have been suitably demeaned and scaled, with
a standard sample covariance matrix.

As we prove in a moment, we can choose the number of subsamples $L$, such that
$\hat{ \Sigma}_{n}$ is a consistent estimator of $\Sigma$. Intuitively, under
suitable conditions on $L$, $V_{l}(f)^{n} \overset{p}{\to} V(f)$, where
$V(f) = \int_{0}^{1} \rho_{\sigma_{s}} (f) \text{d}s$, and its asymptotic
distribution (more or less) follows from Proposition \ref{Thm:BvCLT}, except
that its rate of convergence is $(n/L)^{-1/2}$, i.e. $\displaystyle
\sqrt{\frac{n}{L}} \Bigl(V_l(f)^{n} - V(f) \Bigr) \overset{d_{s}}{\to}
\text{\upshape{MN}}(0,\Sigma)$. Moreover, as each subsample is based on
non-overlapping increments, the $V_{l}(f)^{n}$'s are, asymptotically,
conditionally independent. This suggests that by averaging the sum of outer
products of $\displaystyle \sqrt{\frac{n}{L}} \Bigl( V_{l}(f)^{n} - V(f)
\Bigr)$, we should get a consistent
estimator of $\Sigma$, which is then positive semi-definite by construction.

While $V_{l}(f)^{n}$ should in principle be centered around $V(f)$, the whole
problem to begin with is of course that $V(f)$ is latent. So it has to be
replaced by a consistent estimator to get a feasible estimator of $\Sigma$
that can be computed from data. In Eq. \eqref{Eqn:SigmaN}, we
plug in the full sample power variation estimate $V(f)^{n}$. This
does not affect the asymptotics, because $V(f)^{n}$ converges much
faster than $V_{l}(f)^{n}$.

We pause here for a moment to reflect a bit on the setup. It turns out that the
assumptions of Proposition \ref{Thm:BvCLT} are, by construction,
necessary to prove the consistency of $\hat{ \Sigma}_{n}$. Moreover, as we
show later, they are also sufficient, if $L \rightarrow \infty$ and $n/L
\rightarrow \infty$. Indeed, an underlying principle of this paper is
that, as long as an associated central limit theorem holds, subsampling can be
used to consistently estimate the asymptotic covariance matrix under minimal
conditions on the tuning parameters.
To be able to derive a convergence rate for
$\hat{ \Sigma}_{n}$ and provide an optimal choice of the parameter $L$, however,
we need some additional structure, which is not standard in the high-frequency
literature.

First, we are going to assume that all of the driving terms in both $X$ and
$\sigma$ can be modeled as Brownian semimartingales, i.e. we shall require
that: \\[0.50cm]
\noindent \textbf{Assumption (H)}: $\quad$ $\sigma$ is continuous and follows
Assumption (\textbf{V}),
and each of $a$, $\tilde{a}$, $\tilde{ \sigma}$ and $\tilde{v}$
is continuous  of the form in Eq. \eqref{Eqn:sigma}.\\[-0.50cm]

The second condition we impose is a highly technical requirement that concerns
the Malliavin smoothness of the random variables appearing in Assumption
(\textbf{H}). We summarize some selected elements and notations from
Malliavin calculus,
which are relevant to our paper, in Section \ref{secA.5} at the end of the
Appendix, while the list of necessary conditions for the main text are
comprised as Assumption \textbf{(M)}.\\[-0.50cm]

\noindent \textbf{Assumption (M)}: $\quad$ It holds that for any $0
\leq t \leq r \leq s$: $\sigma_{s}$, $\tilde{ \sigma}_{s}$, $\tilde{v}_{s}$,
$D_{t}( \sigma_{s})$, $D_{t}( \tilde{ \sigma}_{s})$, $D_{t}( \tilde{v}_{s})
\in \mathbb{D}_{1,2}$ with
\begin{align}
\begin{split}
\label{h1}
|| D_{t}( \sigma_{s}) ||_{L^{32}} + || D_{t}( \tilde{ \sigma}_{s}) ||_{L^{32}}
+ || D_{t}( \tilde{v}_{s}) ||_{L^{32}} &\leq C, \\[0.25cm]
|| D_{t}(D_{r}( \sigma_{s})) ||_{L^{16}} + || D_{t}(D_{r}( \tilde{
\sigma}_{s})) ||_{L^{16}} + || D_{t}(D_{r}( \tilde{v}_{s})) ||_{L^{16}} &\leq
C.
\end{split}
\end{align}
Moreover, $f \in C^{3}( \mathbb{R})$, while $f$, $f'$, $f''$ and $f'''$
exbibit polynomial growth. \\[0.50cm]
Note that if the involved processes are solutions to stochastic
differential equations, then condition \eqref{h1} is fulfilled given a
sufficient smoothness of the corresponding drift term and volatility function.
We refer to Eq. \eqref{MalSDE} for the computation of the Malliavin derivative
in this case.

The next result then leads to a rate of convergence for $\hat{ \Sigma}_{n}$.

\begin{theorem}
\label{c0}
Assume that $X$ is a continuous It\^{o} semimartingale as in Eq. \eqref{Eqn:X},
where Assumption $(\normalfont \textbf{H})$ and $(\normalfont \textbf{M})$ are
true, as is Assumption $(\normalfont \textbf{K})$ for each component of $f =
(f_{1}, \dots, f_{m})'$. As $n \to \infty$,
$L \to \infty$, and $n/L \to \infty$, it holds that
\begin{equation}
\hat{ \Sigma}_{n} - \Sigma = \underbrace{O_{p} \left( \frac{1}{ \sqrt{L}}
\right)}_{\text{\upshape{CLT}}} + \underbrace{O_{p} \left(
\frac{L}{n} \right)}_{ \text{\upshape{blocking}}}.
\end{equation}
\end{theorem}
\begin{proof}
See Appendix.
\end{proof}

Theorem \ref{c0} presents the leading errors inherent in $\hat{ \Sigma}_{n}$.
The first term, $1/\sqrt{L}$, intuitively follows from a central limit
theorem result, because $\hat{ \Sigma}_{n} $ is an empirical mean of $L$
asymptotically, conditionally independent statistics. However, it is not
easy to apply this relationship for a formal derivation of the error rate.
The second error is more subtle. It comes from freezing the volatility
process at the beginning of a subblock of length $L/n$. If volatility is
assumed to be H\"older continuous of order $\alpha \in (0,1]$, a rough
estimate implies an error rate of $(L/n)^{ \alpha}$. However, due to
the semimartingale structure of $\sigma$, we can improve this to $L/n$
by applying a more refined estimation technique. As such, we should point
out that the proof of Theorem \ref{c0}
is much more complex compared to subsampling of i.i.d. observations.

We find the fastest rate of convergence by balancing both errors.
This requires:
\begin{equation}
L = O \bigl( n^{2/3} \bigr),
\end{equation}
such that
\begin{equation}
\hat{ \Sigma}_{n} - \Sigma = O_{p} \bigl( n^{-1/3} \bigr).
\end{equation}
In a general model (e.g., where $\sigma$ is a diffusion process), we believe
this rate is sharp and cannot be improved within the context
of subsampling high-frequency data, but we shall not attempt to prove it.
Of course, if we impose stricter, parametric assumptions, such as $\sigma_{t}
= \sigma$ is constant and $\mu_{t} = 0$, the rate is faster and gets
arbitrarily close to $n^{-1/2}$, as here the additional blocking error in
Theorem \ref{c0} drops out, and then we can take $L = O(n^{1-\epsilon})$,
with $\epsilon > 0$ arbitrarily small.

We can combine the consistency of $\hat{ \Sigma}_{n}$ from Theorem
\ref{c0} with the convergence in distribution in Eq. \eqref{Eqn:MN}. If we
then appeal to the properties of stable convergence, we get the
feasible result:
\begin{equation}
\label{Eqn:feasible}
\hat{ \Sigma}_{n}^{-1/2} \sqrt{n} \Bigl( V(f)^{n} - V(f) \Bigr) \overset{d}{
\to} \text{N}(0, I_{m}),
\end{equation}
which can be used to construct confidence intervals for $V(f)$ or
do hypothesis testing. If the convergence had not been stable in law, this
result would not follow in general.

Theorem \ref{c0} is proved under the assumption that $\sigma$ is
continuous---as are all coefficients of the model---and $f$ is
differentiable. We note again the stable central limit theorem of Proposition
\ref{Thm:BvCLT} is also valid for a non-differentiable function $f$, and
possibly discontinuous volatility process, given Assumption \textbf{(H')} and
\textbf{(K')} from BGJPS6, but it appears out of reach to derive a convergence
rate for $\hat{ \Sigma}_{n}$ here.
We can nonetheless show that $\hat{ \Sigma}_{n}$ still converges in probability
to $\Sigma$ under these weaker conditions, which is relevant for applied work.
\begin{theorem}
\label{c0prime}
Assume that $X$ is a continuous It\^{o} semimartingale as in Eq. \eqref{Eqn:X},
where Assumption $(\normalfont \textbf{V})$ is true, as are Assumption
$(\normalfont \textbf{H'})$ and $(\normalfont \textbf{K'})$ from
BGJPS6. As $n \to \infty$, $L \to \infty$, and $n/L \to \infty$, it holds that
\begin{equation}
\hat{ \Sigma}_{n} \overset{p}{\to} \Sigma.
\end{equation}
\end{theorem}
\begin{proof}
See Appendix.
\end{proof}

\subsection{Subsampling for bipower variation}

\label{Sec:BipowerVariation}

In the previous section, we presented a subsampling estimator for the
asymptotic conditional covariance matrix of power variation. If we
are interested in bipower (or multipower) variation, the theory derived there does not
readily apply. This is because the summands in Eq. \eqref{Eqn:BVn}
are, asymptotically, 1-dependent, which the subsampling approach shown
in Figure \ref{Fig:subsample} does not adequately capture.

In order to consistently estimate $\Sigma$ in the bipower case, we use an
intuitive blocking approach, which is described next. We define the $i$th
block of high-frequency data by taking:
\begin{equation} \label{Block}
B_{i}(p) = \Bigl\{ j : (i-1)p \leq j \leq i p \Bigr\},
\end{equation}
where $p \geq 2$ is an integer, and $i \geq 1$.

$B_{i}(p)$ is composed of the observation index associated with the sample
of adjacent log-price observations
$X_{(i-1)p/n}, \ldots, X_{i p /n}$. From this, we can compute $p$
consecutive returns $\Delta_{(i-1)p+1}^{n}X, \ldots,
\Delta_{i p}^{n}X$. Therefore, $B_{i}(p)$ plays the role of the interval
$[(i-1)/n,i/n]$ for power variation, which was used to compute a single return
$\Delta_{i}^{n}X$. The only change is that we need to make this interval
longer, such that we can consistently estimate the covariance structure of
$V(f,g)^{n}$. As the $B_{i}(p)$'s are based on non-overlapping
increments, it still holds that bipower variations computed from different
subsamples are, asymptotically, conditionally independent.

We reset $\hat{ \Sigma}_{n}$ as follows:
\begin{equation}
\label{Eqn:SigmaNHatBipower}
\hat{ \Sigma}_{n} = \frac{1}{L} \sum_{l=1}^{L} \biggl( \sqrt{
\frac{n}{L}} \Bigl( V_{l} (f,g)^{n} - V(f,g)^{n} \Bigr)\biggr) \biggl(
\sqrt{ \frac{n}{L}} \Bigl( V_{l} (f,g)^{n} - V(f,g)^{n} \Bigr)\biggr)',
\end{equation}
where, assuming $L p$ divides $n$,
\begin{align}
\begin{split}
\label{Eqn:Hac}
V_{l}(f,g)^{n} &= \frac{L p}{n} \sum_{i = 1}^{n / L p}
v_{(i - 1)L + l}(f,g)^{n}, \\[0.25cm]
v_{i}(f,g)^{n} &= \frac{1}{p - 1} \sum_{j,j+1 \in B_{i}(p)} f \bigl(
\sqrt{n} \Delta_{j}^{n} X \bigr) g \bigl( \sqrt{n} \Delta_{j+1}^{n}
X \bigr).
\end{split}
\end{align}
Note that $n/L p$ is the number of blocks assigned to each subsample, and that
the subsample statistic $v_{i}(f,g)^{n}$ is computed only from data within the
$i$th block $B_{i}(p)$.
As in the above, we definitely require $n \to \infty$, $p \to \infty$,
$L \to \infty$, and $n/pL \to \infty$ to prove the asymptotic theory for
$\hat{ \Sigma}_{n}$. It turns out, however, we need a slightly
stronger condition for the last part to ensure consistency. This is because
the rate $\displaystyle
\sqrt{ \frac{n}{L}}$ in the definition of Eq. \eqref{Eqn:SigmaNHatBipower}
corresponds to the martingale part of $V_{l} (f,g)^{n} - V(f,g)^{n}$, while
the statistic $V_{l} (f,g)^{n} - V(f,g)^{n}$ also has a bias term, which is
of order $Lp/n$. Thus, to make the bias negligible with respect to the
martingale part, we need $n/Lp^2 \to \infty$. Hence,
our ``minimal'' assumptions are based on this condition.

\begin{theorem}
\label{c1}
Assume that $X$ is a continuous It\^{o} semimartingale as in Eq. \eqref{Eqn:X},
where Assumption $(\normalfont \textbf{H})$ and $(\normalfont \textbf{M})$ are
true, as is Assumption $(\normalfont \textbf{K})$ for each component of $f =
(f_{1}, \dots, f_{m})'$ and $g = (g_{1}, \dots, g_{m})'$. As $n \to \infty$,
$p \to \infty$, $L \to \infty$, and $n/Lp^2 \to \infty$, it holds that
\begin{equation}
\label{3.13}
\hat{ \Sigma}_{n} - \Sigma = \underbrace{O_{p} \left( \frac{1}{
\sqrt{L}} \right)}_{ \text{\upshape{CLT}}} +
\underbrace{O_{p} \left( \frac{Lp^2}{n} \right)}_{ \text{
\upshape{blocking}}} + \underbrace{O_{p} \left( \frac{1}{p}
\right)}_{ \text{ \upshape{HAC}}}.
\end{equation}
\end{theorem}
\begin{proof}
See Appendix.
\end{proof}
The first two errors in Eq. \eqref{3.13} can be interpreted as
to those in Theorem \ref{c0}, except the second is also affected by the
block size $p$. Meanwhile, the decomposition of $\hat{ \Sigma}_{n} - \Sigma$ in
Theorem \ref{c1} has an extra error of order $O_{p}(1/p)$.
The additional term, which emerges from the computation of the conditional
variance of
$v_{i}(f,g)^{n}$, has an intuitive interpretation, if we recall that in the
current setting of bipower variation, the summands in Eq. \eqref{Eqn:Hac}
(or Eq. \eqref{Eqn:BVn}) are asymptotically 1-dependent.

Consider the following stylized example. Assume that $(Z_{i})_{i \geq 1}$ is
a sequence of stationary 1-dependent random
variates. Then,
\begin{equation}
\label{Eqn:Example}
\text{var} \left( \frac{1}{\sqrt{p}} \sum_{i = 1}^{p} Z_{i} \right) =
\text{var}(Z_{1}) + 2\frac{(p - 1)}{p} \text{cov} (Z_{1},Z_{2})
\to \text{var}(Z_{1}) + 2 \text{cov} (Z_{1},Z_{2}),
\end{equation}
as $p \to \infty$.

This calculation shows that the finite sample variance on the
left-hand side
is not equal to, but converges towards, the asymptotic variance.
The difference, i.e. the bias, is the
term $-2 \text{cov} (Z_{1},Z_{2})/p$, which has order $O(1/p)$.
This example also helps to illustrate that Theorem \ref{c1} does not change,
and in particular the convergence rate of $\hat{ \Sigma}_{n}$ is unaffected,
if we were to compute a higher order
multipower variation statistic. Then there would be more covariance
terms in Eq. \eqref{Eqn:Example}, but the bias in each of
them would still be $O(1/p)$.

The fastest rate is again found by balancing the errors, which
means taking:
\begin{equation} \label{opti}
L = O(n^{2/5}), \qquad p = O(n^{1/5}),
\end{equation}
for which
\begin{equation}
\hat{ \Sigma}_{n} - \Sigma = O_{p} (n^{-1/5}).
\end{equation}
While we do not offer a formal proof, we again believe that in a general
diffusion model this rate is optimal within the framework of subsampling
high-frequency data, as elaborated above.
Moreover, the consistency only result of $\hat{ \Sigma}_{n}$ holds under
weaker assumptions that do not require Assumption (\textbf{K}),
(\textbf{H}) and (\textbf{M}), while an extra condition $L/p \to \infty$
is necessary to deal with an additional bias term.

\begin{theorem}
\label{bipowerconsistency}
Assume that $X$ is a continuous It\^{o} semimartingale as in Eq. \eqref{Eqn:X},
where Assumption $(\normalfont \textbf{V})$ is true, as are Assumption
$(\normalfont \textbf{H'})$ and $(\normalfont \textbf{K'})$ from BGJPS6.
As $n \to \infty$, $p \to \infty,$ $L/p \to \infty$, and $n/L p^2 \to \infty$,
it holds that
\begin{equation}
\hat{ \Sigma}_{n} \overset{p}{\to} \Sigma.
\end{equation}
\end{theorem}
\begin{proof}
See Appendix.
\end{proof}

To end this section, we should point out that for the power variation estimator
covered by Theorem \ref{c0} in the previous subsection, it follows from the work
of BGJPS6
that there exits another consistent, positive semi-definite estimator of
$\Sigma$:
\begin{equation}
\hat{S}_{n} = \frac{1}{2n} \sum_{i = 1}^{n - 1} \Bigl( f \bigl( \sqrt{n}
\Delta_{i}^{n} X \bigr) - f \bigl( \sqrt{n} \Delta_{i+1}^{n} X) \Bigr)
\Bigl( f \bigl( \sqrt{n} \Delta_{i}^{n} X \bigr) - f \bigl( \sqrt{n}
\Delta_{i+1}^{n} X) \Bigr)'.
\end{equation}
$\hat{S}_{n}$ has a better rate of convergence $n^{-1/2}$ compared to
$n^{-1/3}$ derived in the previous section for $\hat{\Sigma}_{n}$.
$\hat{S}_{n}$ is therefore more efficient for power variation, but it does
not work for bi- or multipower variation.

\subsubsection{Subsampling for truncated bipower variation}

\label{section:truncation}

In an efficient market, equilibrium prices should adjust instantly to new
information about fundamentals. If this leads to a significant revision of the
fair value of the asset, the price has to move sharply and, potentially,
discretely. This feature of price formation is not captured by the previous
setup, where $X$ has continuous sample paths. In this section, we therefore add
a jump term to $X$ and develop a framework for jump-robust inference about
volatility based on subsampling truncated bipower
variation \citep*[e.g.,][]{jacod-protter:12a,mancini:09a}.
Accordingly, we assume that:\\[-0.50cm]

\noindent \textbf{Assumption (J)}: $X$ is of the form:
\begin{align}
\begin{split}
\label{Eqn:XJumps}
X_{t} &= X_{0} + \int_{0}^{t} a_{s} \text{d}s + \int_{0}^{t} \sigma_{s}
\text{d}W_{s} \\[0.25cm]
&+ \int_{0}^{t} \int_{E} \delta(s,x) 1_{ \{| \delta(s,x)| \leq 1 \}} (\mu -
\nu) ( \text{d}s, \text{d}x) + \int_{0}^{t} \int_{E} \delta(s,x) 1_{ \{|
\delta(s,x)| > 1 \}} \mu( \text{d}s, \text{d}x),
\end{split}
\end{align}
where $X_{0}$, $a = (a_{t})_{t \geq 0}$, $\sigma = ( \sigma_{t})_{t \geq 0}$
and $W = (W_{t})_{t \geq 0}$ are defined as in Eq. \eqref{Eqn:X}, while $(E, \mathcal{E})$ is
a Polish space, $\mu$ is a random measure on $\mathbb{R}_{+} \times E$ with
compensator $\nu( \text{d}s, \text{d}x) = \text{d}s F( \text{d}x)$, where
$F$ is a $\sigma$-finite measure on $(E, \mathcal{E})$.
Also, $\delta: \Omega \times \mathbb{R}_{+} \times E \to \mathbb{R}$ is a
predictable function and $(S_{k})_{k \geq 1}$ is a sequence of stopping times
increasing to $\infty$ such that $| \delta( \omega,s,x)| \wedge 1 \leq
\psi_{k}(x)$ for all $( \omega, s, z)$ with $s \leq S_{k}( \omega)$
and  $\int_{E} \psi^{ \beta}_{k}(x) F( \text{d}x) < \infty$ for all $k \geq
1$ and $\beta \in [0,1)$. \\[-0.50cm]

$\beta$ relates to the activity index of the price jump process. The condition imposed on $\beta$ implies that the jumps in $X$ are (absolutely) summable, i.e. we restrict attention to jump processes with paths of finite variation, but, possibly, infinite activity.\footnote{As explained earlier, for the subsampling estimator to be consistent for the asymptotic conditional variance-covariance matrix, we typically require a central limit theorem to hold for the underlying statistic of interest. In this respect, the restriction on $\beta$ is a standard condition in the high-frequency volatility literature.}

Although the theory derived here should work with a general $f$ and $g$, it requires a lot of notation. To develop ideas and maintain a streamlined exposition, we focus on the class of pure
bipower variations in this section. The $k$th coordinate of the truncated bipower variation $\check{V}(q,r)^{n}$ is therefore:
\begin{equation}
\label{V(q_k,r_k)}
\displaystyle
\check{V}(q_{k}, r_{k})^{n} = \frac{1}{n} \sum_{i=1}^{n-1} | \sqrt{n} \Delta_{i}^{n} \check{X}|^{q_{k}} | \sqrt{n} \Delta_{i+1}^{n} \check{X}|^{r_{k}},
\end{equation}
where $\Delta_{i}^{n} \check{X} = \Delta_{i}^{n} X \cdot \mathbbm{1}_{
\{| \Delta_{i}^{n} X | \leq u_{n} \}}$ is the increment after jump-trucation
and the threshold level $u_{n} = \alpha n^{- \check{ \omega}}$ with
$\alpha >0$ and $\check{ \omega}\in (0,1/2)$. By excluding the largest
increments of $X$, the bipower variation statistic is, asymptotically,
merely based on those high-frequency returns that are compatible with a
continuous sample path model.

First, we recall the central limit theorem for $\check{V}(q,r)^{n}$.
\begin{proposition}
\label{jacodthmj}
Assume that $X$ is a jump-diffusion process as in Assumption
($\normalfont{ \textbf{J}}$) and $\sigma$
follows Assumption ($\normalfont{ \textbf{V}}$) with $\sigma > 0$. We denote by
$s = 1 \wedge \min \{ q_{k}, r_{k} : q_{k} > 0, r_{k} > 0, 1 \leq k \leq m \}$
and $s' = 1 \vee \max \{ q_{k}, r_{k}: 1 \leq k \leq m \}$. Then, if $\beta
\leq s$, $\displaystyle \check{ \omega} > \frac{s' - 1}{2(s' - \beta)}$, and as
$n \to \infty$, it holds that
\begin{equation}
\label{tg4}
\sqrt{n} \Big( \check{V}(q,r)^{n} - V(q,r) \Big) \overset{ d_{s}}{ \to}
\text{\upshape{MN}}(0, \Sigma),
\end{equation}
where the elements of $V(q,r)$ and $\Sigma$ are given as in Eq. \eqref{Eqn:OBV}
and \eqref{SigmaPowers}.
\end{proposition}
\begin{proof}
See Theorem 13.2.1 and Example 13.2.2 in \citet*{jacod-protter:12a}.
\end{proof}
In the jump-diffusion setting, we define the subsample estimator of $\Sigma$ as:
\begin{equation}
\hat{ \Sigma}_{n} = \frac{1}{L} \sum_{l = 1}^{L} \biggl(
\sqrt{ \frac{n}{L}} \Bigl( \check{V}_{l}(q,r)^{n} - \check{V}(q,r)^{n} \Bigr)
\biggr) \biggl( \sqrt{ \frac{n}{L}} \Bigl( \check{V}_{l}(q,r)^{n} -
\check{V}(q,r)^{n} \Bigr) \biggr)',
\end{equation}
where, assuming $L p $ divides $n$,
\begin{align}
\begin{split}
\check{V}_{l}(q_{k},r_{k})^{n} &= \frac{L p}{n}\sum_{i = 1}^{n/L p
} v_{(i-1)L+l}(q_{k},r_{k})^{n}, \\[0.25cm]
v_{i}(q_{k},r_{k})^{n} &= \frac{1}{p-1}
\sum_{j,j+1 \in B_{i}(p)}|\sqrt{n} \Delta_{j}^{n} \check{X}|^{q_{k}}
|\sqrt{n} \Delta_{j+1}^{n} \check{X}|^{r_{k}},
\end{split}
\end{align}
and $B_{i}(p)$ is given as in Eq. \eqref{Block}.

Finally, we are ready to state a consistency result.
\begin{theorem}
\label{theorem:truncation}
Assume that $X$ is a jump-diffusion process as in Assumption
($\normalfont{ \textbf{J}}$) and $\sigma$ follows Assumption
($\normalfont{ \textbf{V}}$) with $\sigma > 0$. Moreover, we require that
$\beta \leq s$ and $\displaystyle \check{ \omega} > \frac{s' - 1}{2(s'
- \beta)}$. Then, as $n \to \infty$, $p \to \infty$, $L/p \to
\infty$ and $n/L p^{2} \to \infty$, it holds that
\begin{equation}\label{mainthm}
\hat{ \Sigma}_{n} \overset{p}{\to} \Sigma.
\end{equation}
\end{theorem}
\begin{proof}
See Appendix.
\end{proof}

\subsection{Microstructure noise}
\label{Sec:Noise}
In practice, assets are not traded within a frictionless market. The
recorded data constitute a discrete sample of transactions or bid-ask
quotes, whose prices are affected by common sources of market
imperfections, such as bid-ask spreads, price discreteness, and so
forth \citep*[e.g.,][]{black:86a,niederhoffer-osborne:66a,roll:84a}.
Even if these were small enough to be ignored, high-frequency data
are also corrupted by outliers (due to bugs in the data transmission,
fat-finger errors, etc.) and subject to other irregularities (e.g., quote
stuffing, screen fighting, etc.). The combination of these effects leads
to marked differences between real data and those generated by a diffusion
model.

To accommodate this, we need to modify the setup. We are going to
take the observed price as the true, underlying price perturbed by an
additive noise term, i.e.
\begin{equation}
\label{Eqn:Y}
Y_{i/n} = X_{i/n} + \epsilon_{i/n},
\end{equation}
where $X$ is defined as in Eq. \eqref{Eqn:X}, while $\epsilon = (
\epsilon_{t})_{t \geq 0}$ is a noise process. We impose the
following: \\[-0.50cm]

\noindent \textbf{Assumption (N)}: (i) $\epsilon$ is i.i.d. with $\mathbb{E}[
\epsilon_{t}] = 0$ and $\text{var}( \epsilon_{t}) = \omega^{2}$ for all $t \geq
0$, (ii) $\epsilon$ is independent of $X$, (iii) the distribution of $\epsilon$
is symmetric around 0, and (iv) $\mathbb{E} \big[ | \epsilon_{t}|^{s} \big] <
\infty$ for some $s > 0$.

\subsubsection{Pre-averaging}

To alleviate the impact of noise, we make use of the notion that as $X$
is
continuous and $\epsilon$ is i.i.d., we can locally smooth $Y_{i/n}$
in the vicinity of $i/n$ to retrieve an estimate, say $\bar{Y}_{i/n}$,
which tends to be close to $X_{i/n}$, because the
noise is largely averaged away
\citep*[e.g.,][]{jacod-li-mykland-podolskij-vetter:09a,
podolskij-vetter:09a,podolskij-vetter:09b}. Averaging our discrete sample of
noisy high-frequency data this way leads to a new set of increments,
$\Delta \bar{Y}_{i}^{n}$, based on pre-averaged prices.\footnote{In the
context of volatility estimation, there are several tools at our disposal
to handle microstructure noise, including the realized kernel; based on
auto-covariance corrections
\citep*[see, e.g.,][]{barndorff-nielsen-hansen-lunde-shephard:08a,
barndorff-nielsen-hansen-lunde-shephard:11a}, and the two- or multi-scale
realized variance; based on price subsampling
\citep*[see, e.g.,][]{zhang:06a,zhang-mykland-ait-sahalia:05a}. Out of these,
pre-averaging is the most general approach, as it is applicable to a large
number of estimation problems.}

To implement pre-averaging, we need some extra notation. We choose
a sequence $k_{n}$ of integers (the pre-averaging window) and a scalar
$\theta > 0$, such that
\begin{equation}
k_{n} = \theta \sqrt{n} + o \bigl( n^{-1/4} \bigr).
\end{equation}

We also need a weight function $w: \mathbb{R} \mapsto
\mathbb{R}$ to do averaging. We assume $w$ is continuous on $[0,1]$
and piecewise continuously differentiable with a
piecewise Lipschitz derivative $w'$. Moreover, we assume
that $w(0) = w(1) = 0$ and $\int_{0}^{1} (w(t))^{2}
\text{d}t > 0$. The following numbers and functions are associated with $w$:
\begin{align}
\label{m2}
\begin{split}
\phi_{1}(s) &= \int_{s}^{1}w'(u)w'(u-s)\text{d}u, \quad \phi_{2}(s) =
\int_{s}^{1}w(u) w(u-s) \text{d}u, \quad \text{for } s \in[0,1], \\[0.25cm]
\psi_{1} &= \phi_{1}(0), \quad \psi_{2} = \phi_{2} (0), \quad
\Phi_{ij} = \int_{0}^{1} \phi_{i}(s) \phi_{j}(s) \text{d}s, \quad \text{for }
i,j = 1,2, \\[0.25cm]
\psi_{1}^{n} &= k_{n} \sum_{j = 0}^{k_{n}}
(w_{j+1}^{n} - w_{j}^{n})^{2}, \quad
\psi_{2}^{n} = \frac{1}{k_{n}} \sum_{j = 1}^{k_{n}} (w_{j}^{n})^{2},
\end{split}
\end{align}
where $w_{j}^{n} = w(j/k_{n})$. In passing, we note that
\begin{equation}
\label{psierror}
\psi_{1}^{n} = \psi_{1} + O \bigl(n^{-1/2} \bigr) \quad \text{ and }
\quad \psi_{2}^{n} = \psi_{2} + O \bigl(n^{-1/2} \bigr),
\end{equation}
which means that in the asymptotic theory only $\psi_{1}$ and $\psi_{2}$
appear. Still, it is recommendable to use $\psi_{1}^{n}$ and $\psi_{2}^{n}$
for simulations and empirical work, as it entails better finite sample
properties.

The return series, following pre-averaging, is:
\begin{equation}
\label{Eqn:PreavgY}
\Delta \bar{Y}_{i}^{n} = \sum_{j = 1}^{k_{n}} w_{j}^{n} \Delta_{i+j}^{n}
Y = - \sum_{j = 0}^{k_{n}} (w_{j+1}^{n} - w_{j}^{n}) Y_{\frac{i+j}{n}}, \quad \text{for }
i = 1, \ldots, n - k_{n} + 2.
\end{equation}

\subsubsection{Pre-averaged bipower variation}

The addition of microstructure noise creates further complications for
inference procedures from high-frequency data. To make our framework
analytically tractable, yet practically relevant, we therefore again
restrict attention to the class of pure bipower variations.

The $k$th coordinate of $V^{*}(q,r)^{n}$ is defined as:
\begin{equation}
\label{Eqn:PreavgBV}
\displaystyle V^{*}(q_{k},r_{k})^{n} = \frac{1}{n - 2k_{n} + 2} \sum_{i =
1}^{n - 2k_{n} + 2} |n^{1/4}\Delta \bar{Y}_{i}^{n}|^{q_{k}} | n^{1/4}
\Delta \bar{Y}_{i+k_{n}}^{n}|^{r_{k}}.
\end{equation}
The intuition behind this construction is that pre-averaging induces some
autocorrelation (of order $k_{n}$) in the pre-averaged price series, which
is broken by multiplying pre-averaged returns that are $k_{n}$ terms apart.
In essence, this leads to a lower, effective sample of size $n - 2k_{n} + 2$.

\citet*{podolskij-vetter:09a} show that
\begin{equation}
\label{Eqn:CLTnoise}
n^{1/4} \Bigl( V^{*}(q,r)^{n} - V^{*}(q,r) \Bigr) \overset{d_{s}}{
\to} \text{MN}(0, \Sigma^{*}),
\end{equation}
where
\begin{equation}
\label{Eqn:plimPBV}
V^{*}(q_{k},r_{k}) = \mu_{q_{k}}
\mu_{r_{k}} \int_{0}^{1} \biggl( \theta \psi_{2} \sigma_{s}^{2} +
\frac{1}{\theta} \psi_{1} \omega^{2} \biggr)^{\frac{q_{k} + r_{k}}{2}}
\text{d}s.
\end{equation}
Thus, pre-averaging slows down the rate of convergence, but $n^{-1/4}$
is nonetheless the fastest rate in noisy
diffusion models \citep*{gloter-jacod:01a,gloter-jacod:01b}.

In the above, $\Sigma^{*}$ is the $m \times m$ conditional covariance
matrix of $V^{*}(q,r)^{n}$.\footnote{\label{Foot:avar_rv}
We do not state the expression of $\Sigma^{*}$ here, but it can be
found in \citet*{podolskij-vetter:09a}. In general, $\Sigma^{*}$ has a
complicated structure (even with i.i.d., independent noise), and it
is typically not known in closed-form. An exception, where $\Sigma^{*}$
can be computed analytically, is if $(q,r)$ consists of even
non-negative integers, as in Theorem \ref{m11}. As an example,
which is used in the simulations, take the pre-averaged realized variance.
It sets $(q,r) = (2,0)$ and has
$\Sigma^{*}(2,0) = 4 \int_{0}^{1} \Bigl( \theta^{3} \Phi_{22} \sigma_{s}^{4}
+ 2 \theta \Phi_{12} \sigma_{s}^{2} \omega^{2} + \frac{1}{\theta} \Phi_{11}
\omega^{4} \Bigr)\text{d}s$.}  \citet*{podolskij-vetter:09a} go on to
develop a consistent estimator of $\Sigma^{*}$ (defined in Eq.
\eqref{Eqn:SigmaPV} in the simulation section), but it is based on
element-by-element estimation. The disadvantage of this approach is that
it does not ensure that the whole covariance matrix estimate is positive
definite in finite samples and, as
our simulations and empirical analysis show, a large fraction of such
estimates in fact fail to be positive definite. Moreover, even if this property does
hold, the estimate is often near-singular, resulting in an ill-conditioned
and highly unstable covariance matrix.

\subsubsection{Subsampling noisy high-frequency data}

To construct our estimator in the noisy setting, we follow the procedure
from before by splitting the full sample of noisy high-frequency data into
subsamples using a blocking approach.

We redefine:
\begin{equation}
B_{i}(p) = \Bigl\{ j : (i-1)pk_{n} \leq j \leq
i p k_{n} \Bigr\},
\end{equation}
where $p \geq 3$ is an integer and $i \geq 1$.

$B_{i}(p)$ is now the $i$th block of noisy high-frequency data. As readily
seen, the only change compared to the noiseless setting is that $B_{i}(p)$
uses a larger block size. This implies we can do a sufficient amount of
averaging within each block in order to diminish the noise, while still
preserving enough of an effective sample size to estimate the correlation
structure of $V^{*}(q,r)^{n}$.

Then, we set
\begin{equation}
\label{Eqn:SigmaNStar}
\hat{ \Sigma}_{n}^{*} = \frac{1}{L} \sum_{l = 1}^{L} \biggl(
\frac{n^{1/4}}{ \sqrt{L}} \Bigl( V^{*}_{l}(q,r)^{n} - V^{*}(q,r)^{n} \Bigr)
\biggr) \biggl( \frac{n^{1/4}}{ \sqrt{L}} \Bigl( V^{*}_{l}(q,r)^{n} -
V^{*}(q,r)^{n} \Bigr) \biggr)',
\end{equation}
where, assuming $L p k_n$ divides $n$,
\begin{align}
\begin{split}
V^{*}_{l}(q_{k},r_{k})^{n} &= \frac{L p k_{n}}{n}\sum_{i = 1}^{n/L p
k_{n}} v_{(i-1)L+l}(q_{k},r_{k})^{n}, \\[0.25cm]
v_{i}(q_{k},r_{k})^{n} &= \frac{1}{pk_{n}-2 k_{n}+2}
\sum_{j,j+k_n-1 \in B_{i}(p)}|n^{1/4} \Delta \bar{Y}_{j}^{n}|^{q_{k}}
|n^{1/4} \Delta \bar{Y}_{j+k_{n}}^{n}|^{r_{k}}.
\end{split}
\end{align}
As in the above, we should note that the summands $v_{i}(q_{k},r_{k})^{n}$
in the subsample estimates $V^{*}_{l}(q_{k},r_{k})^{n}$ exploit data
solely from $B_{i}(p)$. Therefore, pre-averaging has to be done locally
within the block, so that there is no overlap in the pre-averaged returns
across the various blocks.
\begin{theorem}
\label{m11}
Assume that $Y_{t} = X_{t} + \epsilon_{t}$ is a noisy diffusion model, where
$X_{t}$ is given by Eq. \eqref{Eqn:X}, that fulfills Assumption $(\normalfont
\textbf{H})$ and $(\normalfont{\textbf{M}})$. Also, Assumption
$(\normalfont{\textbf{N}})$ with $s > 3 \vee \max \{ 2(q_{k}+r_{k}): 1 \leq k
\leq m \}$ is true. Let $q = (q_{1}, \ldots, q_{m})'$ and
$r = (r_{1}, \ldots, r_{m})'$ be vectors of even non-negative integers. Then,
as $n \to \infty$, $p \to \infty$, $L \to \infty$ and $\sqrt{n}/Lp^2 \to
\infty$, it holds that
\begin{equation}
\label{m12}
\hat{ \Sigma}_{n}^{*} - \Sigma^{*} = \underbrace{O_{p} \left( \frac{1}{
\sqrt{L}} \right)}_{ \text{\upshape{CLT}}} + \underbrace{O_{p} \left(
\frac{Lp^2}{ \sqrt{n}} \right)}_{ \text{ \upshape{blocking}}} +
\underbrace{O_{p} \left( \frac{1}{p} \right)}_{ \text{ \upshape{HAC}}}
\end{equation}
\end{theorem}
\begin{proof}
See Appendix.
\end{proof}
As in the previous subsection, the minimal assumptions we need to prove
consistency are $n \to \infty$, $p \to \infty$, $L \to \infty$ and
$\sqrt{n}/Lp^2 \to \infty$. The last condition again
ensures that a bias term of the statistic $V^{*}_{l}(q,r)^{n}
- V^{*}(q,r)^{n}$ is negligible with respect to its martingale part.

Now, we achieve the best rate
\begin{equation}
\hat{ \Sigma}_{n}^{*} - \Sigma^{*} = O_{p} \bigl( n^{-1/10} \bigr),
\end{equation}
by choosing
\begin{equation}
L = O(n^{1/5}) \quad \text{and} \quad  p = O(n^{1/10}).
\end{equation}
Thus, the existence of microstructure frictions also adversely affects the
speed of convergence of $\hat{ \Sigma}_{n}^{*}$.

Theorem \ref{m11} is binding, because it restricts the choice of the powers
$q_{k}$ and $r_{k}$ to even non-negative integers. But we can actually prove
a weaker consistency result for any pre-averaged bipower variation, which is
useful for practical work. First, we recall that
allowing for general powers $q_{k}, r_{k} \geq 0$ by itself requires some
stronger assumptions to prove the underlying central limit theorem in Eq.
\eqref{Eqn:CLTnoise}. In particular,
\citet*{podolskij-vetter:09a} impose that the noise distribution
is (i) symmetric with (ii) $\mathbb{E}[|\epsilon|^{a}] < \infty$ for $a
\in (-1,0)$ (their Assumption (\textbf{A})) and that the noise distribution
fulfills Cramer's condition, i.e. $\lim \sup_{|t| \to \infty} \chi(t) < 1$,
where $\chi$ is the characteristic function of $\epsilon$ (their Assumption
(\textbf{A'})). Of course, we also need this. Note, however, that we can
again dispense with Assumption (\textbf{M}) for consistency.

\begin{theorem}
\label{NoiseAnyPowerThm}
Assume that $Y_{t} = X_{t} + \epsilon_{t}$ is a noisy diffusion model, where
$X_{t}$ is given by Eq. \eqref{Eqn:X}, $\sigma$ is continuous and fulfills
Assumption $(\normalfont\textbf{V})$ with $\sigma>0$, while the noise fulfills
Assumption $(\normalfont{ \textbf{N}})$ with $s > 3 \vee \max \{ 2(q_{k}+r_{k}):
1 \leq k \leq m \}$ and also Assumption $(\normalfont{\textbf{A}})$ and
$(\normalfont{ \textbf{A'}})$ from
\citet*{podolskij-vetter:09a}. Then, as $n \to \infty$, $p \to \infty$,
$L/p \to \infty$ and $\sqrt{n}/L p^2 \to \infty$, it holds for any $q,r \geq 0$
that
\begin{equation}
\hat{ \Sigma}_{n}^{*} \overset{p}{\to} \Sigma^{*}.
\end{equation}

\end{theorem}
\begin{proof}
See Appendix.
\end{proof}

\begin{remark}
\rm The subsampling idea can be applied to other estimators, which admit
asymptotic mixed normality, such as the two-scale and kernel-based
realized variance of \citet*{zhang-mykland-ait-sahalia:05a} and
\citet*{barndorff-nielsen-hansen-lunde-shephard:08a}. The consistency
of the subsample estimator of the asymptotic conditional covariance matrix
of these estimators has already been proved in prior work
\citep*[i.e.,][]{ikeda:16a,kalnina:11a,kalnina:15a,kalnina-linton:07a,
varneskov:16a}. The optimal choice of the tuning
parameters for pre-averaged bipower variation---derived in this paper---should
continue to hold in the context
of these other estimators, but it needs to be verified formally. We leave
this for future work.
\end{remark}

\subsubsection{Extension to dependent and heteroscedastic noise}

\label{section:general_noise}

The i.i.d. framework on the microstructure noise $\epsilon$ is a
convenient outset, but it is hard to defend at the tick frequency, both
in theory and practice \citep*[e.g.,][]{diebold-strasser:13a}. An
intriguing
ability of the subsampling estimator $\hat{ \Sigma}_{n}^{*}$ is that it
tends to be robust against the intricate features of the noise process,
as long as an associated central limit theorem holds. \cite{kalnina:11a}
studied subsampling in the presence of both autocorrelated and
heteroscedastic noise for the two-scale realized variance. In this
subsection, we show how our theoretical results adapt to such models,
allowing for more general structure in the noise process.

\medskip

\noindent \textit{1. Dependent noise}

\medskip

\noindent Autocorrelation in tick-by-tick returns can extend
beyond the first lag, depending a bit on how you gather the data
\citep*[e.g.,][]{hansen-lunde:06b,ait-sahalia-mykland-zhang:11b}.
This cannot be captured by independent noise, so we start
by weakening this assumption to so-called
$m$-dependent noise. Thus, we now assume that the noise process $(
\epsilon_t)_{t \geq 0}$ is stationary and
that the random variables $\epsilon_{i/n}$ and $\epsilon_{j/n}$ are
independent, only if $|i - j| > m$. \citet*{hautsch-podolskij:13a} prove
a
central limit theorem for pre-averaging in this setup (based on the
estimator $V^{*}(2,0)^{n}$). As indicated by
their results, the law of large numbers and the central limit theorem
for the pre-averaged bipower variation estimator $V^{*}(q,r)^n$ do not
change, except that the noise variance $\omega^{2}$ has to be replaced
by the expression $\rho^{2} = \omega^{2} + 2 \sum_{j = 1}^{m}
\text{cov}( \epsilon_{i/n}, \epsilon_{(i+j)/n})$, i.e.
\begin{equation}
V^{*}(q_{k}, r_{k})^{n} \overset{p}{ \to} V^{*}(q_{k}, r_{k}) =
\mu_{q_{k}}
\mu_{r_{k}} \int_{0}^{1} \biggl( \theta \psi_{2} \sigma_{s}^{2} +
\frac{1}{\theta} \psi_{1} \rho^{2} \biggr)^{\frac{q_{k} + r_{k}}{2}}
\text{d}s.
\end{equation}
Here, the form of $\Sigma^{*}$ changes.\footnote{In particular,
\citet*{hautsch-podolskij:13a} show that the asymptotic variance of the
pre-averaged realized variance $V^{*}(2,0)^{n}$ presented in Footnote
\ref{Foot:avar_rv} is unchanged, apart from replacing $\omega$ with
$\rho$ everywhere.} Our proposed estimator $\hat{ \Sigma}_{n}^{*}$ is still
consistent though, as the change in the bias caused through replacing
$\omega^2$ by $\rho^2$ is corrected by construction
in Eq. \eqref{Eqn:SigmaNStar}.
This is because $\hat{ \Sigma}_{n}^{*}$ imitates
the underlying covariation, irrespective of the true microstructure model.
This implies that the asymptotic results of Theorem \ref{m11} and
\ref{NoiseAnyPowerThm} are also true for the $m$-dependent noise model.

\medskip

\noindent \textit{2. Heteroscedastic noise}

\medskip

\noindent The market microstructure reveals itself, for example, via the
bid-ask spread. It has been noticed in many empirical asset price
series that such margins are not constant through time, but tend to vary
systematically within the day in the form of a U-shape. Thus,
in the equity market spreads are typically larger in the morning and
afternoon than during the middle of the day.
To accommodate this, another setup that has been studied
is heteroscedastic noise
\citep*[e.g.,][]{bandi-russell:06a,kalnina-linton:08a}.
Here, we follow the exposition in
\citet*{jacod-li-mykland-podolskij-vetter:09a} by assuming that
\begin{equation}
\mathbb{E} \big[ \epsilon_{t} \mid X \big] = 0, \qquad \mathbb{E} \big[
\epsilon_{t}^{2} \mid X \big] = \omega_{t}^{2} \quad
\text{is c\`{a}dl\`{a}g and $(\mathcal{F}_{t})$-adapted},
\end{equation}
while, conditional on $X$, $\epsilon_{t}$ and $\epsilon_{s}$ are
independent for any $t \neq s$. By
construction, this model exhibits a time-varying variance structure of the
noise, which depends on the efficient price $X$.\footnote{An
example of this model is additive, uniform
noise plus rounding: $Y_{t} = \gamma \big[ (X_{t} + u_{t}) / \gamma \big]$,
where $u = (u_{t})_{t \geq 0}$ is an i.i.d. $\mathcal{U}([0,
\gamma])$-distributed
process that is independent of $X$, and $\gamma > 0$ is a fixed rounding
level. In this setting, the conditional variance of the noise process is given
by: $\omega_{t}^{2} = \gamma^{2} \bigg( \Big\{ \frac{X_{t}}{ \gamma} \Big\} -
\Big\{ \frac{X_{t}}{ \gamma} \Big\}^{2} \bigg)$, with $\{x\} = x - [x]$
denoting the fractional part of $x$.} Note that these conditions do not
contradict unconditional dependence in $\epsilon$. In this case, the
consistency result translates to:
\begin{equation}
V^{*}(q_{k}, r_{k})^n \overset{p}{\to} V^{*}(q_{k},r_{k})=
\mu_{q_{k}}
\mu_{r_{k}} \int_{0}^{1} \biggl( \theta \psi_{2} \sigma_{s}^{2} +
\frac{1}{\theta} \psi_{1} \omega_{s}^{2} \biggr)^{\frac{q_{k} + r_{k}}{2}}
\text{d}s.
\end{equation}
Again, the estimator $\hat{ \Sigma}_{n}^{*}$ automatically adapts to the new
environment, which, in particular, implies that the consistency result of
Theorem \ref{NoiseAnyPowerThm} is still true.
In order to maintain unchanged error rates as in Theorem \ref{m11}, however,
we need to also impose identical assumptions on the process $( \omega_{t})_{t
\geq 0}$ as for the volatility $(\sigma_{t})_{t \geq 0}$. This is
because the role of both processes are identical in all asymptotic expansions.

\section{Simulations}

\label{Sec:Simulation}

In this section, we conduct a small Monte Carlo study. It takes a closer
look at the finite sample properties of covariance matrix estimation by
subsampling. Throughout, we restrict attention to the noisy setting and
estimation of $\Sigma^{*}$. We examine the ability of our proposed
estimator $\hat{ \Sigma}_{n}^{*}$ to assist in drawing feasible inference
about the pre-averaged bipower variation.

The efficient log-price $X$ is simulated as:
\begin{align}
\label{Eqn:Heston}
\begin{split}
\text{d}X_{t} &= \sigma_{t}\text{d}W_{t}, \\[0.10cm]
\text{d}\sigma_{t}^{2} &= \kappa(\sigma^{2} - \sigma_{t}^{2}) \text{d}t
+ \xi \sigma_{t} (\rho \text{d}W_{t} + \sqrt{1 - \rho^{2}}\text{d}B_{t}),
\end{split}
\end{align}
where $W_{t}$ and $B_{t}$ are independent standard Brownian motions, while
$\kappa, \sigma^{2}, \xi$ and $\rho$ are parameters. The process adopted for
$\sigma_{t}^{2}$ is a \citet*[][]{heston:93a} model, which is mean-reverting;
features square-root volatility; and accommodates a leverage
effect.\footnote{The leverage effect describes a negative correlation between
an asset's return and volatility \citep*[e.g.,][]{black:76a,christie:82a}.
Thus, if a leverage effect is present, one would expect $\rho$ to be
negative.}

To get a version of the model from which we can actually simulate data,
we apply a standard Euler approximation to the continuous time formulation
in Eq. \eqref{Eqn:Heston}. We then simulate 10,000 independent sample
path realizations of the discretized system of bivariate
equations.\footnote{To avoid a systematic effect from an assumed initial
condition of volatility, $\sigma_{0}^{2}$, we restart the variance process
in each simulation by drawing at random from its stationary distribution,
$\sigma_{t}^{2} \sim \text{Gamma}(2 \kappa \sigma^{2} \xi^{-2}, 2 \kappa
\xi^{-2})$.} We use two different sample sizes of $n$ = 2,340 and 23,400.
In our empirical investigation, we look at high-frequency equity data from
NYSE. With a US stock exchange trading session running from 9:30am to
4:00pm---or 6.5 hours---these sample sizes translate into receiving a new
price update every ten and one second(s). Our sample sizes are therefore
representative of more frequently traded securities.

We assume that the parameter values in the volatility equation are
$\kappa = 5$, $\sigma^{2} = 0.04$, $\xi = 0.50$ and $\rho = -0.50$, which
is broadly consistent with prior work
\citep*[e.g.,][]{ait-sahalia-mykland-zhang:11b,kalnina:11a}.
This implies that $\sigma_{t}$ is about 20\% on an average, annualized
basis, but the configuration of the model adopted here can generate a
substantial degree of intraday variation in volatility via $\xi$. An
example simulation is provided in Figure \ref{Figure:Heston}.

\begin{figure}[t!]
\caption{An illustration of a simulation from the Heston model.
\label{Figure:Heston}}
\begin{center}
\begin{tabular}{cc}
\small{Panel A: Cumulative log-return.} &
\small{Panel B: Annualized spot volatility.}\\
\includegraphics[width = 0.45\textwidth]{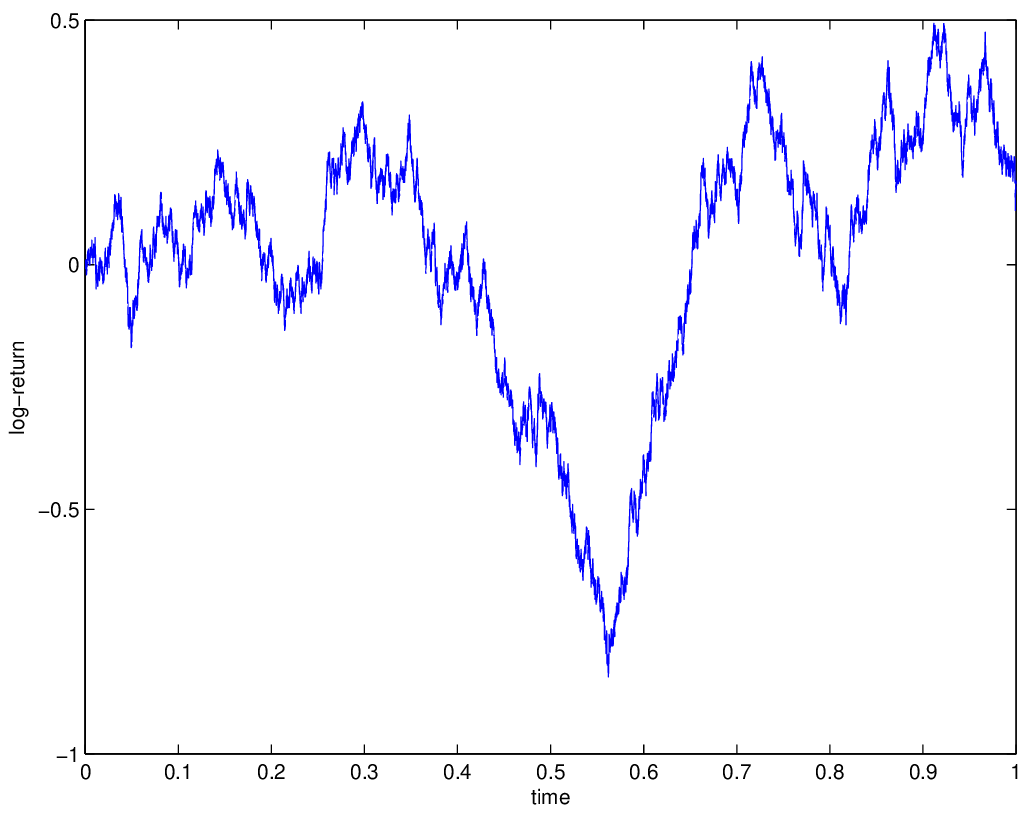} &
\includegraphics[width = 0.45\textwidth]{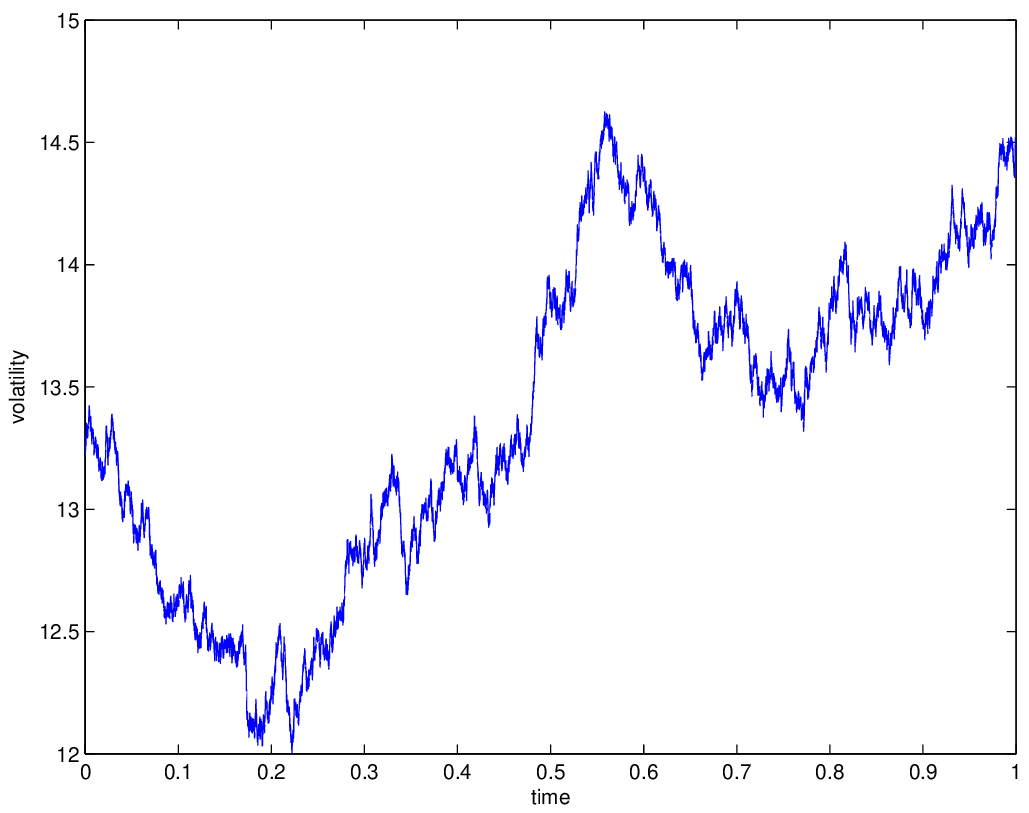}\\
\end{tabular}
\end{center}
\end{figure}

An autocorrelated and heteroscedastic noise term is added to $X$. First, we
create a set of serially dependent standard normal random variables via an MA(1)
filter: $u_{i/n} = u_{i/n}' + \zeta u_{(i-1)/n}'$,
where $\zeta$ is a parameter and $\displaystyle u_{i/n}' \overset{
\text{i.i.d.}} \sim \text{N} \bigg(0, \frac{1}{1 + \zeta^{2}} \bigg)$, independent of
$X$. Hence, $u_{i/n} \sim \text{N} \big(0, 1 \big)$ and $\text{cov}(u_{i/n},
u_{(i-1)/n}) = \displaystyle \frac{ \zeta}{1+\zeta^{2}}$, so that $\zeta$
controls the degree of first-order serial correlation in the noise. In our
simulations, we take $\zeta = -0.4$. Second, as in
\citet*{ait-sahalia-jacod-li:12a}, we set $\displaystyle
\epsilon_{i/n} = \gamma \frac{\sigma_{i/n}}{\sqrt{n}} u_{i/n}$, where $\gamma$
is the noise-ratio parameter \citep*[e.g.,][]{oomen:06a}.
This formulation implies that, conditional on $\sigma$, $\displaystyle
\omega_{i/n} = \gamma \frac{ \sigma_{i/n}}{ \sqrt{n}}$ and ensures
microstructure noise variation is conditionally heteroscedastic and
proportional to the spot volatility of the
efficient price. We assume that $\gamma = 0.50$,
which is a realistic choice for more liquid assets, see, e.g.,
\citet*[][]{christensen-oomen-podolskij:14a}.\footnote{We also experimented
with a much larger noise-ratio of $\gamma = 2$, as done by
\citet*{ait-sahalia-jacod-li:12a}. The results were almost identical,
albeit slightly worse, than those we report in the main text and, hence, are
omitted to conserve space.} Although this noise setting is not formally
covered by our theoretical frame, we include it here as a robustness check.
To alleviate the impact of noise,
we pre-average using the bandwidth $k_{n} = [ \theta \sqrt{n} ]$,
and we experiment with two choices of the tuning parameter $\theta = 1/3$
and $1$. We set the weight function $w(x) = \min(x, 1-x)$, which
has been shown to deliver nearly efficient estimates of the integrated variance,
when further parametric assumptions are imposed
\citep*[see, e.g.,][]{podolskij-vetter:09a}.

\subsection{A preliminary analysis}

We begin by verifying the conjecture made in the introduction of this
paper, namely that the finite sample properties of existing estimators of
$\Sigma^{*}$ are poor, which renders a significant fraction of such
estimates either nonpositive definite or at least ill-conditioned.
To do this, we implement the estimator suggested by
\citet*{podolskij-vetter:09a}. It is constructed by first defining
\begin{equation}
\tilde{Y}_{i,m}^{n} = \frac{1}{\sqrt{n}}
|n^{1/4} \Delta \bar{Y}_{m}^{n}|^{q_{i}}
|n^{1/4} \Delta \bar{Y}_{m+k_{n}}^{n}|^{r_{i}},
\end{equation}
and setting
\begin{equation}
\chi_{ml}^{n} = \frac{1}{2} \left[ \tilde{Y}_{i,m}^{n} \left(
\tilde{Y}_{j,m+l}^{n} - \tilde{Y}_{j,m+2k_{n}}^{n} \right) +
\tilde{Y}_{j,m}^{n} \left( \tilde{Y}_{i,m+l}^{n} - \tilde{Y}_{i,m+2k_{n}}^{n}
\right) \right],
\end{equation}
for any $0 \leq m \leq n - 4k_{n} + 1$ and $0 \leq l \leq 2k_{n}$. Then,
\begin{equation}
\label{Eqn:SigmaPV}
\tilde{\Sigma}_{ij,n}^{*} = \frac{2}{\sqrt{n}} \sum_{m = 0}^{n - 4k_{n} + 1}
\sum_{l = 0}^{2k_{n} - 1} \chi_{ml}^{n} \overset{p}{\to} \Sigma_{ij}^{*},
\end{equation}
and it follows that $\tilde{\Sigma}_{n}^{*} = \bigl( \tilde{
\Sigma}_{ij,n}^{*} \bigr) \overset{p}{\to} \Sigma^{*}$.

\renewcommand{\baselinestretch}{1.4}

\setlength{\tabcolsep}{0.30cm}
\begin{sidewaystable}[!ht]
\begin{center}
\caption{Proportion of ill-conditioned covariance matrix estimates,
$\tilde{ \Sigma}_{n}^{*}$.\label{Table:Proportion}}
\smallskip
\begin{tabular}{llcccccccccccccc}
\hline
\hline
&&& \multicolumn{5}{c}{general noise}
&& \multicolumn{5}{c}{i.i.d. noise}\\
\cline{4-8} \cline{10-14}
&&& \multicolumn{2}{c}{BM} && \multicolumn{2}{c}{SV} && \multicolumn{2}{c}{BM} && \multicolumn{2}{c}{SV} && SPY\\
\cline{4-5} \cline{7-8} \cline{10-11} \cline{13-14} \cline{16-16}
&$n =$ && 2,340 & 23,400 && 2,340 & 23,400  && 2,340 & 23,400 && 2,340 & 23,400 && $n_{\text{actual}}$ \\
\hline
\\[-0.25cm]
\multicolumn{16}{c}{2-dimensional setting}\\
\multicolumn{16}{l}{\emph{Panel A: Nonpositive definite}}\\
~~~~~$\theta=$ & 0.33 && 0.177 & 0.094 && 0.180 & 0.101 && 0.178 & 0.099 && 0.181 & 0.111 && 0.111 \\
               & 1.00 && 0.340 & 0.242 && 0.343 & 0.251 && 0.348 & 0.241 && 0.347 & 0.249 && 0.233 \\
\multicolumn{16}{l}{\emph{Panel B: Negative variance}}\\
~~~~~$\theta=$ & 0.33 && 0.141 & 0.066 && 0.149 & 0.072 && 0.140 & 0.070 && 0.146 & 0.077 && 0.064 \\
               & 1.00 && 0.286 & 0.199 && 0.290 & 0.210 && 0.292 & 0.201 && 0.289 & 0.206 && 0.184 \\
\multicolumn{16}{l}{\emph{Panel C: Condition number $\geq 20$}}\\
~~~~~$\theta=$ & 0.33 && 0.073 & 0.053 && 0.071 & 0.058 && 0.069 & 0.057 && 0.071 & 0.062 && 0.045 \\
               & 1.00 && 0.080 & 0.096 && 0.081 & 0.101 && 0.079 & 0.095 && 0.081 & 0.100 && 0.088 \\
\\[-0.25cm]
\multicolumn{16}{c}{4-dimensional setting}\\
\multicolumn{16}{l}{\emph{Panel D: Nonpositive definite}}\\
~~~~~$\theta=$ & 0.33 && 0.454 & 0.303 && 0.457 & 0.305 && 0.453 & 0.312 && 0.458 & 0.314 && 0.312 \\
               & 1.00 && 0.799 & 0.577 && 0.810 & 0.578 && 0.809 & 0.568 && 0.806 & 0.589 && 0.544 \\
\hline
\hline
\end{tabular}
\smallskip
\begin{scriptsize}
\parbox{0.95\textwidth}{\emph{Note}. We show the proportion of
ill-conditioned covariance matrix estimates, when the
\citet*{podolskij-vetter:09a} estimator $\tilde{ \Sigma}_{n}^{*}$ is used.
$\tilde{ \Sigma}_{n}^{*}$ is defined
in Eq. \eqref{Eqn:SigmaPV}. The columns report results from a general noise
model, where
$\epsilon$ is autocorrelated and heteroscedastic, and an i.i.d.
noise process, plus for a Brownian motion (BM) and stochastic volatility (SV)
model for $\sigma$. $n$ is the sample
size. The simulation design appears in Section \ref{Sec:Simulation}.
SPY is based on real high-frequency data, which are further analyzed in
Section \ref{Sec:Empirical}. $n_{\text{actual}}$ denotes the actual sample
size, which varies across days, cf. Table \ref{Table:DescriptiveStat}.
In Panel A, we report the
fraction of $\tilde{ \Sigma}_{n}^{*}$ estimates that are nonpositive definite,
i.e. with a minimum eigenvalue $\min(\lambda_{i}) \leq 0$. In Panel B, we
report how often the linear combination
$\omega = [1, -\mu_{1}^{-2}]'$ of the pre-averaged bipower variation
$V^{*}(q,r)^{n}$ leads to a negative variance estimate $\omega'\tilde{
\Sigma}_{n}^{*}\omega \leq 0$. $V^{*}(q,r)^{n}$ is defined in Eq.
\eqref{Eqn:PreavgBV}. B implies A, but not vice versa. In Panel C,
we compute the fraction of the positive definite $\tilde{ \Sigma}_{n}^{*}$
estimates that return a condition number $\text{cond}(\tilde{ \Sigma}_{n}^{*})
\geq 20$. Throughout Panels A -- C, the table is based on $q = (2,1)'$ and
$r = (0,1)'$, while Panel D reports the updated numbers from Panel A,
after changing the estimation problem to $q = (2,1,4,2)'$ and
$r = (0,1,0,2)'$.}
\end{scriptsize}
\end{center}
\end{sidewaystable}

\renewcommand{\baselinestretch}{1.6}

\afterpage{\clearpage}

Table \ref{Table:Proportion} shows several diagnostics that highlight the
properties of $\tilde{ \Sigma}_{n}^{*}$ based on $q = (2,1)'$ and $r = (0,1)'$,
i.e. pre-averaged realized variance and $(1,1)$-bipower variation.
In the table, we report the outcome from the general noise model, but we also
include a comparable i.i.d. noise environment, which we base on setting $\zeta =
0$ in the above and replacing $\sigma_{i/n}$ by $\sqrt{ \int_{0}^{1}
\sigma_{s}^{2} \text{d}s}$ in the noise variance, while noting that
\citet*{podolskij-vetter:09a} operate under the latter conditions.
In addition, the results
are also obtained for a scaled Brownian motion (BM), in which volatility has been
fixed at its steady-state value of $\sigma^{2}$.
The column SV is for the Heston stochastic volatility model, while
SPY represents some real high-frequency data that are further commented on in
Section \ref{Sec:Empirical}.

In Panel A, we report the fraction of the computed $\tilde{ \Sigma}_{n}^{*}$,
which fail to be positive definite, i.e. which have a minimum eigenvalue $\min(
\lambda_{i}) \leq 0$. It suggests that for a small sample of $n = 2,340$ and
depending on $\theta$, between 18\% -- 35\% of $\tilde{ \Sigma}_{n}^{*}$ are
nonpositive definite. As expected, these numbers decrease as the sample size
increases, but even with a fairly large sample of $n = 23,400$ the failure
rate is far from negligible. Moreover, it increases if a longer pre-averaging
horizon is employed. As such, this issue therefore has substantial bite in
practice, because larger values of $\theta$ are typically preferred, when
the noise is suspected to violate the i.i.d. assumption
\citep*[e.g.,][]{christensen-oomen-podolskij:14a,hautsch-podolskij:13a}.
Note that going from constant to stochastic volatility changes the
numbers only slightly, so allowing volatility to be time-varying has no
discernable impact on the failure rate.

Turning next to Panel B, we investigate how often the linear combination
$\omega' V^{*}(q,r)^{n}$ with $\omega = (1, -\mu_{1}^{-2})'$ results
in a negative variance estimate $\omega' \tilde{ \Sigma}_{n}^{*} \omega \leq
0$. The difference $V^{*}(2,0)^{n} - \mu_{1}^{-2}V^{*}(1,1)^{n}$ is often
used in applied work, as it provides information about presence of jumps
in the price process and permits a statistical test of this hypothesis.
Even when the covariance matrix estimate is not positive definite,
it could still result in a positive variance estimate $\omega' \tilde{
\Sigma}_{n}^{*} \omega > 0$, thereby allowing the
t-statistic to be computed. While the numbers in Panel B are lower compared
to Panel A, they are still high.

In Panel C, we look at those $\tilde{ \Sigma}_{n}^{*}$ estimates
that are positive definite. We compute the percentage of these, which return
a condition number $\text{cond}( \tilde{ \Sigma}_{n}^{*}) \geq 20$.%
\footnote{The condition number of an
invertible matrix $\text{A}$ is defined as $\text{cond}(\text{A}) =
||\text{A}|| \cdot ||\text{A}^{-1}||$, where $|| \cdot ||$ is the L$_{2}$
matrix norm. $\text{cond}(\text{A})$ can be shown to be the ratio between
the largest
and smallest singular value of $\text{A}$. It can be interpreted as
saying how much small changes in the input matrix get
amplified under matrix inversion.}
The condition number measures the numerical accuracy of a matrix, and
a value above $20 = 10 \times \text{dim}( \tilde{ \Sigma}_{n}^{*})$ is
generally taken as a sign of an ill-conditioned and nearly singular matrix
\citep*[e.g.,][]{greene:11a,hautsch-kyj-oomen:12a}. As readily seen, we find
that about 5\% -- 10\% of the $\tilde{ \Sigma}_{n}^{*}$ that are deemed ok
by a definiteness criteria show signs of being badly scaled.

Lastly, in Panel D we attempt to estimate the joint asymptotic covariance
matrix of a 4-dimensional parameter by using $q = (2,1,4,2)'$ and
$r = (0,1,0,2)'$. We then report the percentage of the $\tilde{
\Sigma}_{n}^{*}$ estimates, which are not positive definite, i.e. the numbers
can be compared to Panel A. Not surprisingly, increasing the complexity of the
problem has a detrimental impact on the estimation errors, and up to 80\% of
the estimates are now nonpositive definite, making this a devastating issue for
inference
\citep*[e.g.,][employ a 3-dimensional statistic to test for the parametric
form of volatility in diffusion models (both with or without noise)]{
dette-podolskij-vetter:06a,dette-podolskij:08a,vetter-dette:12a}.

\subsection{Implementation of the subsampler}

We now turn to the subsampler, where we again base our investigation on
$V^{*}(q,r)^{n}$ using the parameters $q = (2,1)'$ and $r = (0,1)'$.
As $\hat{ \Sigma}_{n}^{*}$ depends on two
tuning parameters, $p$ and $L$, we compute
it by varying these across a broad range of values in order to gauge the
sensitivity of our estimator to specific choices. We set
$p = 3, 5$ and $10$, so that the block length of noisy returns before
pre-averaging goes from three to ten times the pre-averaging horizon
$k_{n}$. Moreover, we slice the sample into $L = 5, 10$ and $15$
subsamples, yielding a
total of nine combinations of $p$ and $L$.\footnote{This implies that for
the sample size $n = 2,340$, there are some combinations of $p$ and $L$,
for which there is not enough data to compute the subsampler. Therefore,
we restrict attention to $n = 23,400$ in the following. The results for
$n = 2,340$, when attainable, are not materially worse, reflecting the
slow rates of convergence, and all are available by request.}

Our initial set of simulations suggested that, for small $p$ and $L$, the
raw estimator defined by Eq. \eqref{Eqn:SigmaNStar} is downward biased,
thereby leading to a systematic underestimation of $\Sigma^{*}$. We
therefore start by briefly outlining a few corrections that
are important in finite samples.

First, in order to center the $V_{l}^{*}(q,r)^{n}$, we should in theory
use the unobserved
$V^{*}(q,r)$, which we are forced to replace by a feasible, consistent
estimator, i.e.
$V^{*}(q,r)^{n}$. While this has no impact asymptotically, because
$V^{*}(q,r)^{n}$ converges much faster than $V_{l}^{*}(q,r)^{n}$, a
closer inspection of $\hat{ \Sigma}_{n}^{*}$ shows that the substitution does
entail a standard small sample correction. This implies that $\hat{
\Sigma}_{n}^{*}$ should be divided by $1 - 1/L$, i.e. the ``right''
normalization in Eq. \eqref{Eqn:SigmaNStar} is $L - 1$ and
not $L$.

Second, there is a HAC error associated with $p$,
which---in contrast to the $L$ correction---is more subtle to deal with. The
problem originates from the estimation of the autocovariances of the
pre-averaged returns, $|n^{1/4} \Delta \bar{Y}_{i}^{n}|^{q_{k}} |
n^{1/4} \Delta \bar{Y}_{i+k_{n}}^{n}|^{r_{k}}$, as exemplified by Eq.
\eqref{Eqn:Example}. As such, it depends on the
covariance structure of this series, which, in turn, is a function of several
parameters and variables, including spot volatility, the weight
function and the
variance of the noise process \citep*[see, e.g., the proof of Theorem
\ref{m11} in the Appendix around Eq. \eqref{RE1}, or][]{podolskij-vetter:09a}.
If we approximate this function, e.g., by assuming that volatility
is constant, a detailed calculation (which is omitted here, but available upon
request) shows that for the pre-averaged bipower variation estimator
$V^{*}(q_{k}, r_{k})^{n}$ defined by Eq. \eqref{Eqn:PreavgBV}, we can
roughly correct for the $p$ error by dividing $\hat{ \Sigma}_{n}^{*}$ with
$1 - 1/p$.

It turns out, however, that this is too much, if either $q_{k}$
or $r_{k}$ is zero, as it happens for the pre-averaged realized variance,
$V^{*}(2,0)^{n}$. This is because the summands in Eq. \eqref{Eqn:PreavgBV}
are, asymptotically, $2k_{n}$-dependent for non-zero values of both $q_{k}$
and $r_{k}$, while
they are only $k_{n}$-dependent, if either is zero. Thus, the bias induced
by $p$ in the latter is, loosely speaking, half that of the former. This
indicates that we
ought to divide the elements of $\hat{ \Sigma}_{n}^{*}$ involving
the variance of $V^{*}(2,0)^{n}$ and its covariance with
$V^{*}(1,1)^{n}$ only by $1 - 0.5/p$. This is
less appealing, because it would break the positive semi-definiteness property
of $\hat{ \Sigma}_{n}^{*}$. We therefore proceed by using a constant scaling
for the entire matrix, and---to strike a balance between the two
alternatives---we propose to meet in the middle and rescale
$\hat{ \Sigma}_{n}^{*}$ by $(1 - 0.75/p)$. This choice produces excellent
results for the values of $p$ considered in this paper, as corroborated by our
numerical experiments below, while for $p \geq 10$ the correction is of
limited importance and can be ignored.

Third, our theoretical results hinge on $n$ being a
multiple of $Lpk_{n}$. In practice, where $n$ varies randomly over time and can
be very odd, this is an unrealistic assumption, which is almost never
satisfied.
Instead, for a given choice of $k_{n}$, $p$ and $L$, the maximum number
of blocks of length $pk_{n}$ that can be assigned to each of the $L$ subsample
estimates $V_{l}^{*}(q_{k},r_{k})^{n}$ is:
\begin{equation}
\displaystyle n_{\text{block}} = \biggl\lfloor \frac{ \left\lfloor \displaystyle
n/pk_{n} \right \rfloor}{L} \biggr\rfloor.
\end{equation}
The effective amount of data used to construct the subsampler is
therefore often less than the total sample size, i.e. $n_{ \text{block}}Lpk_{n}
\leq n$.

We compute $\hat{ \Sigma}_{n}^{*}$ from the data
that fall within the window $[0, n_{ \text{block}}Lpk_{n}/n]$ and subsequently
inflate this
estimate to cover the whole unit interval. While this entails
some loss of information about the underlying variation of the process towards
the end of the sample, in our experience this has a very limited influence on
the results, unless the data, from which $\hat{ \Sigma}_{n}^{*}$ is computed, is
not representative of the overall level of volatility. In practice, one can
minimize this effect by
choosing the parameters, such that $n_{ \text{block}}Lpk_{n}$ is close to $n$.

\subsection{Results}

In Figure \ref{Fig:Univariate}, we plot some kernel smoothed density
estimates of the standardized pre-averaged bipower variation, i.e. $n^{1/4}
\bigl( V^{*}(q_{k},r_{k})^{n} - V^{*}(q_{k},r_{k}) \bigr) / \sqrt{ \hat{
\Sigma}_{kk,n}^{*}}$, where $\hat{ \Sigma}_{kk,n}^{*}$ is the $k$th
diagonal element of our subsampling covariance matrix estimate $\hat{
\Sigma}_{n}^{*}$. Here, we use $\theta = 1$. The results in Panels A -- B
are for $V^{*}(1,1)^{n}$,
while Panels C -- D are for $V^{*}(2,0)^{n}$. In addition, the left-hand
portion of the figure is for $L = 15$ and $p$ changing, while the right-hand
part is based on $p = 10$ and for different $L$. The infeasible result
for $V^{*}(2,0)^{n}$ replaces the subsampler with the true variance, which
is known here (cf. footnote \ref{Foot:avar_rv}).

As the figure shows, the studentized pre-averaging estimators tend to track
the asymptotic normal approximation closely
across combinations of $p$ and $L$. The sole exception, appearing
in Panel A, is for $V^{*}(1,1)^{n}$, when $\hat{ \Sigma}_{n}^{*}$ is
implemented using $p = 3$. Note that if $p = 3$, the effective sample size
within a block of noisy high-frequency data is $k_{n}+2$ after pre-averaging,
whereas the summands of
$V^{*}(1,1)^{n}$ are $2k_{n}$-dependent. With such a small value of $p$,
the block length is therefore inadequate to permit estimation of all the
required autocovariances. As $|n^{1/4} \Delta \bar{Y}_{i}^{n}|^{q_{k}} |
n^{1/4} \Delta \bar{Y}_{i+k_{n}}^{n}|^{r_{k}}$ is strongly positively
autocorrelated, this leads to
a severe underestimation of the true variation of $V^{*}(1,1)^{n}$ and,
hence,
a pronounced overdispersion of the estimated density.
As a practical guide, one should therefore avoid computing $\hat{
\Sigma}_{n}^{*}$ with $p = 3$, if both $q_{k}$ and $r_{k}$ are different
from zero, for any $i = 1, \dots, m$.
In comparison, the
corresponding graph for $V^{*}(2,0)^{n}$ in Panel C is much better scaled,
reflecting the lesser dependence inherent in this estimator. Apart
from that, the fit tends to improve for larger values of $p$ and $L$, as
expected.\footnote{Of course, with a fixed sample size $n$ and
pre-averaging window $k_{n}$, the parameters $p, L$, and $n_{ \text{block}}$ are
not free. Thus, everything else is not ``fixed'', because
$n_{ \text{block}}$ is decreasing for larger values of either $p$ or $L$ (while
holding the other fixed).} Lastly, comparing with the
infeasible result in Panels C -- D, we note that the estimated densities appear
slightly negatively skewed, owing to
a modest, positive correlation between $V^{*}(q_{k},r_{k})^{n}$ and $\hat{
\Sigma}_{kk,n}^{*}$.

\begin{figure}[t!]
\caption{Kernel density estimate of the standardized
$V^{*}(q_{k},r_{k})^{n}$: changing $p$ and $L$. \label{Fig:Univariate}}
\begin{center}
\begin{tabular}{cc}
\footnotesize{Panel A: $V^{*}(1,1)^{n}$ ($L = 15$)} &
\footnotesize{Panel B: $V^{*}(1,1)^{n}$ ($p = 10$)}\\
\includegraphics[height=6cm,width=0.45\textwidth]{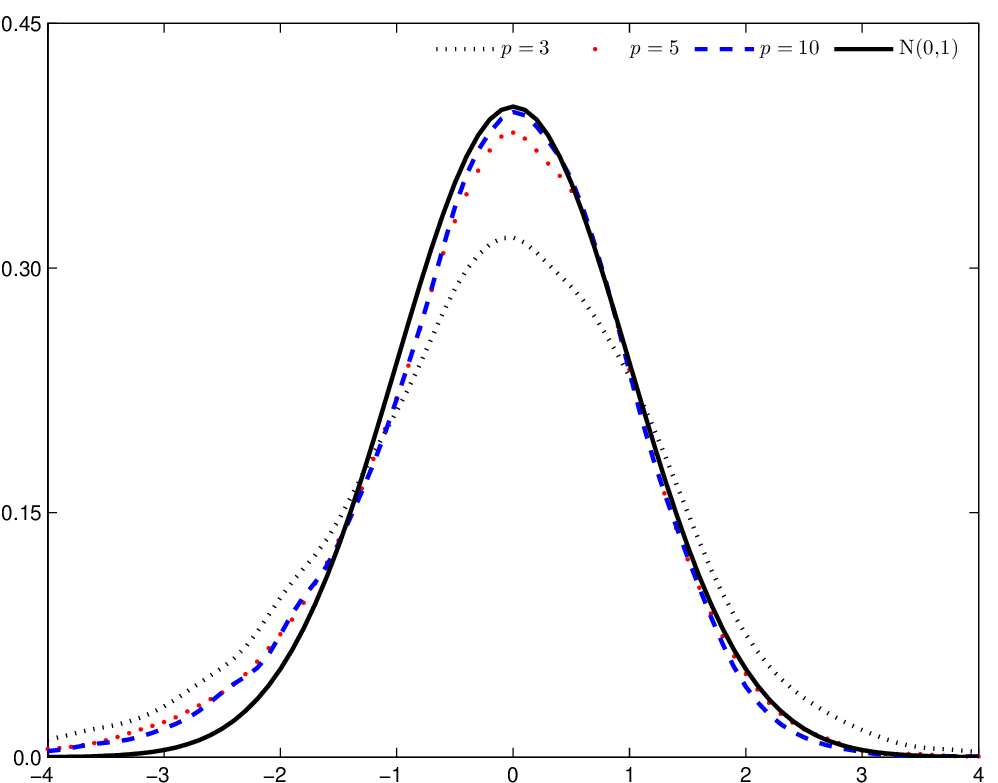} &
\includegraphics[height=6cm,width=0.45\textwidth]{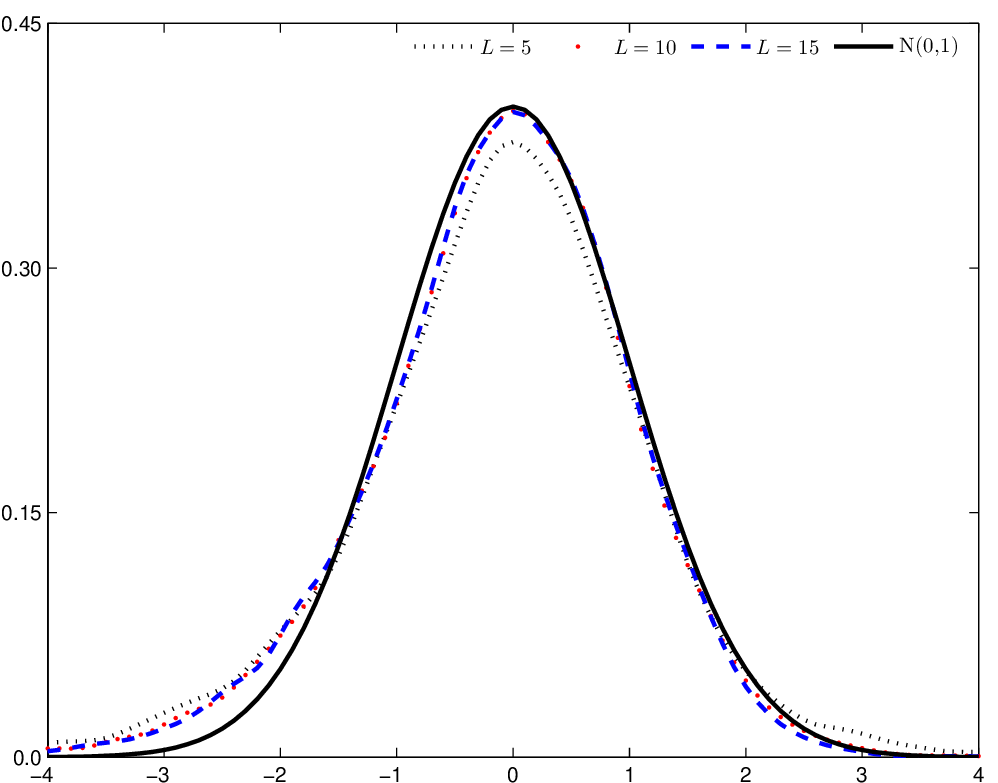}\\
\footnotesize{Panel C: $V^{*}(2,0)^{n}$ ($L = 15$)} &
\footnotesize{Panel D: $V^{*}(2,0)^{n}$ ($p = 10$)}\\
\includegraphics[height=6cm,width=0.45\textwidth]{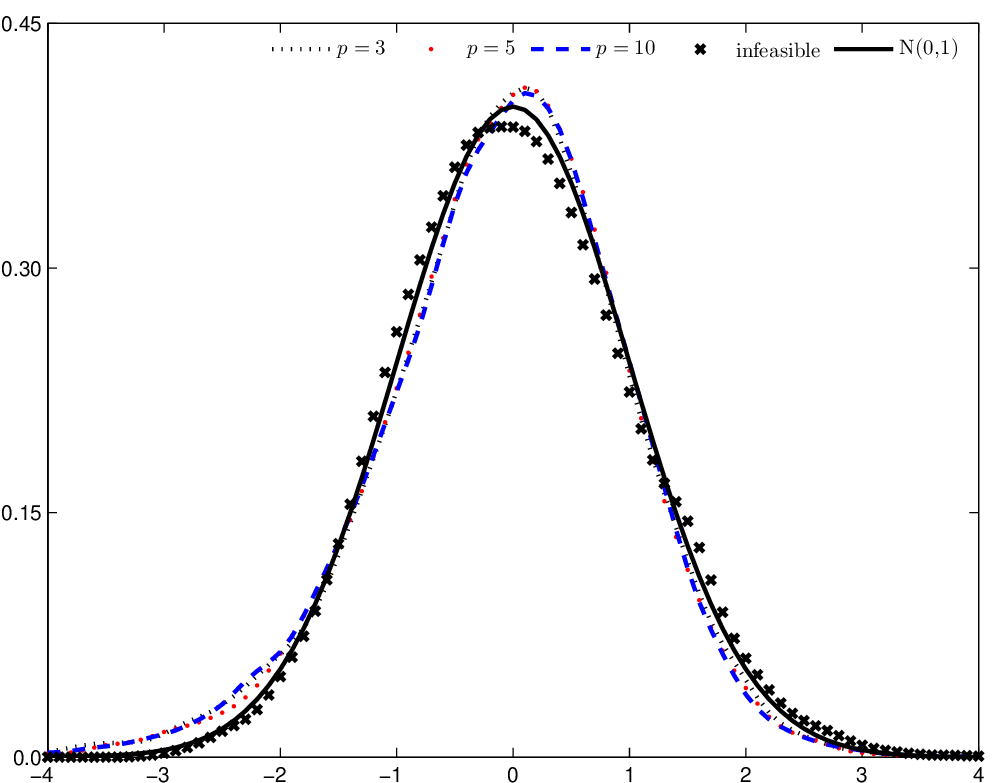} &
\includegraphics[height=6cm,width=0.45\textwidth]{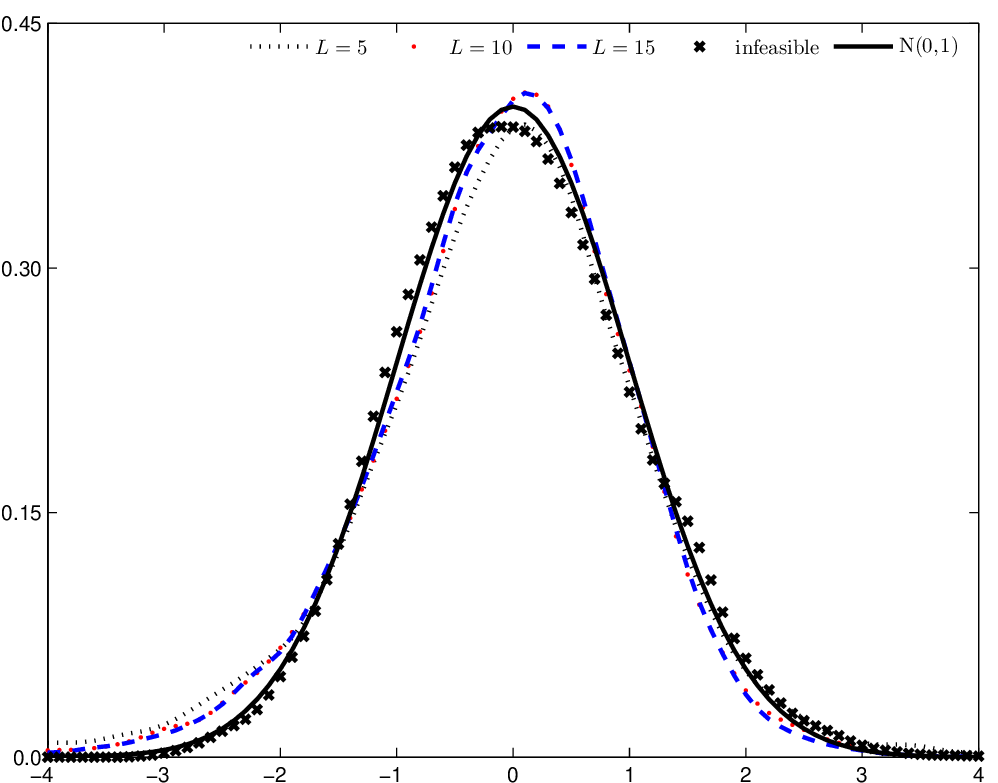}\\
\end{tabular}
\begin{scriptsize}
\parbox{0.925\textwidth}{\emph{Note}. We show the
kernel smoothed density estimates of the standardized pre-averaged bipower
variation estimator: $n^{1/4} \bigl( V^{*}(q_{k},r_{k})^{n} - V^{*}(q_{k},r_{k})
\bigr) / \sqrt{ \hat{ \Sigma}_{kk,n}^{*}}$, where $\hat{ \Sigma}_{kk,n}^{*}$
is
the $k$th diagonal element of $\hat{ \Sigma}_{n}^{*}$. Throughout this figure,
$n = 23,400$, $\theta = 1$ and we set $k_{n} = [\theta \sqrt{n}]$ to implement
pre-averaging. The simulation data is from a Heston stochastic volatility
model, as described in the main text. In the left panel, the subsampler is based
on $L = 15$ subsamples, while varying the block length at $p = 3, 5$ or $10
\times k_{n}$. The right panel is based on $p = 10$,
while changing the number of subsamples at $L = 5, 10$ or $15$. The $n_{
\text{sim}} = 10,000$ simulated t-statistics
are smoothed using a Gaussian kernel with optimal bandwidth
selection $h = 1.06\hat{ \sigma}n_{ \text{sim}}^{-1/5}$, where $\hat{ \sigma}$
is the sample standard deviation of the data. The infeasible result
for $V^{*}(2,0)^{n}$ replaces the subsampler with the true variance
(cf., footnote \ref{Foot:avar_rv}). The density function of a standard
normal random variable (the solid black line) is superimposed as a visual
reference.}
\end{scriptsize}
\end{center}
\end{figure}

We turn next to Figure \ref{Fig:UnivariateTheta}, where we explore how
sensitive our findings are to the choice of $\theta$. In this figure,
and throughout the remainder of this section, we fix the parameters of
$\hat{ \Sigma}_{n}^{*}$ to $p = 10$ and $L = 15$. As apparent from both
panels, the relatively large change in $\theta$ has only a
minuscule effect on the shape of the estimated density for $V^{*}(1,1)^{n}$
and $V^{*}(2,0)^{n}$.

\begin{figure}[t!]
\caption{Kernel density estimate of the standardized $V^{*}(q_{k},r_{k})^{n}$:
changing $\theta$. \label{Fig:UnivariateTheta}}
\begin{center}
\begin{tabular}{cc}
\footnotesize{Panel A: $V^{*}(1,1)^{n}$ $(p = 10, L = 15)$} &
\footnotesize{Panel B: $V^{*}(2,0)^{n}$ $(p = 10, L = 15)$}\\
\includegraphics[height = 6cm, width = 0.45\textwidth]{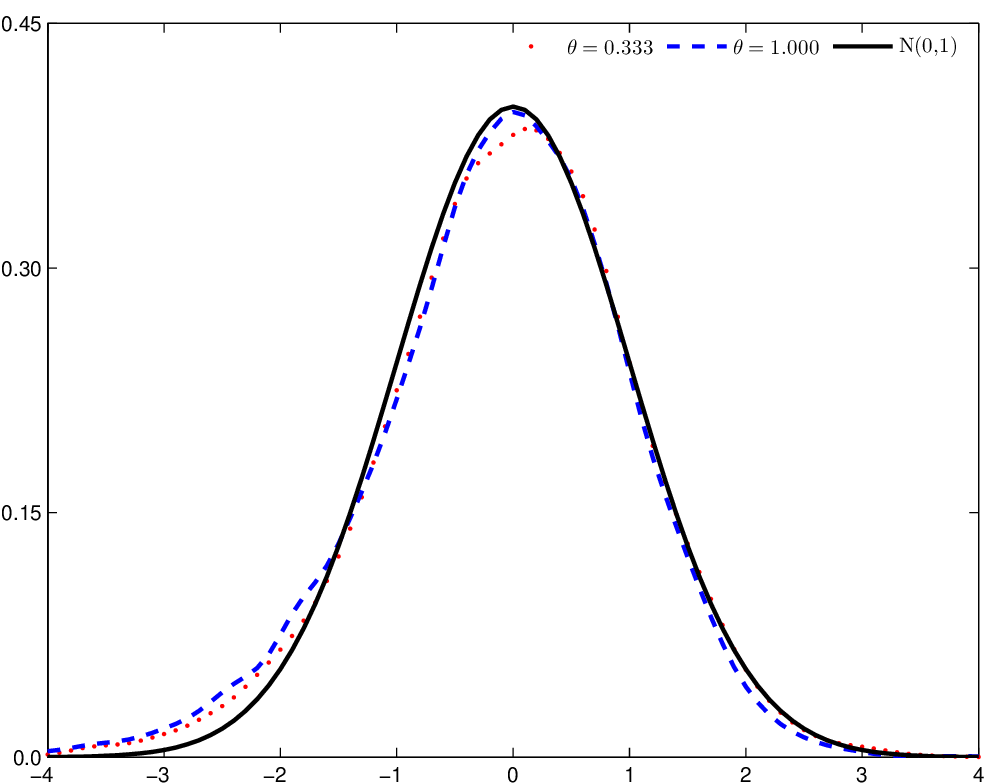} &
\includegraphics[height = 6cm, width = 0.45\textwidth]{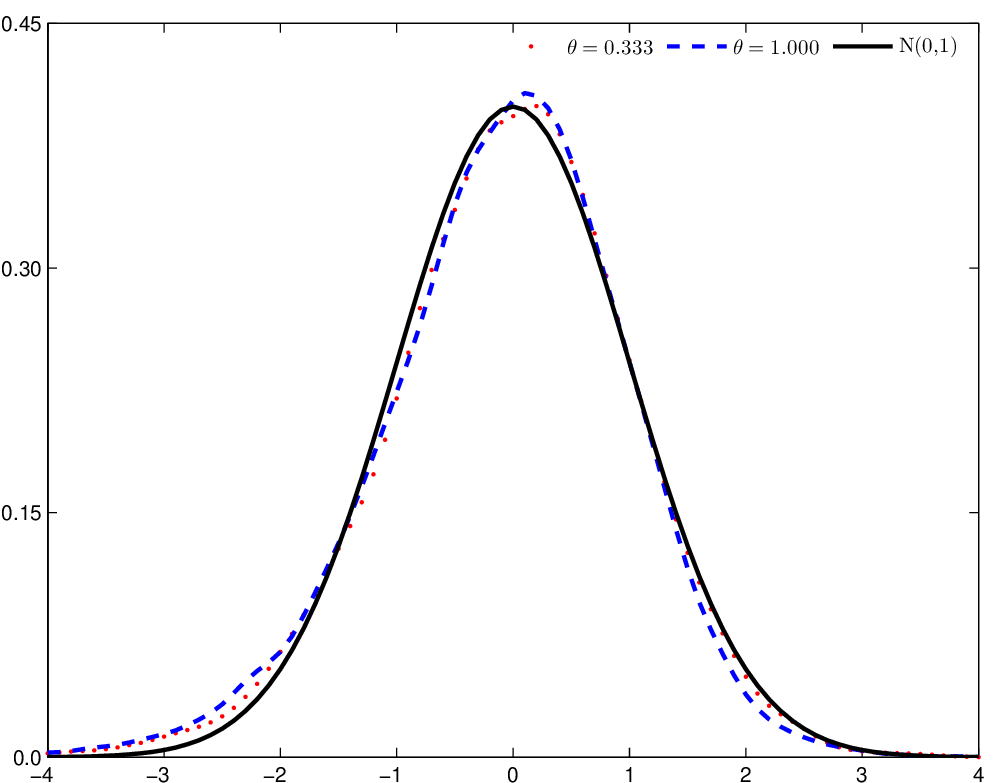}\\
\end{tabular}
\begin{scriptsize}
\parbox{0.925\textwidth}{\emph{Note}. We show the
kernel smoothed density estimates of the standardized pre-averaged bipower
variation estimator: $n^{1/4} \bigl( V^{*}(q_{k},r_{k})^{n} - V^{*}(q_{k},r_{k})
\bigr) / \sqrt{ \hat{ \Sigma}_{kk,n}^{*}}$, where $\hat{ \Sigma}_{kk,n}^{*}$
is
the $k$th diagonal element of $\hat{ \Sigma}_{n}^{*}$. Throughout this figure,
$n = 23,400$, $L = 15$, $p = 10$, and we set $k_{n} = [\theta \sqrt{n}]$ to
implement pre-averaging using $\theta = 1/3$ and $\theta = 1$.
The simulation data is from a Heston stochastic volatility model, as
described in the main text. The left panel holds the results for
$V^{*}(1,1)^{n}$, while the right panel is for $V^{*}(2,0)^{n}$.
The $n_{ \text{sim}} = 10,000$ simulated t-statistics
are smoothed using a Gaussian kernel with optimal bandwidth
selection $h = 1.06\hat{ \sigma}n_{ \text{sim}}^{-1/5}$, where $\hat{ \sigma}$
is the sample standard deviation of the data. The density function of a standard
normal random variable (the solid black line) is superimposed as a visual
reference.}
\end{scriptsize}
\end{center}
\end{figure}

In Figure \ref{Fig:SimJV}, using $\theta = 1$, we look at an application that
requires one to use information about the full covariance matrix estimate
by reporting some results
for the linear combination $\omega' V^{*}(q,r)^{n}$ with $\omega =
(1, -\mu_{1}^{-2})'$, i.e. $V^{*}(2,0)^{n} - \mu_{1}^{-2} V^{*}(1,1)^{n}$.
In Table \ref{Table:Proportion}, we noted it was problematic to standardize
this difference with the covariance matrix estimator put forth by
\citet*{podolskij-vetter:09a}, which was often found to be nonpositive
definite. The subsampler does not suffer from this issue.
In Panel A of Figure \ref{Fig:SimJV}, we
therefore plot the time series of the studentized statistic across the
simulations runs, using the delta method to conclude that $n^{1/4} \omega'
\bigl( V^{*}(q,r)^{n} - V^{*}(q,r) \bigr) / \sqrt{ \omega' \hat{
\Sigma}_{n}^{*} \omega} \overset{d}{ \to} \text{N}(0,1)$. Panel B
inspects the kernel smoothed density estimate of the t-statistic. As
evident, the asymptotic distribution theory is a decent description of
the actual finite sample variation, although the fit is not perfect. As
explained above, $V^{*}(2,0)^{n} - \mu_{1}^{-2} V^{*}(1,1)^{n}$
provides information about the presence of jumps in asset prices, and
significant positive values would lend support to this hypothesis.
Here, the shape of the estimated density implies that the one-sided
coverage probabilities of the t-statistic are slightly too
large in the right tail. This would render such a hypothesis test mildly
conservative, which is preferable in practice (e.g., using the 95\%
quantile from the standard normal distribution gives a coverage rate of
96.7\% in the above figure).
We note that it requires additional properties of $\hat{ \Sigma}_{n}^{*}$
for such a test to retain any power under the alternative. That is,
$\hat{ \Sigma}_{n}^{*}$ is not jump-robust in its current form, so it
needs to be accompanied by truncation, as shown in the noiseless setting
in Theorem \ref{theorem:truncation}.
We shall explore the joint impact of noise and price jumps on the
subsampler in a
companion paper, and the results could of course
change with a different, albeit related, estimator of $\Sigma^{*}$.

\begin{figure}[t!]
\caption{Properties of the standardized $V^{*}(2,0)^{n} -
\mu_{1}^{-2} V^{*}(1,1)^{n}$.
\label{Fig:SimJV}}
\begin{center}
\begin{tabular}{cc}
\footnotesize{Panel A: Point estimate} &
\footnotesize{Panel B: Kernel density estimate}\\
\includegraphics[height = 6cm, width = 0.45\textwidth]{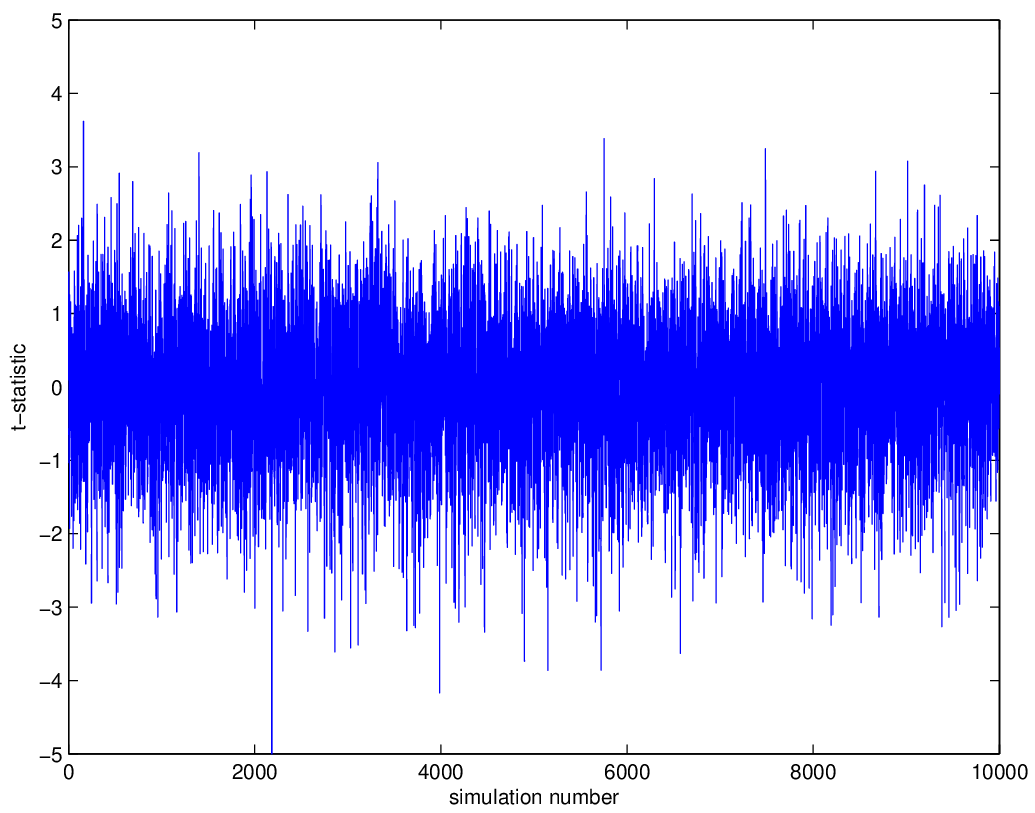} &
\includegraphics[height = 6cm, width = 0.45\textwidth]{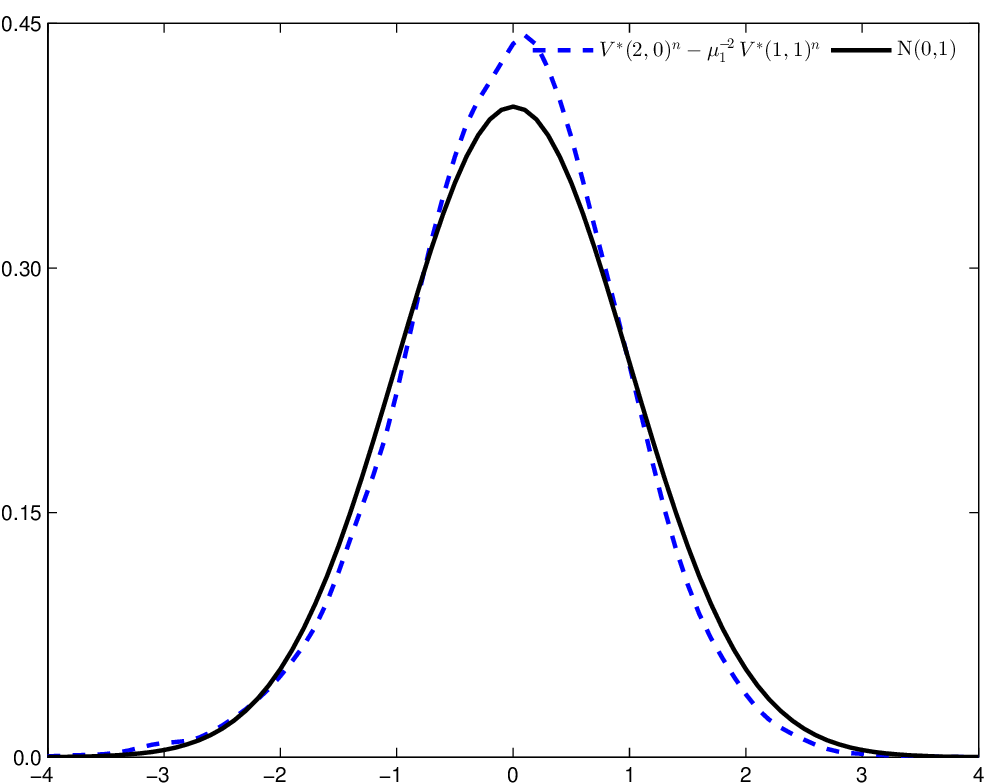}\\
\end{tabular}
\begin{scriptsize}
\parbox{0.925\textwidth}{\emph{Note}. We plot $\omega' V^{*}(q,r)^{n}$ with
$\omega = (1, -\mu_{1}^{-2})'$, after it has been standardized by
$\hat{ \Sigma}_{n}^{*}$
based on $p = 10$ and $L = 15$, i.e. $n^{1/4} \bigl( V^{*}(2,0)^{n} - \mu_{1}^{-2}
V^{*}(1,1) \bigr) / \sqrt{ \omega' \hat{ \Sigma}_{n}^{*} \omega}$. In Panel A,
we plot the point
estimates of this t-statistic across simulations, while
Panel B displays the corresponding kernel smoothed density estimate. Throughout
the figure, $n = 23,400$, and we set $k_{n} = [\theta \sqrt{n}]$ to implement
pre-averaging using $\theta = 1$. The simulation data is from a
Heston stochastic volatility model, as described in the main text.
In Panel B, the $n_{ \text{sim}} = 10,000$ simulated t-statistics
are smoothed using a Gaussian kernel with optimal bandwidth
selection $h = 1.06\hat{ \sigma}n_{ \text{sim}}^{-1/5}$, where $\hat{ \sigma}$
is the sample standard deviation of data. The density function of a standard
normal random variable (the solid black line) is superimposed as a visual
reference.}
\end{scriptsize}
\end{center}
\end{figure}

At last, we compare our subsampler
$\hat{ \Sigma}_{n}^{*}$ to an alternative, nonparametric estimator
of $\Sigma^{*}$, namely the observed asymptotic
variance (AVAR) of \citet*{mykland-zhang:17a}. As in our setting, the observed
AVAR is based on squared increments (or outer products) of the original
statistic(s) computed on smaller streches of high-frequency data, but there are
several differences between the construction of the subsampler and observed
AVAR. Moreover, there is little guidance on how to select tuning
parameters for the latter. We therefore proceed as follows. The
sampling grid consists (using their notation) of $B = L$ blocks at
the outset. This helps to ensure comparability with $\hat{ \Sigma}_{n}^{*}$.
We then compute the observed AVAR using a two-scale approach, as a linear
combination of the $K$-averaged apparent quadratic covariation with
$K_{1} = 1$ and $K_{2} = 2$;
see Eq. (24) in \citet*{mykland-zhang:17a}.\footnote{The observed AVAR has
a bias, which---although asymptotically negligible---could impair its
accuracy in finite samples. The virtue of the two-scale construction, as
advocated by \citet*{mykland-zhang:17a}, is that
the bias term cancels out.}
A forward half-interval approach is adopted to reduce the impact of edge
effects induced by pre-averaging.
The outcome is reported in Figure \ref{figure:mykland-zhang}, where we plot
the standardized pre-averaged bipower variation $V^{*}(2,0)^{n}$ and
$V^{*}(1,1)^{n}$ using both $\hat{ \Sigma}_{n}^{*}$ and the observed AVAR.
As apparent, standardization with the subsampler tracks the standard normal
curve closer compared to the observed AVAR. Indeed, the standard error of the
studentized pre-averaged bipower variation is about 1.05 using $\hat{
\Sigma}_{n}^{*}$, while it is about 1.20 using the observed AVAR.

\begin{figure}[t!]
\caption{Comparison of the subsampler and observed asymptotic variance.
\label{figure:mykland-zhang}}
\begin{center}
\begin{tabular}{cc}
\footnotesize{Panel A: $V^{*}(1,1)^{n}$} &
\footnotesize{Panel B: $V^{*}(2,0)^{n}$}\\
\includegraphics[height = 6cm, width = 0.45\textwidth]{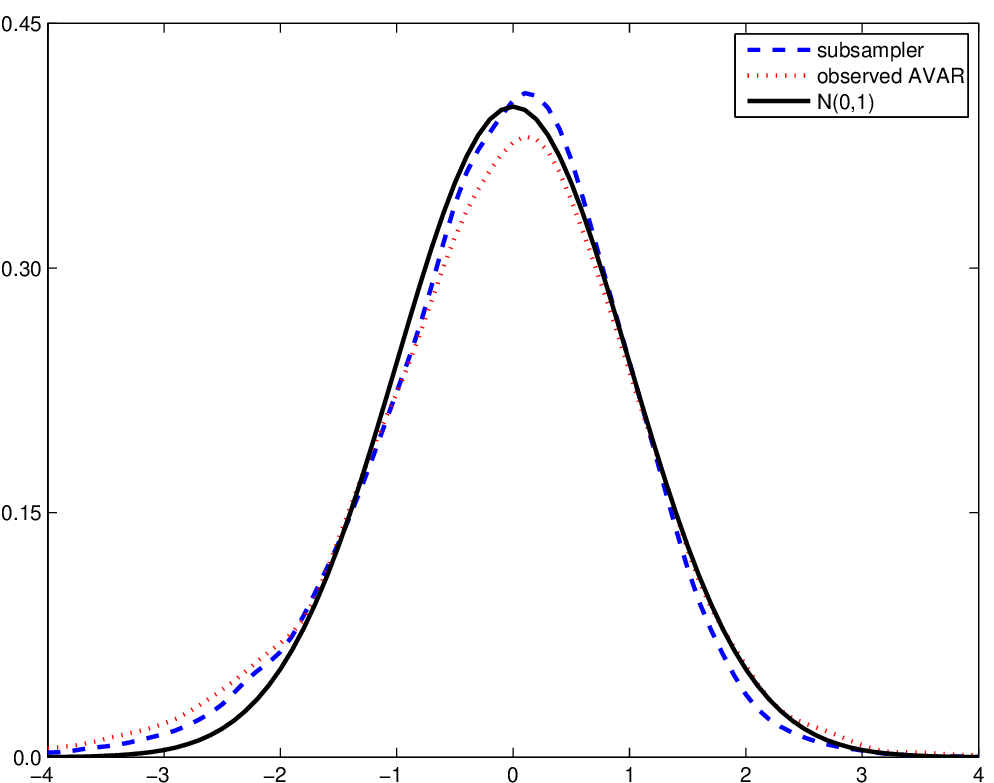} &
\includegraphics[height = 6cm, width = 0.45\textwidth]{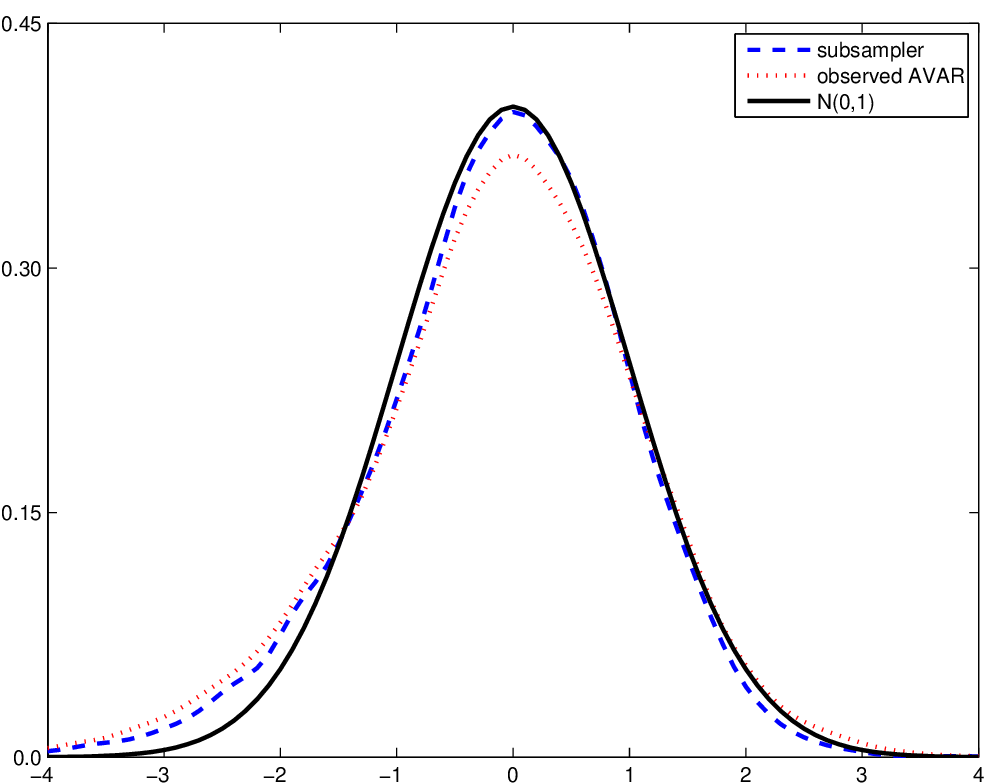}\\
\end{tabular}
\begin{scriptsize}
\parbox{0.925\textwidth}{\emph{Note}. We show
kernel smoothed density estimates of the standardized pre-averaged bipower
variation estimator: $n^{1/4} \bigl( V^{*}(q_{k},r_{k})^{n} - V^{*}(q_{k},r_{k})
\bigr) / \sqrt{ \Sigma_{kk}^{*}}$, where $\Sigma_{kk}^{*}$ is the $k$th
diagonal element of $\Sigma^{*}$. We replace $\Sigma^{*}$ by the subsampler
(based on $L = 15$ and $p = 10$) and the observed asymptotic variance.
The latter is computed with $B = 15$, and a
two-scale combination of the $K$-averaged apparent quadratic covariation with
$K_{1} = 1$ and $K_{2} = 2$ using
forward half-interval estimators, as explained in \citet*{mykland-zhang:17a}
around Eq. (24).
In the figure, $n = 23,400$ and we set $k_{n} = [\theta \sqrt{n}]$ to
implement pre-averaging using $\theta = 1$.
The simulation data is from a Heston stochastic volatility model, as
described in the main text. The left panel holds the results for
$V^{*}(1,1)^{n}$, while the right panel is for $V^{*}(2,0)^{n}$.
The $n_{ \text{sim}} = 10,000$ simulated t-statistics
are smoothed using a Gaussian kernel with optimal bandwidth
selection $h = 1.06\hat{ \sigma}n_{ \text{sim}}^{-1/5}$, where $\hat{ \sigma}$
is the sample standard deviation of the data. The density function of a standard
normal random variable (the solid black line) is superimposed as a visual
reference.}
\end{scriptsize}
\end{center}
\end{figure}

Overall, the simulation results suggest that inference based on
$\hat{ \Sigma}_{n}^{*}$ is fairly robust and delivers excellent outcomes,
even for modest values of its tuning parameters.

\section{Empirical work}
\label{Sec:Empirical}

Here, we provide a brief illustration of the subsample estimator in the context
of some real financial high-frequency data. We analyze tick-data
from the Standard and Poor's depository receipts, which is an exchange-traded
fund that tracks the performance of the S\&P 500 stock index.
The shares are listed on several U.S. stock exchanges and
trade under the ticker symbol SPY. It is a highly liquid
security and provides a good starting point for the subsampler, which is
data-intensive. We extracted
a transaction price series of the SPY from the TAQ database. The data are
recorded at milli-second precision and our complete sample covers the
time period from January, 2007 to March, 2011; a total of 1,169 business days.
The raw data was
filtered for outliers using the recommendations of
\citet*{christensen-oomen-podolskij:14a}.\footnote{The
\citet*{christensen-oomen-podolskij:14a} filter is to a large extent based
on the cleaning routines of
\citet*{barndorff-nielsen-hansen-lunde-shephard:09a}. The former
use a backward-forward matching algorithm to compare a trade to the
quote conditions prevailing in the market around the time of the transaction.
The latter evaluate each trade against a single preceding bid-ask quote, which
may lead to excessive removal of data in fast-moving markets.
Apart from that, the filters are identical.} We also restrict
attention to the NYSE trading session, which runs from 9:30am till 4:00pm
Eastern Time. Table \ref{Table:DescriptiveStat} provides a few descriptive
statistics that summarize key features of the dataset.

\setlength{\tabcolsep}{0.45cm}
\begin{table}[!ht]
\begin{center}
\caption{Summary statistics of the SPY high-frequency data.
\label{Table:DescriptiveStat}}
\smallskip
\begin{tabular}{lrcccccccc}
\hline
\hline
Statistic & Sample average & [Min; Max]\\
$r_{oc}$ & 0.003 & [-8.254;7.349]\\
$\hat{ \sigma}_{r_{oc}}$ & 17.037 & [3.474;127.923]\\
$n$ & 113.769 & [13.127;533.203]\\
$K$ & 320 & [115;730]\\
\hline
\hline
\end{tabular}\smallskip
\begin{scriptsize}
\parbox{0.95\textwidth}{\emph{Note.} We report some descriptive statistics of
the SPY high-frequency data. $r_{oc}$ is the open-to-close return
(in percent), i.e. the difference between the log-price of the last and
first transaction of the day. $\hat{ \sigma}_{r_{oc}}$ is a realized measure of
the standard deviation of $r_{oc}$. We set $\hat{ \sigma}_{r_{oc}} =
100 \times \sqrt{ 250 \times \widehat{IV}}$, where
$\widehat{IV}$ is defined in Eq. \eqref{Eqn:ivhat}. $n$ is the
sample size (in 1,000s), while $K$ is the pre-averaging window. The sample
period is January, 2007 through March, 2011.}
\end{scriptsize}
\end{center}
\end{table}

\begin{figure}[!ht]
\caption{Autocorrelation function of SPY return series.\label{Fig:acf}}
\begin{center}
\begin{tabular}{cc}
\footnotesize{Panel A: Noisy returns, $\Delta_{i}^{n} Y$} &
\footnotesize{Panel B: Pre-averaged returns, $\Delta \bar{Y}_{i}^{n}$}\\
\includegraphics[height = 6cm, width = 0.45\textwidth]{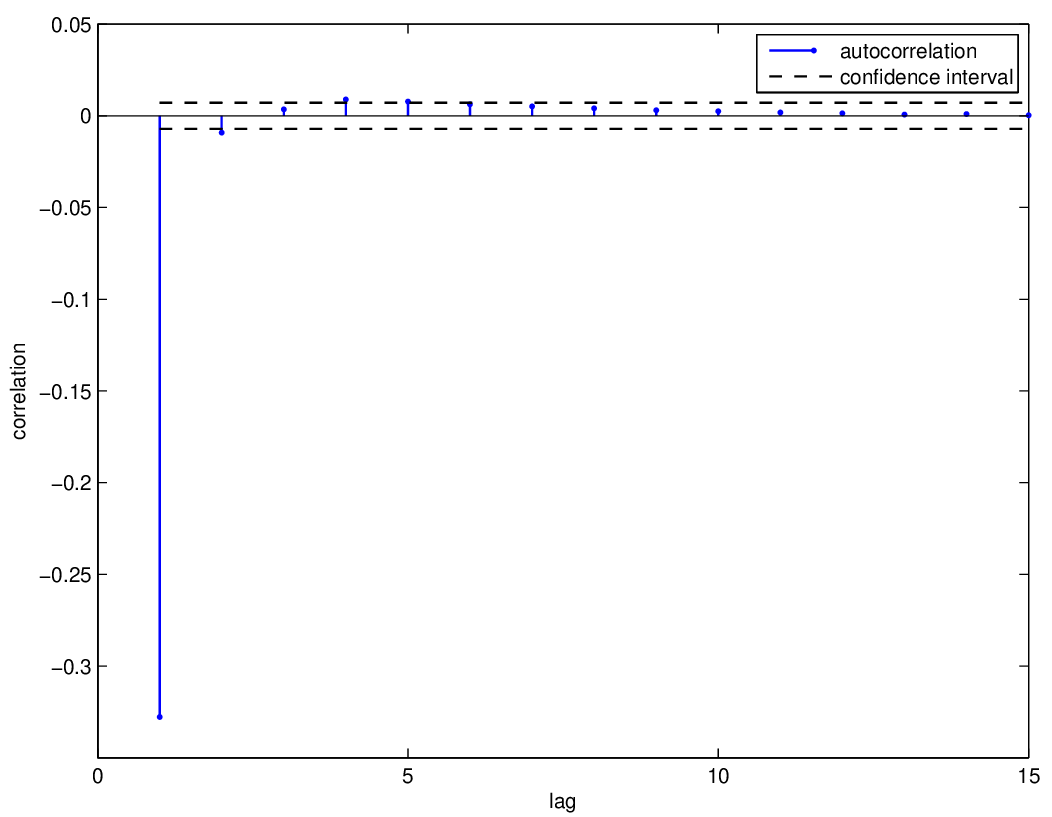} &
\includegraphics[height = 6cm, width = 0.45\textwidth]{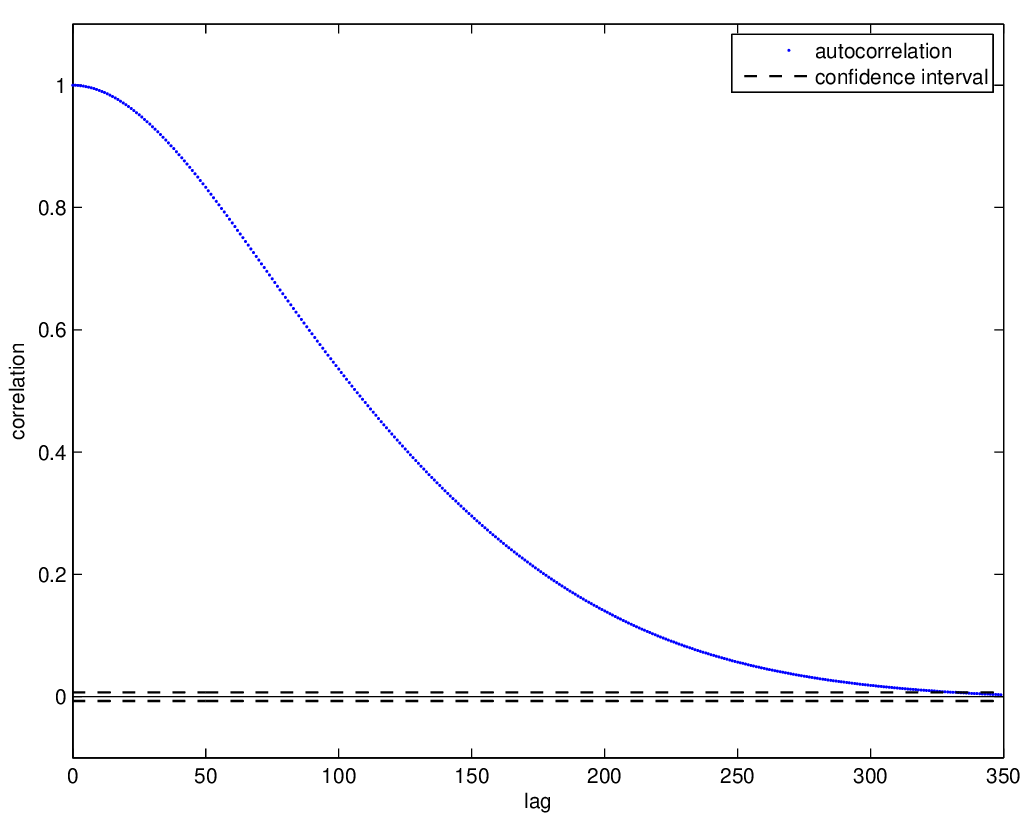}\\
\end{tabular}
\begin{scriptsize}
\parbox{0.925\textwidth}{\emph{Note}. We compute the empirical
autocorrelation function (acf) of the SPY returns. Panel A is for
the noisy returns (defined in Eq. \eqref{Eqn:Y}), while Panel B is for
the pre-averaged returns (defined in Eq. \eqref{Eqn:PreavgY}). The
acf is estimated daily and then averaged over time. The sample period
covers January, 2007 through March, 2011. The dashed line represents
a 95\% confidence interval for assessing the white noise null hypothesis.}
\end{scriptsize}
\end{center}
\end{figure}
In Panel A of Figure \ref{Fig:acf}, we plot the sample autocorrelation
function (acf) of the noisy return series, $\Delta_{i}^{n} Y$, up to lag 15.
There is a pronounced, significantly negative first-order autocorrelation of
about -0.35, which is consistent with a bid-ask bounce interpretation of
microstructure noise. The acf then increases and turns positive
at lag three. The fourth and fifth autocorrelation actually fall
outside the 95\% confidence bands based on a white noise
null hypothesis. Together with the subsequent monotonic decay of the
acf, this indicates that
noise operating at the tick-level is not i.i.d., as consistent with prior
work \citep*[e.g.,][]{hansen-lunde:06b}. We therefore
proceed using a pre-averaging window based on $\theta = 1$, which should
be a robust choice in light of the empirical evidence. The
acf of the corresponding pre-averaged returns, $\Delta \bar{Y}_{i}^{n}$,
is presented in Panel B of the figure. As expected,
there is a strong dependence in this series up to lag $k_{n}$.

\begin{figure}[!ht]
\caption{Time series of integrated variance estimates
and standard error.\label{Fig:ts}}
\begin{center}
\begin{tabular}{cc}
\footnotesize{Panel A: Integrated variance} &
\footnotesize{Panel B: Standard error}\\
\includegraphics[height = 6cm, width = 0.45\textwidth]{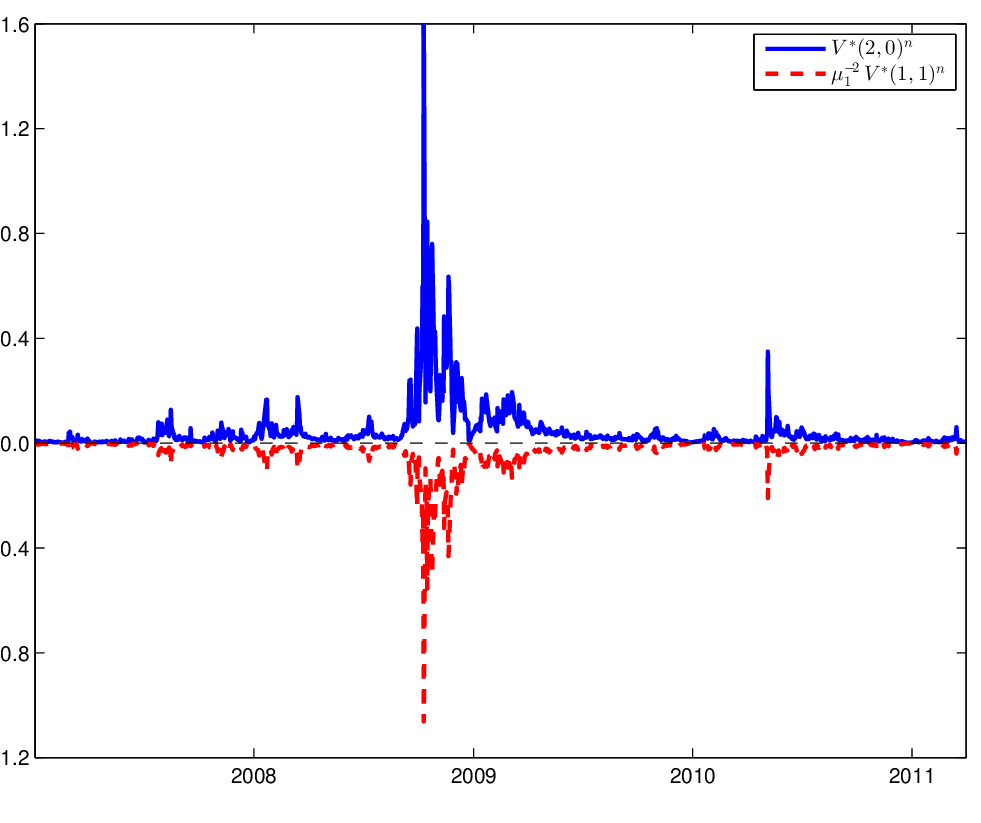} &
\includegraphics[height = 6cm, width = 0.45\textwidth]{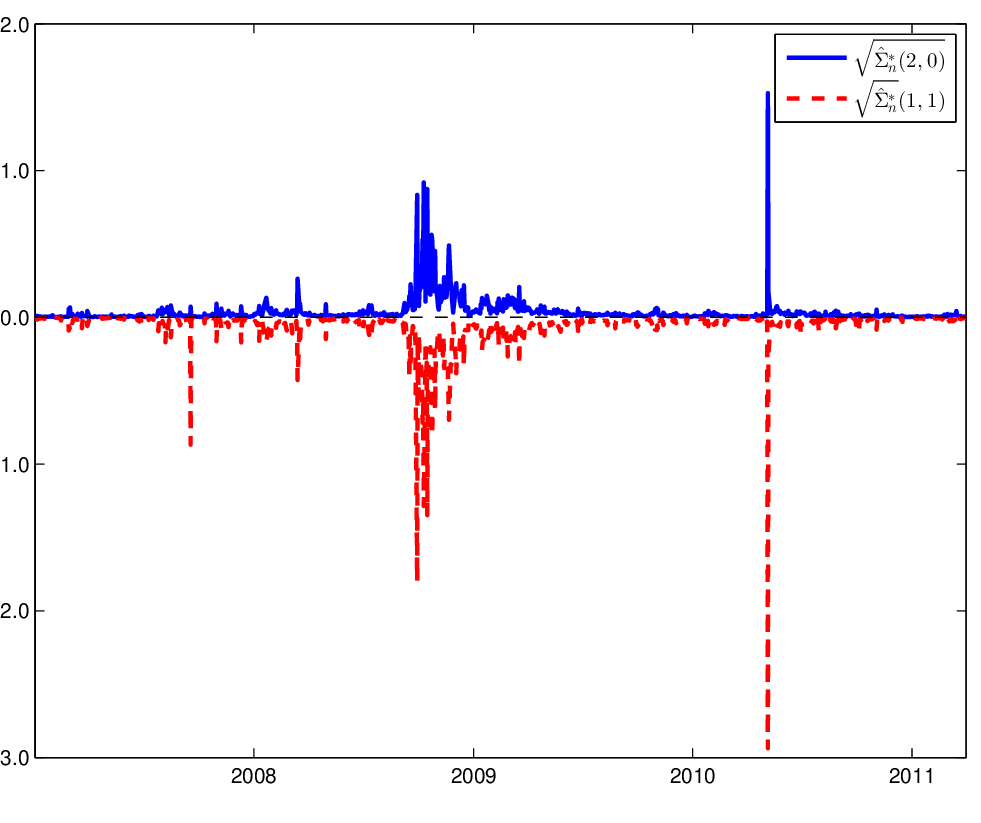}
\end{tabular}
\begin{scriptsize}
\parbox{0.925\textwidth}{\emph{Note}. In Panel A, we report the time series
of the daily $V^{*}(2,0)^{n}$ and $\mu_{1}^{-2} V^{*}(1,1)^{n}$ estimates.
The series were transformed into measures of the daily integrated variance,
as detailed in Eq. \eqref{Eqn:ivhat}. In Panel B, we plot the associated
standard errors, based on $\sqrt{\hat{ \Sigma}_{n}^{*}(2,0)}$ and $\sqrt{\hat{
\Sigma}_{n}^{*}(1,1)}$. We use $p = 10$ and $L = 15$ to implement the
subsampler. The sample period covers January, 2007 through March, 2011.
To facilitate the readability of the figure, the series based on $\mu_{1}^{-2}
V^{*}(1,1)^{n}$ and $\sqrt{\hat{ \Sigma}_{n}^{*}(1,1)}$ are reflected in the
$x$-axis.}
\end{scriptsize}
\end{center}
\end{figure}
We report the resulting time series of $V^{*}(2,0)^{n}$ and $\mu_{1}^{-2}
V^{*}(1,1)^{n}$ in Panel A of Figure \ref{Fig:ts}. Note that the graph for
$\mu_{1}^{-2} V^{*}(1,1)^{n}$ has been reflected in the $x$-axis. The
statistics are first computed day-by-day across the whole sample and
subsequently updated using Eq. \eqref{Eqn:plimPBV} to provide annualized
measures of the integrated variance (assuming 250 trading days p.a., on
average), i.e.:
\begin{equation}
\label{Eqn:ivhat}
\widehat{\text{IV}} = \frac{V^{*}(2,0)^{n}}{ \theta \psi_{2}^{k_{n}}} -
\frac{\psi_{1}^{k_{n}} \hat{ \omega}^{2}}{ \theta \psi_{2}^{k_{n}}}
\overset{p}{ \to} \int_{0}^{1} \sigma_{s}^{2} \text{d}s,
\end{equation}
with an identical transformation of $\mu_{1}^{-2} V^{*}(1,1)^{n}$.

The term $\displaystyle \frac{\psi_{1}^{k_{n}} \hat{ \omega}^{2}}{\theta
\psi_{2}^{k_{n}}}$
is a small bias
correction that compensates the pre-averaged bipower variation for the
residual effect of microstructure noise.\footnote{The bias correction
in Eq. \eqref{Eqn:ivhat} is only correct, when the noise is i.i.d.
Meanwhile, the estimator of $\omega^{2}$ we propose in Eq.
\eqref{equation:omega2_hat} is robust to the presence of a heteroscedastic
noise process, but it is generally not consistent for $\rho^{2}$ from
Section \ref{section:general_noise}, if the noise is
autocorrelated. As the current application is merely illustrative, we
ignore that issue here.}
$\hat{ \omega}^{2}$ is an estimate of the noise variance,
$\omega^{2}$. There are several estimators, which can serve the role of
$\hat{ \omega}^{2}$ \citep*[see, e.g.,][]{gatheral-oomen:10a}.
Among these, we adopt the one from \citet*{oomen:06a}, which relies on
the first-order autocorrelation of the noisy returns:
\begin{equation}
\label{equation:omega2_hat}
\hat{ \omega}^{2} = - \frac{1}{n - 1} \sum_{i = 1}^{n-1} \Delta_{i}^{n}Y
\Delta_{i+1}^{n}Y \overset{p}{ \to} \omega^{2}.
\end{equation}
We find a high degree of time-variation and persistence in the $\widehat{\text{IV}}$ series. The onset of the financial crisis and---in particular---the unprecedented volatility surrounding the collapse of Lehman Brothers in 2008 stands out visibly. To attach a measure of uncertainty to these estimates, Panel B charts the associated standard error estimate, $\sqrt{ \hat{ \Sigma}_{11,n}^{*}} / \theta \psi_{2}^{k_{n}}$ and $\sqrt{ \hat{ \Sigma}_{22,n}^{*}} / \theta \psi_{2}^{k_{n}}$, where the latter are based on the subsampler with $L = 15$ and $p = 10$. As expected, high levels of volatility spill over into the standard errors and tend to decrease estimation accuracy. The apparent outliers showing up in the standard error series in Q3, 2007 and Q2, 2010 correspond to single days with unusual market activity. The first is September 18, 2007, where the Federal Open Market Committee (FOMC) announced an unexpected reduction of its target for the federal funds rate by 50 basis points, while the second is May 6, 2010; the day of the S\&P 500 Flash Crash.

\begin{figure}[!ht]
\caption{Inference about $\ln \bigl( V^{*}(2,0)^{n} \bigr) - \ln \bigl(
\mu_{1}^{-2} V^{*}(1,1)^{n} \bigr)$.
\label{Fig:JV}}
\begin{center}
\begin{tabular}{cc}
\footnotesize{Panel A: Point estimate} & \footnotesize{Panel B: Confidence interval}\\
\includegraphics[height = 6cm, width = 0.45\textwidth]{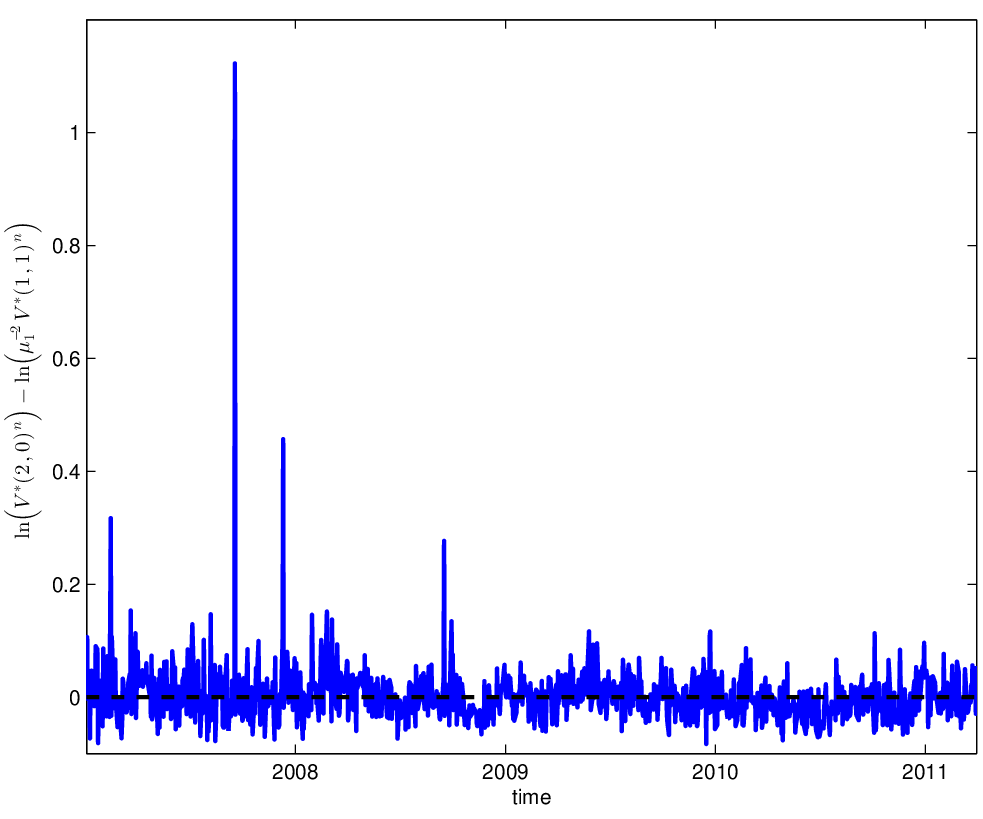} &
\includegraphics[height = 6cm, width = 0.45\textwidth]{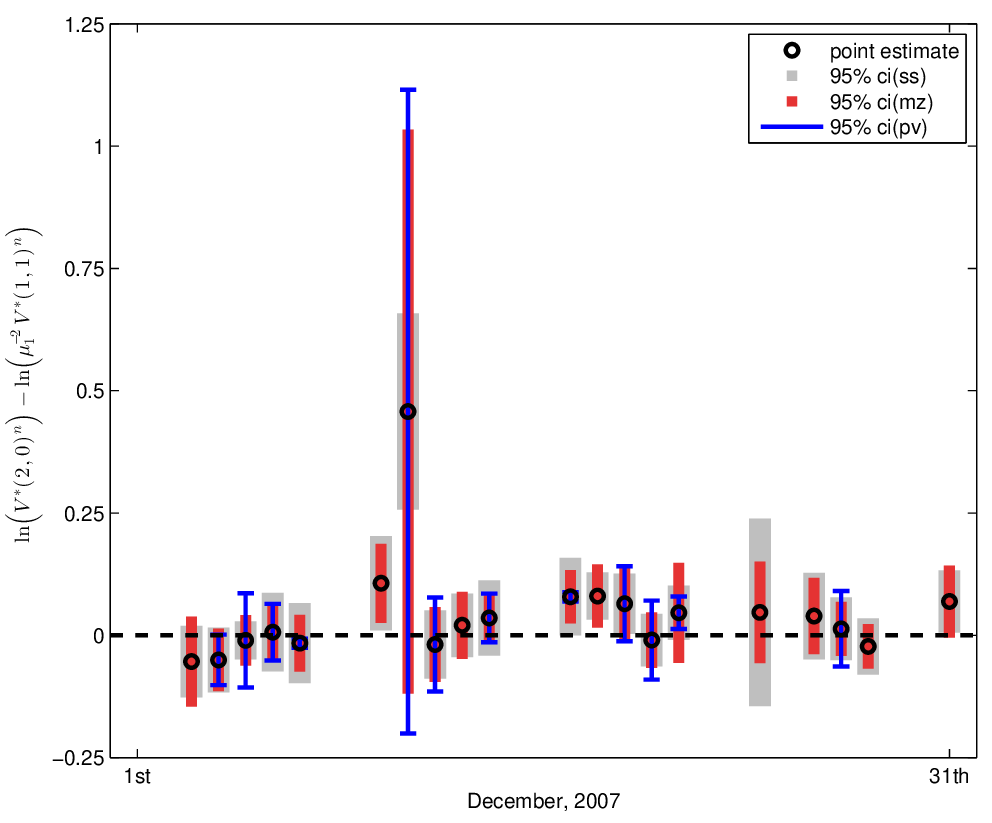}
\end{tabular}
\begin{scriptsize}
\parbox{0.925\textwidth}{\emph{Note}. We compute the difference $\ln \bigl(
V^{*}(2,0)^{n} \bigr) - \ln \bigl( \mu_{1}^{-2} V^{*}(1,1)^{n} \bigr)$, which
shows potential violations of the assumed continuous sample path model.
$V^{*}(2,0)^{n}$ and $V^{*}(1,1)^{n}$ are defined in Eq.
\eqref{Eqn:PreavgBV}. In Panel A, we plot the time series of the daily estimates
of this number across the sample, which covers January, 2007 through March, 2011.
In Panel B, we add a two-sided 95\% confidence interval for the log-difference
during the month
of December, 2007. The standard errors are found by applying the delta method to
the joint asymptotic distribution in Eq. \eqref{Eqn:CLTnoise}. We replace the
asymptotic covariance matrix by the
subsampler (wide, grey box), $\hat{ \Sigma}_{n}^{*}$,
the \citet*{mykland-zhang:17a} observed asymptotic variance computed as described
in the simulation section (narrow, red box), and the estimator proposed in
\citet*{podolskij-vetter:09a} (blue whisker), $\tilde{ \Sigma}_{n}^{*}$. The former
is implemented by setting
$p = 10$ and $L = 15$. In both panels, the dashed line
represents the limiting value in a pure diffusion model.}
\end{scriptsize}
\end{center}
\end{figure}

Turn next to Figure \ref{Fig:JV}, where we conduct inference about
$V^{*}(2,0)^{n} - \mu_{1}^{-2} V^{*}(1,1)^{n}$. In Panel A, we
compute the difference in the logarithms of these numbers, i.e. $\ln \bigl(
V^{*}(2,0)^{n} \bigr) - \ln \bigl( \mu_{1}^{-2} V^{*}(1,1)^{n} \bigr)$, which
tends to be less volatile compared to the raw statistic. As shown,
the majority of the point estimates hover around zero, which
is the theoretical limit in diffusion models. There are some notable
exceptions though, and in Panel B we examine one of these by zooming in on
the month of December, 2007. Alongside the statistic, we here report a
two-sided 95\% confidence interval. Standard
errors were found by applying the delta method (for the function $f(x,y) =
\ln(y) - \ln(\mu_{1}^{-2}x)$) to the joint asymptotic distribution in Eq.
\eqref{Eqn:CLTnoise} and then replacing the asymptotic variance
of the difference by a feasible estimate. In particular, we compare a
set of intervals based on the subsampler, $\hat{ \Sigma}_{n}^{*}$,
with those computed from the observed AVAR of
\citet*{mykland-zhang:17a}, which is again computed as explained in the
simulation section, and to the \citet*{podolskij-vetter:09a} estimator,
$\tilde{ \Sigma}_{n}^{*}$. If the latter leads to a negative variance
estimate, it is excluded. As consistent with
Table \ref{Table:Proportion}, this is a recurrent problem. Moreover,
if all three estimates are well-defined, they are often
closely aligned, but both the subsampler and observed AVAR
appear less erratic, while
$\tilde{ \Sigma}_{n}^{*}$ is often very narrow or wide. This is most
visible from the big discrepancy on December 11, 2007, marking a day
with yet another rate cut by the Fed. On this day, the condition number
of $\tilde{ \Sigma}_{n}^{*}$ is $\text{cond}( \tilde{ \Sigma}_{n}^{*}) =
452.36$, which suggest that the underlying covariance matrix estimate is
very fragile. The corresponding figure for the observed asymptotic variance
is $64.68$, which is again rather high, and indeed it also leads to a very
large confidence interval here. Meanwhile, the condition number of the
subsampler is more modest at
$\text{cond}( \tilde{ \Sigma}_{n}^{*}) = 15.16$, and it generally appears
to be the most stable over time.

To end the paper, we provide an alternative application, where
the subsampler is used to draw inference about
the amount of heteroscedasticity in noisy high-frequency data. To do
this, we start by computing the statistics $V^{*}(2,0)^{n}$ and
$V^{*}(4,0)^{n}$, i.e. the pre-averaged bipower variation based on the
parameter $q = (4,2)'$ (and $r = (0,0)'$). Taking these as input, we
appeal to Eq. \eqref{Eqn:plimPBV} by forming the estimate:
\begin{equation}
\label{Eqn:iqhat}
\widehat{\text{IQ}} = \frac{ \mu_{4}^{-1} V^{*}(4,0)^{n}}{ (\theta
\psi_{2}^{k_{n}})^{2}} - \frac{2 \psi_{2}^{k_{n}}
\psi_{1}^{k_{n}} \hat{ \omega}^{2}}{ (\theta
\psi_{2}^{k_{n}})^{2}}
\widehat{\text{IV}} - \frac{ (\psi_{1}^{k_{n}} \hat{
\omega}^{2})^{2}}{(\theta^{2}
\psi_{2}^{k_{n}})^{2}}
\overset{p}{ \to} \int_{0}^{1} \sigma_{s}^{4} \text{d}s,
\end{equation}
which converges to the so-called integrated quarticity. We then exploit
that $\sqrt{
\widehat{\text{IQ}}} / \widehat{\text{IV}} \overset{p}{ \to} \sqrt{
\int_{0}^{1} \sigma_{s}^{4} \text{d}s} / \int_{0}^{1} \sigma_{s}^{2}
\text{d}s \geq 1$, with equality if and only if $\sigma$ is constant.
Thus, an estimated ratio far above one suggests there is
significant variation in volatility within the day, while a ratio close to
one means $\sigma$ can be regarded, as if it was approximately
constant. This type of statistic has been exploited in earlier work to
test for the parametric form of volatility
\citep*[e.g.,][]{dette-podolskij-vetter:06a,vetter-dette:12a}, and it
is also finds use in the jump-testing literature
\citep*[e.g.,][]{barndorff-nielsen-shephard:06a,kolokolov-reno:16a}.

\begin{figure}[!ht]
\caption{Inference about $\ln \Big( \sqrt{ \widehat{ \text{IQ}}} / \widehat{
\text{IV}} \Big)$.
\label{Fig:constsigma}}
\begin{center}
\begin{tabular}{cc}
\footnotesize{Panel A: Point estimate} & \footnotesize{Panel B: Confidence interval}\\
\includegraphics[height = 6cm, width = 0.45\textwidth]{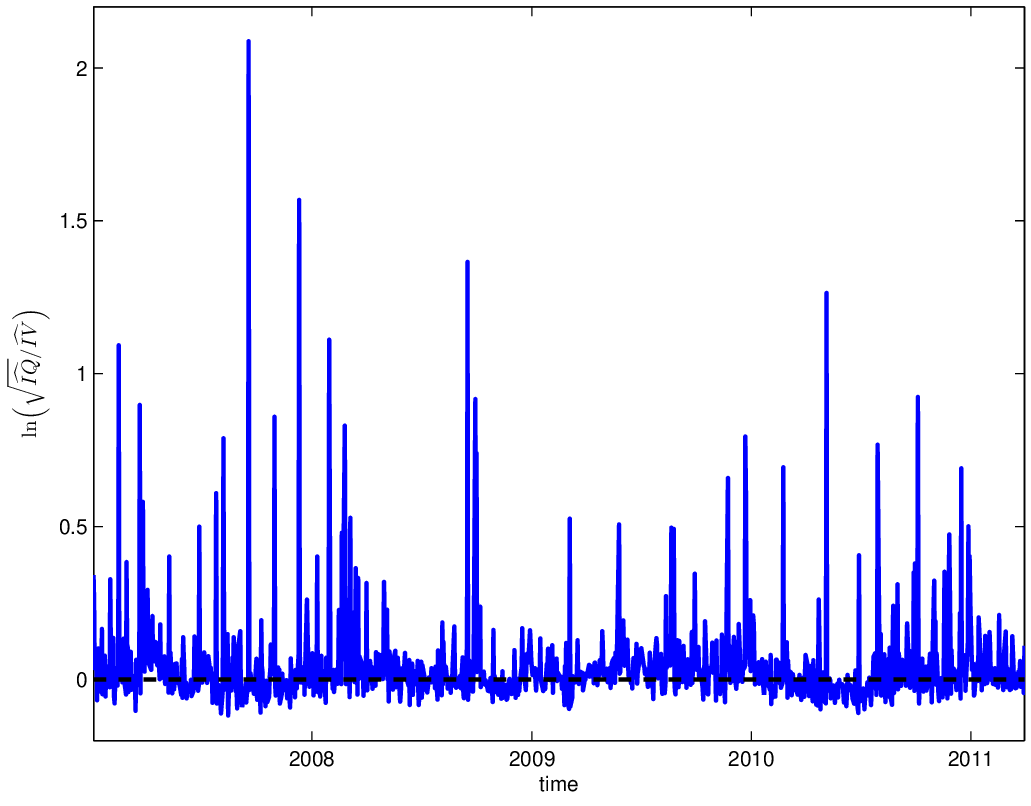} &
\includegraphics[height = 6cm, width = 0.45\textwidth]{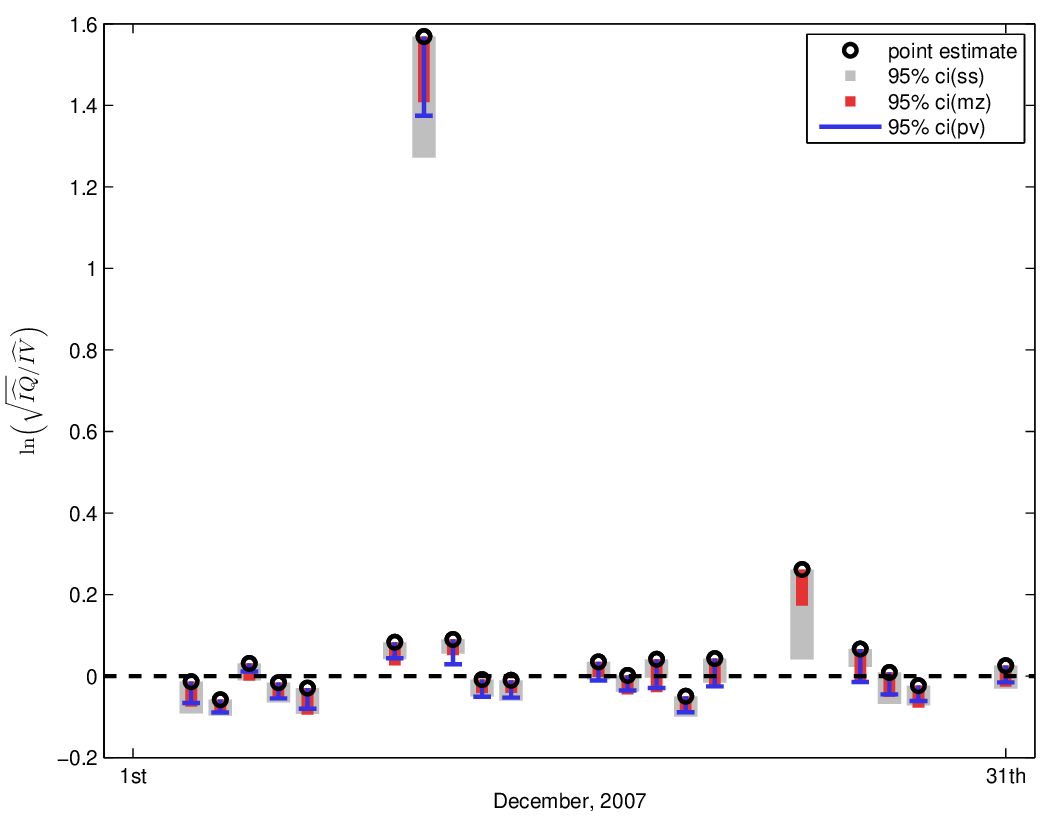}
\end{tabular}
\begin{scriptsize}
\parbox{0.925\textwidth}{\emph{Note}. We compute the log-ratio $\ln \Big( \sqrt{
\widehat{\text{IQ}}} / \widehat{ \text{IV}} \Big)$, which measures the degree of
heteroscedasticity in $\sigma$ within the day. $\widehat{ \text{IV}}$ and
$\widehat{ \text{IQ}}$ are defined in Eqs.
\eqref{Eqn:ivhat} and \eqref{Eqn:iqhat}. In Panel A, we plot the time series of the
daily estimates of this number across the sample, which covers January, 2007 through
March, 2011. In Panel B, we add a left-sided 95\% confidence interval for the log-ratio
during the month of December, 2007 (the upper end of the interval extending to
$+\infty$ is not shown).
The standard errors are found by applying the delta method to
the joint asymptotic distribution in Eq. \eqref{Eqn:CLTnoise}.
We replace the asymptotic covariance matrix by the
subsampler (wide, grey box), $\hat{ \Sigma}_{n}^{*}$,
the \citet*{mykland-zhang:17a} observed asymptotic variance computed as described
in the simulation section (narrow, red box), and the estimator proposed in
\citet*{podolskij-vetter:09a} (blue whisker), $\tilde{ \Sigma}_{n}^{*}$.
The former is implemented by setting
$p = 10$ and $L = 15$. In both panels, the dashed line
represents the limiting value in a constant volatility model.}
\end{scriptsize}
\end{center}
\end{figure}
The outcome of this exercise is collected in Figure \ref{Fig:constsigma}.
In Panel A, we plot the time series of the estimated log-ratio, i.e.
$\ln \Bigl( \sqrt{ \widehat{\text{IQ}}} / \widehat{ \text{IV}} \Bigr)$.
We again use a log-transformation in order to improve the scaling of the
results and facilitate interpretation of the graphs. The
log-ratio should cluster around zero, if volatility is constant. We
observe an extreme degree of fluctuation in this statistic over time,
and, as anticipated, there are many days, where the log-ratio
is large. Still, we also find a decent portion of estimates, which are
close to zero. Small negative numbers can be observed
as a result of sampling variation. In Panel B, we complement the analysis
by looking at the month of December, 2007. We add a left-sided 95\%
confidence interval for $\ln \Bigl( \sqrt{ \widehat{\text{IQ}}} /
\widehat{ \text{IV}} \Bigr)$, where standard errors are again
retrieved via the delta method and three estimates of the asymptotic
covariance matrix. The interpretation is that on some days, such as
the day of the FOMC meeting, volatility is changing a lot, while on others
it is not moving much, which is consistent with the findings of
\citet*[][Figure 1]{mykland-zhang:17a}. Of course, the latter finding can
also arise, if the high-frequency data is not informative
enough to discriminate random sampling errors from genuine parameter
variation in $\sigma$, which could be difficult in
times of severe stress in financial markets. In this respect, it is important
to acknowledge the limitations of the subsampler, which, albeit
consistent, is itself subject to a substantial degree of sampling uncertainty
in practice.

\section{Conclusion}

In this paper, we propose a subsample estimator of the asymptotic conditional
covariance matrix of bipower variation. The theory is developed for
diffusion models both with and without microstructure noise. We show our
estimator is consistent and, under suitable conditions, we find an error
decomposition of our statistic, from which we are able to derive its rate of
convergence. To complement the theory, we conduct a Monte Carlo study, which
documents how the subsampler can be used to draw feasible
inference about pre-averaged bipower variation in the presence of noise.
The results are compelling and show that the subsampler delivers accurate
inference and that it is robust to the choice of its tuning parameters,
i.e. the number of subsamples and the block length. We concluded with
an empirical analysis that provides an illustration of a few of the
directions, in which the subsampler can be applied to real
high-frequency data.

In future research, several extensions of the current paper are possible.
First, as we noted here, we can use the subsampler as an ingredient to test for
the presence of jumps in noisy high-frequency data. This requires the
subsampler to be robust against such jumps, while also being resistant to
the influence of microstructure noise, which the current implementation is
not. This is a line of research that we are currently looking into ourselves in
a companion paper. Second, there is some work to be done on the optimal
selection of tuning parameters. Finally, as we stressed in the main text,
pre-averaging can not only be used to consistently estimate integrated variance
and quarticity, but also more general functionals of volatility. Meanwhile,
subsampling facilitates estimation of the asymptotic
covariance matrix of such pre-averaged high-frequency statistics under mild
conditions on the noise process. The combination is therefore potent, and we
envision that a lot of empirical papers can benefit from and
follow in the aftermath of our work.

\clearpage
% LIST OF REFERENCES
\renewcommand{\baselinestretch}{1.0}
\small
\bibliography{userref}

\pagebreak

\appendix

\renewcommand{\baselinestretch}{1.6}
\normalsize

\section{Appendix of proofs}

In this appendix, we prove the theoretical results that are found in the main
text. Throughout the proofs, all positive constants are denoted by
$C$ (or $C_{p}$, if it depends on a parameter $p$), although they may change from
line to line.  It follows from a standard localization procedure that we can
assume that $a, \sigma, \tilde{a}, \tilde{ \sigma}$ and $\tilde{v}$ are bounded,
e.g., BGJPS6. Furthermore, under Assumption ($\mathbf{V}$), we also assume without
loss of generality that
\begin{align}
\label{Eqn:sigmanew}
\sigma_{t} = \sigma_{0} + \int_{0}^{t} \tilde{a}_{s}\text{d}s + \int_{0}^{t}
\tilde{ \sigma}_{s}\text{d}W_{s} + \int_{0}^{t} \tilde{v}_{s}\text{d}B_{s} +
\int_{0}^{t} \int_{E} \tilde{ \delta} (s,x) ( \tilde{ \mu} - \tilde{ \nu})
( \text{d}s, \text{d}x),
\end{align}
and it holds that
\begin{equation*}
\sup_{\omega \in \Omega, s \geq 0} | \tilde{ \delta}( \omega,s,x)| \leq
\tilde{ \psi}(x) \leq C \qquad \text{and}
\qquad \int_{E}  \tilde{ \psi}^2(x) \tilde{F}( \text{d}x) <\infty.
\end{equation*}
By Assumption ($\mathbf{J}$), we can also apply this localization to the jump
component of $X$. Moreover, due to the polarization identity, we can (and shall)
assume throughout that $m = 1$, so that all
statistics are 1-dimensional.

\subsection{Proof of Theorem \ref{c0}}
We start by introducing some notation. We define:
\begin{equation}
\label{c2}
\alpha_{i}^{n} = \sqrt{n} \sigma_{ \frac{i-1}{n}} \Delta_{i}^{n}W \qquad \text{and} \qquad
\chi_{i}^{n} = f( \alpha_{i}^{n}) - \mathbb{E} \left[ f\left(\alpha_{i}^{n}
\right) \mid \mathcal{F}_{ \frac{i-1}{n}} \right].
\end{equation}
We note that $\alpha_{i}^{n}$ is a first-order approximation of $\sqrt{n} \Delta_{i}^{n} X$.
%Also, we set
%\begin{equation*}
%\chi_{i}^{n} = f( \alpha_{i}^{n}) - \mathbb{E} \left[ f\left(\alpha_{i}^{n}
%\right) \mid \mathcal{F}_{ \frac{i-1}{n}} \right].
%\end{equation*}
The next lemma is shown in BGJPS6.

\begin{lemma}
\label{Lemma:Estimate}
Assume that $p \geq 2$ and let $h$ be any function of polynomial growth. Then,
\begin{equation}
\label{c3}
\mathbb{E} \bigl[ | \alpha_{i}^{n} |^{p} \bigr] + \mathbb{E} \bigl[ | \sqrt{n}
\Delta_{i}^{n} X |^{p} \bigr] + \mathbb{E} \bigl[ | h( \alpha_{i}^{n}) |^{p}
\bigr] + \mathbb{E} \bigl[ | \chi_{i}^{n}|^{p} \bigr] \leq C_{p},
\end{equation}
and, if $\sigma$ is continuous,
\begin{equation}
\label{c5}
\mathbb{E} \bigl[ | \sqrt{n} \Delta_{i}^{n} X - \alpha_{i}^{n} |^{p} \bigr]
\leq C_{p} n^{-p/2}.
\end{equation}
\end{lemma}
In the proofs, we make use of Burkholder's inequality several times. This
inequality basically says that for any continuous process $Z$ of the form in Eq.
\eqref{Eqn:X} and for any $p \geq 2$:
\begin{equation}
\label{Eqn:Burkholder}
\mathbb{E} \bigl[ |Z_{t} - Z_{s}|^{p} \bigr] \leq C_{p} |t - s|^{p/2}.
\end{equation}
Following the comments above, the definition of $\hat{ \Sigma}_{n}$
here collapses to:
\begin{equation*}
\hat{ \Sigma}_{n} = \frac{1}{L} \sum_{l = 1}^{L} \biggl( \sqrt{
\frac{n}{L}} \Bigl( V_{l}(f)^{n} - V(f)^{n} \Bigr) \biggr)^{2} ,
\end{equation*}
To estimate $\hat{ \Sigma}_{n} - \Sigma$, we make the following approximations:
\begin{align*}
\Sigma_{n} &= \frac{1}{L} \sum_{l = 1}^{L} \biggl( \sqrt{\frac{n}{L}} \Bigl(
V_{l}(f)^{n} - V(f) \Bigr) \biggr)^{2}, \qquad
Q_{n} = \frac{1}{n} \sum_{l = 1}^{L} \Biggl( \sum_{i = 1}^{n/L} \chi_{(i - 1)L
+ l}^{n} \Biggr)^{2}, \\[0.25cm]
U_{n} &= \frac{1}{n} \sum_{l = 1}^{L} \sum_{i = 1}^{n/L} \bigl( \chi_{(i -
1)L + l}^{n} \bigr)^{2}, \qquad
R_{n} = \frac{1}{n} \sum_{l = 1}^{L} \sum_{i = 1}^{n/L} \mathbb{E} \left[
\bigl( \chi_{(i - 1)L + l}^{n} \bigr)^{2} \mid \mathcal{F}_{ \frac{(i - 1)L
+ l - 1}{n}} \right].
\end{align*}
The goal is then to find an optimal upper bound on the error entailed by these
approximations. In particular, the proof is completed, if we can show that
under the conditions of Theorem \ref{c0}:
\begin{align*}
\mathbb{E} \bigl[ | \Sigma_{n} - Q_{n} | \bigr] &\leq C \biggl( \frac{L}{n} +
\frac{1}{ \sqrt{n}} \biggr), \tag{i} \\[0.25cm]
\mathbb{E} \bigl[ | Q_{n} - U_{n} | \bigr] &\leq \frac{C}{\sqrt{L}}, \tag{ii}
\\[0.25cm]
\mathbb{E} \bigl[ | U_{n} - R_{n} | \bigr] &\leq \frac{C}{\sqrt{n}}, \tag{iii}
\\[0.25cm]
\mathbb{E} \bigl[ | R_{n} - \Sigma | \bigr] &\leq \frac{C}{n}, \tag{iv}
\\[0.25cm]
\mathbb{E} \bigl[ |\hat{ \Sigma}_{n} - \Sigma_{n} | \bigr] &\leq \frac{C}{L}.
\tag{v}
\end{align*}
We prove these estimates in the order (iii), (v), (ii), (iv), and (i); the last
step being the hardest.

\begin{proof} [Proof of (iii)]
This part is almost trivial. We note that
\begin{equation*}
U_{n} - R_{n} = \frac{1}{n} \sum_{i = 1}^{n} \bigg( \bigl( \chi_{i}^{n} \bigr)^{2} -
\mathbb{E} \left[ \bigl( \chi_{i}^{n} \bigr)^{2} \mid \mathcal{F}_{ \frac{i -
1}{n}} \right] \bigg),
\end{equation*}
which is a sum of martingale differences. Lemma \ref{Lemma:Estimate} implies
that $\mathbb{E} \Bigl[ \bigl( \chi_{i}^{n} \bigr)^{4} \Bigr] \leq C$. Then,
we are done, as
\begin{equation}
\label{cc5}
\mathbb{E} \bigl[ | U_{n} - R_{n} |^{2} \bigr] \leq \frac{C}{n}.
\end{equation} \qed
\end{proof}

\begin{proof}[Proof of (v)]
First, the identity
\begin{equation*}
\sum_{l = 1}^{L} V_{l}(f)^{n} = L V(f)^{n},
\end{equation*}
together with some simple algebra imply that
\begin{equation*}
\hat{\Sigma}_{n} - \Sigma_{n} = - \frac{n}{L} \bigl( V(f)^{n} -
V(f) \bigr)^{2}.
\end{equation*}
%Note that this implies that in the frictionless setting, the conditional
%expectation of $\hat{\Sigma}_{n}$ is smaller than that of the infeasible
%statistic $\Sigma_{n}$ (which should be closer to $\Sigma$), so the
%estimator has a slight downward bias, which tends to dissipate as $L$
%increases.
Then, we finish the proof of (v) by using that
\begin{equation*}
\mathbb{E} \Bigl[ \bigl( V(f)^{n} - V(f) \bigr)^{2} \Bigr] \leq \frac{C}{n},
\end{equation*}
which has already been established in prior work cited above. \qed
\end{proof}

\begin{proof}[Proof of (ii)]
We define $S_{l}^{m} = \sum_{i = 1}^{m} \chi_{(i - 1)L + l}^{n}$ and
$T_{l}^{m} = \sum_{i = 1}^{m} \bigl(\chi_{(i - 1)L + l}^{n} \bigr)^2$ and observe that
\begin{align*}
Q_{n} - U_{n} = \frac{1}{n} \sum_{l = 1}^{L} A_{l}^{n},
\end{align*}
where, for each $l$, we use the notation
\begin{align*}
A_{l}^{n}=\left( S_{l}^{n/L} \right)^{2} - T_{l}^{n/L} = \sum_{i,j=1, i \neq
j}^{n/L} \chi_{(i-1)L+l}^{n} \chi_{(j-1)L+l}^{n}.
\end{align*}
As $A_{l_{1}}^n$ and $A_{l_{2}}^n$ are uncorrelated for every $l_{1} \neq
l_{2}$, we find that
\begin{align}
\begin{split}
\label{c14}
\mathbb{E} \bigl[ (Q_{n} - U_{n})^{2} \bigr] &= \frac{1}{n^{2}}
\sum_{l = 1}^{L} \mathbb{E} \bigl[ (A_{l}^{n})^{2} \bigr] \\[0.25cm]
&\leq \frac{C}{n^{2}} \sum_{l = 1}^{L} \left( \mathbb{E} \bigl[
(S_{l}^{n/L})^{4} \bigr] + \mathbb{E} \bigl[ (T_{l}^{n/L})^{2} \bigr]
\right).
\end{split}
\end{align}
We note that $(S_{l}^{m})_{m = 1}^{n/L}$ is a discrete martingale for each
fixed $l$. Then, the discrete Burkholder and Cauchy-Schwarz inequalities
together with Lemma \ref{Lemma:Estimate} imply that
\begin{equation}
\label{cc16}
\mathbb{E} \Bigl[ \bigl(S_{l}^{n/L} \bigr)^{4} \Bigr] \leq C \mathbb{E} \Biggl[
\biggl( \sum_{i=1}^{n/L} \bigl( \chi_{(i - 1)L + l}^{n} \bigr)^{2} \biggr)^{2}
\Biggr] \leq C \left( \frac{n}{L} \right)^{2}.
\end{equation}
We conclude the proof by using Lemma \ref{Lemma:Estimate} once again, but for
the $T_{l}^{n/L}$ term. This yields:
\begin{equation*}
\mathbb{E} \bigl[ (Q_{n} - U_{n})^{2} \bigr] \leq \frac{C}{L}.
\end{equation*} \qed
\end{proof}

\begin{proof}[Proof of (iv)]
We start by noting that
\begin{equation*}
\mathbb{E} \left[ \bigl( \chi_{(i - 1)L + l}^{n})^{2} \mid \mathcal{F}_{
\frac{(i - 1)L + l - 1}{n}} \right] = \rho_{ \sigma_{ \frac{(i - 1)L + l -
1}{n}}}(f^{2}) - \rho_{ \sigma_{ \frac{(i - 1)L + l - 1}{n}}}^{2}(f),
\end{equation*}
so $R_{n}$ is a Riemann approximation of $\Sigma$, because
\begin{equation*}
R_{n} = \frac{1}{n} \sum_{i = 1}^{n} \rho_{ \sigma_{ \frac{i - 1}{n}}}(f^{2})
- \rho^{2}_{ \sigma_{ \frac{i - 1}{n}}}(f).
\end{equation*}
In the remainder of the proof, we suppress the dependence on $f$ by defining the
function $\tau(x) = \rho_{x}(f^{2}) - \rho^{2}_{x}(f)$. Now, we define the mapping:
\begin{equation*}
\phi(x) \equiv \rho_{x}(f) = \int_{\mathbb{R}} \frac{1}{\sqrt{ 2 \pi x^{2}}}
\exp \left( -\frac{y^{2}}{2x^{2}} \right) f(y) \text{d}y.
\end{equation*}
We have $\phi \in C^{3}( \mathbb{R})$, as $f \in C^{3}(\mathbb{R}).$ Hence, it follows that
$\tau \in C^{3}(\mathbb{R})$.
Then, Taylor's theorem and the inequality in Eq. \eqref{Eqn:Burkholder} applied
to $\sigma$ mean that
\begin{equation*}
R_{n} - \Sigma = \sum_{i = 1}^{n} \int_{ \frac{i - 1}{n}}^{ \frac{i}{n}} \Bigl[
\tau_{ \sigma_{ \frac{i - 1}{n}}}- \tau_{ \sigma_{s}}\Bigr] \text{d}s
= \sum_{i = 1}^{n} \mu_{i}^{n}(1)+\sum_{i = 1}^{n} \mu_{i}^{n}(2) + O_{p}
\biggl( \frac{1}{n} \biggr),
\end{equation*}
where
\begin{align*}
\mu_{i}^{n}(1) &= -\tau' \bigl( \sigma_{ \frac{i - 1}{n}} \bigr) \int_{
\frac{i - 1}{n}}^{ \frac{i}{n}} \biggl( \int_{ \frac{i - 1}{n}}^{s}
\tilde{a}_{u} \text{d}u \biggr) \text{d}s \quad \text{and}
\\[1.5 ex]
\mu_{i}^{n}(2) &= -\tau' \bigl( \sigma_{ \frac{i - 1}{n}} \bigr) \int_{
\frac{i - 1}{n}}^{ \frac{i}{n}} \biggl( \int_{ \frac{i - 1}{n}}^{s}
\tilde{\sigma}_{u} \text{d}W_{u}+\int_{ \frac{i - 1}{n}}^{s}
\tilde{v}_{u} \text{d}B_{u}
 \biggr) \text{d}s.
\end{align*}
We find that:
\begin{equation}
\label{e2}
\mathbb{E} \bigl[ | \mu_{i}^{n}(1)|^{2} \bigr] \leq \frac{C}{n^{4}} \qquad
\text{and} \qquad \mathbb{E} \bigl[ | \mu_{i}^{n}(2)|^{2} \bigr] \leq
\frac{C}{n^{3}}.
\end{equation}
Using the martingale difference property, this implies that
\begin{equation}
\label{e3}
\mathbb{E} \biggl[ \bigl| \sum_{i = 1}^{n} \mu_{i}^{n}(2) \bigr|^{2} \biggr] =
\mathbb{E} \biggl[ \sum_{i = 1}^{n} | \mu_{i}^{n}(2)|^{2} \biggr] \leq
\frac{C}{n^{2}}.
\end{equation}
The result then follows from Eqs. \eqref{e2} -- \eqref{e3} and the
Cauchy-Schwarz inequality. \qed
\end{proof}
To prove Eq. (i), we need a bit of preparation. We let
\begin{equation*}
\tilde{V}_{l}(f)^{n} = \frac{1}{n/L} \sum_{i = 1}^{n/L} f( \alpha_{(i - 1)L +
l}^{n}) \qquad \text{and} \qquad \hat{V}_{l}(f)^{n} = \frac{1}{n/L} \sum_{i =
1}^{n/L} \mathbb{E} \Bigl[ f( \alpha_{(i - 1)L + l}^{n}) \mid \mathcal{F}_{
\frac{(i-1)L+l-1}{n}} \Bigr].
\end{equation*}
With the decomposition
\begin{equation*}
V_{l}(f)^{n} - V(f) = \Bigl( V_{l}(f)^{n} - \tilde{V}_{l}(f)^{n} \Bigr) +
\Bigl( \tilde{V}_{l}(f)^{n} - \hat{V}_{l}(f)^{n} \Bigr) + \Bigl(
\hat{V}_{l}(f)^{n} - V(f) \Bigr),
\end{equation*}
we get that
\begin{equation*}
\Sigma_{n} - Q_{n} = D_{n}^{(1)} + D_{n}^{(2)} + D_{n}^{(3)} + D_{n}^{(4)},
\end{equation*}
where
\begin{align*}
D_{n}^{(1)} &= \frac{2n}{L^{2}} \sum_{l = 1}^{L} \Bigl( V_{l}(f)^{n} -
\tilde{V}_{l}(f)^{n} \Bigr) \Bigl( \tilde{V}_{l}(f)^{n} - V(f) \Bigr),
\\[0.25cm]
D_{n}^{(2)} &= \frac{2n}{L^{2}} \sum_{l = 1}^{L} \Bigl( \hat{V}_{l}(f)^{n} -
V(f) \Bigr) \Bigl( \tilde{V}_{l}(f)^{n} - \hat{V}_{l}(f)^{n} \Bigr),
\\[0.25cm]
D_{n}^{(3)} &= \frac{n}{L^{2}} \sum_{l = 1}^{L} \Bigl( V_{l}(f)^{n} -
\tilde{V}_{l}(f)^{n} \Bigr)^{2}, \\[0.25cm]
D_{n}^{(4)} &= \frac{n}{L^{2}} \sum_{l = 1}^{L} \Bigl( \hat{V}_{l}(f)^{n} -
V(f) \Bigr)^{2}.
\end{align*}
To provide an estimate of these terms, we exploit the following preliminary
result.
\begin{lemma}
\label{l1}
Assume that the conditions of Theorem \ref{c0} are fulfilled. Then,
uniformly in $l$:
\begin{align}
\mathbb{E} \left[ | \tilde{V}_{l}(f)^{n} - \hat{V}_{l}(f)^{n}|^{2} \right]
&\leq C \frac{L}{n}, \tag{A.2a} \\[0.25cm]
\mathbb{E} \left[ | \hat{V}_{l}(f)^{n} - V(f)|^{2} \right] &\leq C
\frac{L^{2}}{n^{2}}, \tag{A.2b} \\[0.25cm]
\mathbb{E} \left[ |V_{l}(f)^{n} - \tilde{V}_{l}(f)^{n}|^{2} \right]
&\leq C \frac{L}{n^{2}}. \tag{A.2c}
\end{align}
\end{lemma}
\begin{proof}[Proof of Lemma \ref{l1}]
Part (A.2a) is shown by using the discrete Burkholder inequality as in Eq.
\eqref{cc16}. The proof of part (A.2b) follows along the lines of the proof
of Eq. (iv). To prove part (A.2c), we recall
condition $(\textbf{V})$ and make the decomposition
\begin{equation*}
\xi_{i}^{n} =\sqrt{n}\Delta_{i}^{n}X - \alpha_{i}^{n} = \sqrt{n} \left(
\int_{ \frac{i - 1}{n}}^{ \frac{i}{n}} a_{s} \text{d}s + \int_{ \frac{i -
1}{n}}^{ \frac{i}{n}} \left( \sigma_{s} - \sigma_{ \frac{i - 1}{n}} \right)
\text{d}W_{s} \right) \equiv \xi_{i}^{n}(1) + \xi_{i}^{n}(2),
\end{equation*}
where
\begin{align*}
\xi_{i}^{n}(1) &= \sqrt{n} \left( \frac{1}{n} a_{ \frac{i - 1}{n}} + \int_{
\frac{i - 1}{n}}^{ \frac{i}{n}} \left[ \tilde{ \sigma}_{ \frac{i - 1}{n}}
\left( W_{s} - W_{ \frac{i - 1}{n}} \right) + \tilde{v}_{ \frac{i - 1}{n}}
\left(B_{s} - B_{ \frac{i - 1}{n}} \right)
 \right] \text{d}W_{s} \right),
\\[0.25cm]
\xi_{i}^{n}(2) &= \sqrt{n} \Biggl( \int_{ \frac{i - 1}{n}}^{ \frac{i}{n}}
\left( a_{s} - a_{ \frac{i - 1}{n}} \right) \text{d}s + \int_{ \frac{i -
1}{n}}^{ \frac{i}{n}} \left[ \int_{ \frac{i - 1}{n}}^{s} \tilde{a}_{u}
\text{d}u \right] \text{d}W_{s} \\[0.25cm]
&+ \int_{ \frac{i - 1}{n}}^{ \frac{i}{n}} \left[ \int_{ \frac{i - 1}{n}}^{s}
\left( \tilde{ \sigma}_{u} - \tilde{ \sigma}_{ \frac{i - 1}{n}} \right)
\text{d}W_{u} + \int_{ \frac{i - 1}{n}}^{s} \left( \tilde{v}_{u} - \tilde{v}_{
\frac{i - 1}{n}} \right) \text{d}B_{u}
\right] \text{d}W_{s} \Biggr).
\end{align*}
By the assumptions of Theorem \ref{c0}, the Burkholder and Cauchy-Schwarz
inequalities yield that
\begin{align}
\label{d1}
\mathbb{E} \bigl[ | \xi_{i}^{n}(1) |^{4} \bigr] &\leq \frac{C}{n^{2}},
\\[0.25cm]
\label{d2}
\mathbb{E} \bigl[ | \xi_{i}^{n}(2) |^{4} \bigr] &\leq \frac{C}{n^{4}}.
\end{align}
Then, using Taylor's theorem, we may write $V_{l}(f)^{n} - \tilde{V}_{l}(f)^{n}
= S_{l}^{n}(1) + S_{l}^{n}(2) + O_p(1/n)$, where
\begin{align*}
S_{l}^{n}(1)= \frac{1}{n/L} \sum_{i=1}^{n/L} f'( \alpha_{(i-1)L+l}^{n})
\xi_{(i - 1)L + l}^{n}(1) \quad \mbox{ and } \quad
S_{l}^{n}(2)= \frac{1}{n/L} \sum_{i=1}^{n/L} f'( \alpha_{(i-1)L+l}^{n})
\xi_{(i-1)L+l}^{n}(2).
\end{align*}
As $f$ is even, $f'$ is odd. This implies a martingale difference property
\begin{equation*}
\mathbb{E} \Bigl[ f'( \alpha_{(i-1)L+l}^{n}) \xi_{(i-1)L+l}^{n}(1) \mid
\mathcal{F}_{ \frac{(i-1)L+l-1}{n}} \Bigr] = 0.
\end{equation*}
Then, the Cauchy-Schwarz inequality, $f'$ being of polynomial growth, Lemma
\ref{Lemma:Estimate} and Eq. \eqref{d1} imply
\begin{equation}
\label{d4}
\mathbb{E} \Bigl[ |S_{l}^{n}(1)|^{2} \Bigr] = \frac{L^{2}}{n^{2}}
\sum_{i=1}^{n/L} \mathbb{E} \bigl[ |f'( \alpha_{(i-1)L+l}^{n})
\xi_{(i-1)L+l}^{n}(1)|^{2} \bigr] \leq \frac{L}{n^{2}}.
\end{equation}
From the Cauchy-Schwarz inequality, Lemma \ref{Lemma:Estimate} and Eq. \eqref{d2}, we
also get that
\begin{align}
\label{d5}
\mathbb{E} \Bigl[ |S_{l}^{n}(2)|^{2} \Bigr] \leq \frac{L}{n} \sum_{i=1}^{n/L}
\mathbb{E} \bigl[ |f'( \alpha_{(i-1)L+l}^{n}) \xi_{(i-1)L+l}^{n}(2)|^{2} \bigr]
\leq \frac{1}{n^{2}}.
\end{align}
Then, we finish the proof via Eqs. \eqref{d4} -- \eqref{d5}. \qed
\end{proof}
The next result then implies Eq. (i), and the entire proof is complete.

\begin{lemma}
\label{m30}
Assume that the conditions of Theorem \ref{c0} are fulfilled. Then,
\begin{align*}
\mathbb{E} \left[ | D_{n}^{(4)} | \right] &\leq C \frac{L}{n}, \tag{A.3a}
\\[0.25cm]
\mathbb{E} \left[ | D_{n}^{(3)} | \right] &\leq C \frac{1}{n}, \tag{A.3b}
\\[0.25cm]
\mathbb{E} \left[ | D_{n}^{(1)} | \right] &\leq C \frac{1}{\sqrt{n}},
\tag{A.3c} \\[0.25cm]
\mathbb{E} \left[ | D_{n}^{(2)} | \right] &\leq C \left( \frac{L}{n} +
\frac{1}{ \sqrt{n}} \right) \tag{A.3d}.
\end{align*}
\end{lemma}

\begin{proof}[Proof of Lemma \ref{m30}]
We observe that part (A.3a) is a direct consequence of Lemma \ref{l1} Eq.
(A.2b). Regarding parts (A.3b) and (A.3c), we note that Lemma \ref{l1}
implies that
\begin{equation*}
\mathbb{E} \Bigl[ \bigl( \tilde{V}_{l}(f)^{n} - V(f) \bigr)^{2} \Bigr] \leq
C \frac{L}{n}.
\end{equation*}
Then, the Cauchy-Schwarz inequality and the above equation yield
\begin{equation*}
\Bigl( \mathbb{E} \left[ | D_{n}^{(1)} | \right] \Bigr)^{2} \leq C
\frac{n}{L^{2}} \sum_{l = 1}^{L} \mathbb{E} \left[ |V_{l}(f)^{n} -
\tilde{V}_{l}(f)^{n}|^{2} \right] = C \mathbb{E} \left[ | D_{n}^{(3)} |
\right].
\end{equation*}
Hence, it suffices to show part (A.3b), which follows from Lemma \ref{l1}
Eq. (A.2c).

We proceed to the proof of part (A.3d). In this part of the proofs, we adopt
the notation $t_{i,l} = (iL+l)/n$. Applying Taylor's theorem and Eq.
\eqref{Eqn:Burkholder} for $\sigma$, we can rewrite
\begin{equation*}
D_{n}^{(2)} = E_{n} + F_{n} + O_{p}(L/n)+O_{p}(1/\sqrt{n}),
\end{equation*}
with
\begin{align*}
E_{n} &= \frac{2n}{L^{2}} \sum_{l=1}^{L} \Biggl( \sum_{i=1}^{n/L} \phi'(
\sigma_{t_{i-1, l-1}}) \int_{t_{i-1,l-1}}^{t_{i,l-1}} \bigl[ \sigma_{
t_{i-1, l-1}} - \sigma_{s} \bigr]
\text{d}s \Biggr) \times \Bigl( \tilde{V}_{l}(f)^{n} - \hat{V}_{l}(f)^{n}
\Bigr), \\[0.25cm]
F_{n} &= \frac{-n}{L^{2}} \sum_{l=1}^{L} \Biggl( \sum_{i=1}^{n/L} \phi''(
\sigma_{t_{i-1, l-1}}) \int_{t_{i-1, l-1}}^{t_{i, l-1}}  \bigl[
\sigma_{ t_{i-1, l-1}} - \sigma_{s} \bigr]^{2} \text{d}s \Biggr) \times
\Bigl( \tilde{V}_{l}(f)^{n} - \hat{V}_{l}(f)^{n} \Bigr).
\end{align*}
We note that the $O_{p}(1/\sqrt{n})$ error is due to the boundary integral term
around 0 and 1. If we then recall Assumption (\textbf{H}), an application of
Eq. \eqref{Eqn:Burkholder} for $\tilde{a}$, $\tilde{\sigma}$ and $\tilde{v}$
implies that
\begin{equation*}
E_{n} = - E_{n}(1) - E_{n}(2) + O_{p}(L/n),
\end{equation*}
where
\begin{align*}
E_{n}(1) =& \frac{2n}{L^{2}} \sum_{l=1}^{L} \Biggl( \sum_{i=1}^{n/L}
\phi'(\sigma_{t_{i-1, l-1}})
\frac{L^{2}}{2n^{2}} \tilde{a}_{t_{i-1, l-1}} \Biggr) \Bigl(
\tilde{V}_{l}(f)^{n} -
\hat{V}_{l}(f)^{n} \Bigr), \\[0.25cm]
E_{n}(2) =& \frac{2n}{L^{2}} \sum_{l=1}^{L} \Biggl( \sum_{i=1}^{n/L}
\phi'(\sigma_{ t_{i-1, l-1}})
\int_{t_{i-1, l-1}}^{t_{i, l-1}} \Bigl[ \tilde{\sigma}_{t_{i-1, l-1}}
\bigl(W_{s} - W_{ t_{i-1, l-1}} \bigr)
+ \tilde{v}_{t_{i-1, l-1}} \bigl( B_{s} - B_{t_{i-1, l-1}}\bigr)
\Bigr] \text{d}s  \Biggr)
\\[0.25cm]
&\times
\Bigl( \tilde{V}_{l}(f)^{n}-\hat{V}_{l}(f)^{n} \Bigr).
\end{align*}
Now, to deal with the term $E_{n}(1)$, we first define:
\begin{equation*}
Q_{l}^{n} = \frac{L}{n} \sum_{i=1}^{n/L} \phi'( \sigma_{t_{i-1, l-1}})
\tilde{a}_{t_{i-1, l-1}},
\end{equation*}
which has the limit $\displaystyle Q = \int_{0}^{1} \phi'( \sigma_{s})
\tilde{a}_{s}\text{d}s$. We find that:
\begin{align*}
E_{n}(1) &= \frac{1}{L} \sum_{l=1}^{L} Q \bigl( \tilde{V}_{l}(f)^{n} -
\hat{V}_{l}(f)^{n} \bigr) + \frac{1}{L} \sum_{l=1}^{L} \bigl( Q_{l}^{n} - Q
\bigr) \bigl( \tilde{V}_{l}(f)^{n} - \hat{V}_{l}(f)^{n} \bigr) \\[0.25cm]
&\equiv E_{n}(1.1) + E_{n}(1.2).
\end{align*}
We note that
\begin{equation*}
E_{n}(1.1) =  \frac{Q}{n} \sum_{i=1}^n \chi_{i}^{n} \qquad \text{and} \qquad
\mathbb{E} [ | Q_{l}^{n} - Q |^{2} ] \leq C \frac{L}{n},
\end{equation*}
uniformly in $l$. Then, using the Cauchy-Schwarz inequality, the martingale
difference property of the $\chi_{i}^{n}$'s and Lemma (A.2a), we find that
\begin{equation*}
\label{m35}
\mathbb{E} \bigl[ | E_{n}(1) | \bigr] \leq \mathbb{E} \bigl[ | E_{n}(1.1)
| \bigl] + \mathbb{E} \bigl[ | E_{n}(1.2) | \bigr] \leq C \biggl(
\frac{1}{ \sqrt{n}} + \frac{L}{n} \biggr).
\end{equation*}
If we also apply these techniques to the term $F_{n}$, we get that
\begin{equation*}
\mathbb{E} \bigl[ | F_{n} | \bigr] \leq C \biggl( \frac{1}{ \sqrt{n}} +
\frac{L}{n} \biggr).
\end{equation*}
So in the rest of the proof, we are left with the term $E_{n}(2)$ only.
Here we assume that $\tilde{v}_{s} = 0$. Apart from expositional purposes,
this is without loss of generality, as the terms involving the product of
$\tilde{v}$ and $B$ are much simpler to handle, because $W$ and $B$
are independent.

We define:
\begin{align}
\label{m40}
G_{i,l}^{n} &= \phi'(\sigma_{ t_{i-1, l-1}})
\int_{t_{i-1, l-1}}^{t_{i, l-1}} \Bigl[ \tilde{\sigma}_{t_{i-1, l-1}}
\bigl(W_{s} - W_{ t_{i-1, l-1}} \bigr)
\Bigr] \text{d}s.
 \end{align}
Then,
\begin{equation}
\label{m41}
E_{n}(2) = \frac{2{n}}{L^{2}} \sum_{l=1}^{L} \Bigg( \sum_{i=1}^{n/L}
G_{i,l}^{n} \Bigg) \Bigl( \tilde{V}_{l}(f)^{n} - \hat{V}_{l}(f)^{n} \Bigr).
\end{equation}
We find that
\begin{equation}\label{m39}
\bigl( E_{n}(2) \bigr)^{2} = E_{n}(2.1) + E_{n}(2.2),
\end{equation}
where
\begin{align*}
E_{n}(2.1) &= \frac{4n^{2}}{L^{4}} \sum_{l=1}^{L} \biggl( \sum_{i=1}^{n/L}
G_{i,l}^{n} \biggr)^{2} \Bigl( \tilde{V}_{l}(f)^{n} - \hat{V}_{l}(f)^{n}
\Bigr)^{2}, \\[0.25cm]
E_{n}(2.2) &= \frac{4n^{2}}{L^{4}} \sum_{l_{a} \neq l_{b}}^{L} \biggl(
\sum_{i=1}^{n/L} G_{i,l_{a}}^{n} \biggr) \Bigl( \tilde{V}_{l_{a}}(f)^{n}
- \hat{V}_{l_{a}}(f)^{n} \Bigr) \biggl( \sum_{i=1}^{n/L} G_{i,l_{b}}^{n}
\biggr) \Bigl( \tilde{V}_{l_{b}}(f)^{n} - \hat{V}_{l_{b}}(f)^{n} \Bigr).
\end{align*}
As, for each fixed $l$, $G_{i,l}^{n}$ and $G_{j,l}^{n}$ are uncorrelated for
$i \neq j$, the Cauchy-Schwarz inequality and Lemma (A.2a) imply that
\begin{equation}
\label{m55}
\mathbb{E} \bigl[ E_{n}(2.1) \bigr] \leq  \frac{C}{n}.
\end{equation}
Recalling Eq. \eqref{m39}, we will be done with the proof of Lemma
(A.3d), if we show
\begin{equation}
\label{E22}
\mathbb{E} \bigl[ E_{n}(2.2) \bigr] \leq \frac{C}{n}.
\end{equation}
It turns out that proving Eq. \eqref{E22} is rather advanced.
We start by noting that
\begin{align}
\label{E22term}
E_{n}(2.2) = \frac{4}{L^{2}} \sum_{l_{a} \neq l_{b}}^{L} \sum_{i_{1}, i_{2},
i_{3}, i_{4} = 1}^{n/L} G_{i_{1},l_{a}}^{n} \chi_{(i_{2}-1)L+l_{a}}^{n}
G_{i_{3},l_{b}}^{n} \chi_{(i_{4}-1)L+l_{b}}^{n}.
\end{align}
We fix $l_{a}$ and $l_{b} $ in Eq. \eqref{E22term}. Then, for several choices
of $i_{1}, i_{2}, i_{3}$ and $i_{4}$, the expectation of the summand is zero.
To see this, recall that expected value of each $G_{i}^{n}$ and $\chi_{i}^{n}$
is zero. Hence, if we take the conditional expectation on the left endpoint of
the largest interval, and if the other three terms are measurable with respect
to this point in time, then the expectation vanishes. Thus, we need to study
the expectation of the summands, only when the two largest indexes are equal,
which for a fixed pair $l_{a}, l_{b}$ reduces the effective number of summands
to $Cn^{3}/L^{3}$. Without loss of generality, we assume in the following that
$i_{2} \leq i_{4} < i_{1} = i_{3}$ and $l_a < l_b$ (other relations between
the indexes are handled in an identical way as shown below).

Thus, by conditioning, we find that
\begin{align}
\begin{split}
\label{step1}
\mathbb{E} \Big[ G_{i_{1},l_{a}}^{n} \chi_{(i_{2}-1)L+l_{a}}^{n}
G_{i_{3},l_{b}}^{n} \chi_{(i_{4}-1)L+l_{b}}^{n} \Big] &=
\mathbb{E} \bigg[ \mathbb{E} \Big[G_{i_{1},l_{a}}^{n}
\chi_{(i_{2}-1)L+l_{a}}^{n} G_{i_{3},l_{b}}^{n} \chi_{(i_{4}-1)L+l_{b}}^{n}
\mid \mathcal{F}_{ t_{i_3-1, l_b-1}} \Big] \bigg]
\\[0.25cm]
& = C \frac{L^3}{n^3} \mathbb E \big[ \chi_{(i_{2}-1)L+l_{a}}^{n}
\chi_{(i_{4}-1)L+l_{b}}^{n} H_{i_{1},i_{3},l_{a},l_{b}} \big],
\end{split}
\end{align}
where $C = C_{i_{1},i_{3},l_{a},l_{b}}$ is a uniformly bounded constant and
$H_{i_{1},i_{3},l_{a},l_{b}}$ is defined by
\begin{equation*}
H_{i_{1},i_{3},l_{a},l_{b}} = \phi' \big( \sigma_{ t_{i_{1}-1, l_{a}-1}} \big)
\phi' \big( \sigma_{ t_{i_{3}-1, l_{b}-1}} \big) \tilde{ \sigma}_{t_{i_{1}-1,
l_{a}-1}} \tilde{ \sigma}_{t_{i_{3}-1, l_{b}-1}}.
\end{equation*}
By applying the Clark-Ocone formula to $\chi_{(i_{2}-1)L+l_{a}}^{n}$ and
$\chi_{(i_{4}-1)L+l_{b}}^{n}$, we deduce the representation
\begin{align*}
\chi_{(i_{2}-1)L+l_{a}}^{n} &= \sqrt{n} \int_{t_{i_{2}-1,
l_{a}-1}}^{t_{i_{2}-1, l_{a}}} \zeta_{t}^{n,i_{2},l_{a}} \text{d}W_{t},
\qquad \zeta_{t}^{n,i_{2},l_{a}} = \mathbb{E} \Big[ D_{t} f \big(
\sigma_{t_{i_{2}-1, l_{a}-1}}
\sqrt{n} \Delta_{t_{i_{2}-1, l_{a}}} W \big) \mid \mathcal G_{t} \Big]
\\[0.25cm]
\chi_{(i_{4}-1)L+l_{b}}^{n} &= \sqrt{n} \int_{t_{i_{4}-1,
l_{b}-1}}^{t_{i_{4}-1, l_{b}}} \zeta_{t}^{n,i_{4},l_{b}} \text{d}W_{t},
\qquad  \zeta_{t}^{n,i_{4},l_{b}} = \mathbb{E} \Big[ D_{t} f \big(
\sigma_{t_{i_{4}-1, l_{b}-1}}
\sqrt{n} \Delta_{t_{i_{4}-1, l_{b}}} W \big) \mid \mathcal G_{t} \Big],
\end{align*}
where $\mathcal{G}_{t} = \sigma(W_{s} \mid s \leq t)$.
Hence,
\begin{align*}
\chi_{(i_2-1)L+l_a}^{n}  \chi_{(i_4-1)L+l_b}^{n} = \delta^2 \left( 1_{[t_{i_2-1, l_a-1},
t_{i_2-1, l_a}] \times [t_{i_4-1, l_b-1}, t_{i_4-1, l_b}]}(u_1,u_2) \zeta_{u_1}^{n,i_2,l_a}
\zeta_{u_2}^{n,i_4,l_b} \right),
\end{align*}
where $\delta^{2}$ is the adjoint operator of $D^2$ introduced in Section \ref{secA.5}.
Finally, due to Assumption (\textbf{M}) and the integration by parts formula
in Eq. \eqref{IntegrationByParts}, we get that
\begin{align}
%\label{step2}
%\begin{split}
\Big| \mathbb{E} \big[ \chi_{(i_{2}-1)L+l_{a}}^{n} \chi_{(i_{4}-1)L+l_{b}}^{n}
H_{i_1,i_3,l_a,l_b} \big] \Big| &= n
\Big| \mathbb{E} \bigg[ \int_{t_{i_4-1, l_b-1}}^{t_{i_4-1, l_b}}
\int_{t_{i_2-1, l_a-1}}^{t_{i_2-1, l_a}}  \zeta_{u_1}^{n,i_2,l_a}
\zeta_{u_2}^{n,i_4,l_b} D_{u_1,u_2} \big( H_{i_1,i_3,l_a,l_b} \big)
\text{d}u_{1} \text{d}u_{2} \bigg] \Big| \nonumber \\[0.25 cm]
& \leq \frac{C}{n}, \label{step2}
%\end{split}
\end{align}
where the last line follows by Cauchy-Schwarz inequality.
From Eq. \eqref{step1} and \eqref{step2}, we conclude that
\begin{equation*}
\mathbb{E} \bigl[ E_{n}(2.2) \bigr] \leq \frac{C}{n}.
\end{equation*}
This completes the proof of Eq. \eqref{E22}, Eq. (i), and Theorem \ref{c0}. \qed
\end{proof}

\subsection{Proof of Theorem \ref{c0prime}}
We need to show that the left-hand side of Eq. (i)-(v) of the last subsection
all converge to 0 under the weaker assumptions of Theorem \ref{c0prime}. By
observing the proof of Theorem \ref{c0}, we discover that the steps behind Eq.
(ii), (iii) and (v) do not depend on the stronger Assumption (\textbf{H}) and
(\textbf{M}), nor on the differentiability of the function $f$.
Hence, we can immediately deduce that
\begin{align*}
\mathbb{E} \bigl[ | Q_{n} - U_{n} | \bigr] \to 0, \qquad
\mathbb{E} \bigl[ | U_{n} - R_{n} | \bigr] \to 0, \qquad
\mathbb{E} \bigl[ |\hat{ \Sigma}_{n} - \Sigma_{n} | \bigr] \to 0.
\end{align*}
On the other hand, (iv) follows from Section 8 (Step 2) in BGJPS6:
\begin{align*}
R_{n} \overset{p}{ \to} \Sigma.
\end{align*}
Furthermore, Step 3 and 4 in that section also imply that
\begin{align*}
\sup_{1 \leq l \leq L} \mathbb{E} \Biggl[ \Bigl| \sqrt{\frac{n}{L}} \Bigl(
V_{l}(f)^{n} - V(f) \Bigr) -  \sqrt{ \frac{L}{n}} \sum_{i = 1}^{n/L}
\chi_{(i - 1)L+ l}^{n}  \Bigr|^{1+\epsilon} \Biggr] \to 0.
\end{align*}
for $\epsilon>0$ small enough. Hence, we find by the H\"older inequality that
\begin{align*}
\mathbb{E} \bigl[ | \Sigma_{n} - Q_{n} | \bigr] \to 0,
\end{align*}
which corresponds to part (i) of the previous subsection. This finishes the
proof. \qed

%\subsection{Proof of rate optimality of subsampled power variation}
%\begin{proof}
%To prove that $n^{-1/3}$ is the optimal rate of convergence for the subsampled
%power variation, we note that the leading rates $1/\sqrt{L}$ and $L/n$ in
%Theorem \ref{c0} appear in both parts Eq. (i) and (ii) above.
%It is therefore enough to show that these two terms are optimal. As stated in
%the main text, we let
%\begin{equation}
%X_{t} = \int_{0}^{t} W_{s} \text{d}W_{s}, \quad t \in [0,1].
%\end{equation}
%Then, it follows that
%\begin{align*}
%V(f) &= \int_{0}^{1} W_{s}^{2} \text{d}s, \\[0.25cm]
%V_{l}(f)^{n} &= \frac{1}{n/L} \sum_{i = 1}^{n/L} \Biggl( \frac{ \sqrt{n}}{2}
%\biggl( \Delta_{(i - 1)L + l}^{n} W^{2} - \frac{1}{n} \biggr) \Biggr)^{2}, \\[0.25cm]
%\Sigma &= \int_{0}^{1} 2 W_{s}^{4} \text{d}s.
%\end{align*}
%As $\sigma_{s} = W_{s}$, then $\sigma_{0} = 0, \mu_{s}' = 0, \sigma_{s}' = 1$
%and $v_{s}'=0$, for all $s \in [0,1]$. We observe that $\displaystyle X_{t} =
%\frac{W_{t}^2 - t}{2}$.
%In particular, we will show the following:
%\begin{align*}
%c \left(\frac{1}{\sqrt{L}}+\frac{L}{n} \right) \leq \mathbb{E} \left[ \left(\hat \Sigma_n-\Sigma \right)^2  \right]^{1/2} \leq C \left(\frac{1}{\sqrt{L}}+\frac{L}{n} \right)
%\end{align*}
%for some $C>c>0.$ By going over the proof of Theorem \ref{c0}, it is easy to show the above upper bound. Regarding the lower bound, we observe that only the terms in (i) and (ii), that is $(\Sigma_n-Q_n)$ and $(Q_n-U_n),$ contribute to the leading order.
 %\qed
%\end{proof}

\subsection{Proof of Theorem \ref{m11}}
We begin by introducing some notation. For $m \geq i$, we define
\begin{equation}
\label{m13}
\Delta \bar{Y}_{m,i}^{n} = \sum_{j = 1}^{k_{n}} w \biggl( \frac{j}{k_{n}}
\biggr) \Bigl( \sigma_{ \frac{i}{n}} \Delta_{m+j}^{n}W + \Delta_{m+j}^{n}
\epsilon \Bigr).
\end{equation}
We note that $\Delta \bar{Y}_{m,i}^{n}$ approximates $\Delta \bar{Y}_{m}^{n}$
by evaluating $\sigma$ at the point $i/n$. Moreover, we state two auxiliary
results from \citet*{podolskij-vetter:09a}, which provide the stochastic order
of the statistics $\Delta \bar{S}_{i}^{n}$ for the processes, $S= W, X,
\epsilon$ or $Y$ (Lemma \ref{m14}), and which permits to substitute the limit
$\psi_{i}$ for $\psi_{i}^{n}$ for $i = 1,2$ without altering the consistency
statements (Lemma \ref{m16}).
\begin{lemma}
\label{m14}
Assume that $s$ is a non-negative real number, such that $\mathbb{E}[ |\epsilon_{t}|^{s}] <\infty.$ Then, for any $i$ and $n$,
\begin{equation}
\label{m15}
\mathbb{E} \Bigl[ | \Delta \bar{Y}_{m, i}^{n} |^{s} \mid \mathcal{F}_{
\frac{i}{n}} \Bigr] + \mathbb{E} \Bigl[ | \Delta \bar{Y}_{i}^{n} |^{s} \mid
\mathcal{F}_{ \frac{i}{n}} \Bigr] \leq Cn^{-s/4}.
\end{equation}
\end{lemma}

\begin{lemma}
\label{m16} Let $s \geq 0$. Then,
\begin{equation}\label{m17}
\int_{0}^{1} \biggl( \theta \psi_{2}^{n} \sigma_{u}^{2} + \frac{1}{ \theta}
\psi_{1}^{n} \omega^{2} \biggr)^{s} \text{\upshape{d}}u - \int_{0}^{1}
\biggl( \theta \psi_{2} \sigma_{u}^{2} + \frac{1}{ \theta} \psi_{1} \omega^{2}
\biggr)^{s} \text{\upshape{d}}u = o_{p} \bigl( n^{-1/4} \bigr).
\end{equation}
\end{lemma}
We again use the short form notation $t_{i, l}=(i L+l) p k_n/n$ and introduce
some approximations of $\hat{ \Sigma}_{n}^{*}$ and $\Sigma^{*}$:
\begin{align*}
\Sigma_{n}^{*} &= \frac{1}{L} \sum_{l=1}^{L} \Biggl( \frac{n^{1/4}}{\sqrt{L}}
\Bigl( V_l^{*}(q,r)^n-V^{*}(q,r) \Bigr) \Biggr)^{2}, \quad
Q_{n} = \frac{k_{n}^{2}p^{2}}{n^{3/2}} \sum_{l=1}^{L} \Biggl(
\sum_{i=1}^{n/{L p k_{n}}} \chi_{(i-1)L+l}^{n} \Biggr)^{2}, \\[0.25cm]
U_{n} &= \frac{k_{n}^{2}p^{2}}{n^{3/2}} \sum_{l=1}^{L} \sum_{i=1}^{n/{L
p k_{n}}} \bigl( \chi_{(i-1)L+l}^{n} \bigr)^{2}, \quad
R_{n} = \frac{k_{n}^{2}p^{2}}{n^{3/2}} \sum_{l=1}^{L}
\sum_{i=1}^{n/{L p k_{n}}} \mathbb{E} \bigl[ \bigl( \chi_{(i-1)L+l}^{n}
\bigr)^{2} \mid \mathcal{F}_{t_{i-1, l-1}} \bigr],
\end{align*}
where
\begin{equation*}
\eta_{i}^{n} = \frac{n^{ \frac{q+r}{4}}}{p k_{n}-2k_{n}+2}
\sum_{m, m + k_{n} - 1 \in B_{i}(p)} |
\Delta \bar{Y}^{n}_{m, (i-1) p k_{n}} |^{q} |
\Delta \bar{Y}^{n}_{m+k_{n}, (i-1) p k_{n}}|^{r},
\end{equation*}
and
\begin{equation*}
\chi_{i}^{n} = \eta_{i}^{n} - \mathbb{E} \Bigl[ \eta_{i}^{n} \mid
\mathcal{F}_{ \frac{(i-1) p k_{n}}{n}} \Bigr].
\end{equation*}
There exists a $C > 0$, independent of $i$, such that
\begin{equation}
\label{m19}
\mathbb{E} \Bigl[ \bigl( \eta_{i}^{n} \bigr)^{4} \Bigr] \leq C \qquad
\text{and} \qquad \mathbb{E} \Bigl[ \bigl( \chi_{i}^{n} \bigr)^{4} \Bigr]
\leq \frac{C}{p^{2}},
\end{equation}
where the last inequality holds as, for fixed $i$, the terms
$\Delta \bar{Y}^{n}_{m,(i-1)p k_{n}}$ are $k_{n}$-dependent.

As in the no noise setting, we complete the proof by showing
the following results and the relationship
$p/\sqrt{n} \ll \sqrt{p}/n^{1/4} \ll 1/\sqrt{L}$ (which follow
from $\sqrt{n}/Lp^2 \to \infty$):
\begin{align}
\mathbb{E} \bigl[ | \Sigma_{n}^{*} - Q_{n} | \bigr] &\leq C \biggl(
\frac{L p^{2}}{ \sqrt{n}} + \frac{1}{ \sqrt{L}}\biggr),
\tag{i} \\[0.25cm]
\mathbb{E} \bigl[ | Q_{n} - U_{n} | \bigr] &\leq C \frac{1}{ \sqrt{L}}, \tag{ii}
\\[0.25cm]
\mathbb{E} \bigl[ | U_{n} - R_{n} | \bigr] &\leq C \frac{\sqrt{p}}{n^{1/4}},
\tag{iii} \\[0.25cm]
\mathbb{E} \bigl[ | R_{n} - \Sigma^{*} | \bigr] &\leq C \biggl( \frac{p}{
\sqrt{n}} + \frac{1}{p} \biggr) \tag{iv},\\[0.25cm]
\mathbb{E} \bigl[ | \hat{\Sigma}_{n}^{*} - \Sigma_{n}^{*} | \bigr] &\leq C \frac{1}{\sqrt{L}}
\tag{v}.
\end{align}
We proceed as in the part absent of microstructure noise. First, we prove
the steps (ii), (iii), (iv), (v), which are relatively simple and short. Then,
we offer the proof of (i), which is much longer and more complicated.
\begin{proof}[Proof of (iii)]
This part is again easy. We observe that
\begin{equation*}
U_{n} - R_{n} = \frac{k_{n}^{2}p^{2}}{n^{3/2}} \sum_{i=1}^{n/{p k_{n}}}
\bigl( \chi_{i}^{n} \bigr)^{2} - \mathbb{E} \Bigl[ \bigl( \chi_{i}^{n}
\bigr)^{2} \mid \mathcal{F}_{ \frac{(i-1)p k_{n}}{n}} \Bigr] \qquad
\text{and} \qquad \mathbb{E} \bigl[ | U_{n} - R_{n} |^{2} \bigr]
\leq C \frac{p}{\sqrt{n}},
\end{equation*}
where the last part is due to the martingale difference property and Eq.
\eqref{m19}.
\qed
\end{proof}
\begin{proof}[Proof of (v)] This follows by using a difference of
squares and that $\displaystyle \frac{n^{1/4}}{\sqrt{L}} \Bigl(
V_{l}^{*}(q,r)^{n} - V^{*}(q,r) \Bigr) = O_p(1)$ and $\displaystyle
\frac{n^{1/4}}{\sqrt{L}}\Bigl(V^{*}(q,r)^{n} - V^{*}(q,r) \Bigr) =
O_{p} \biggl( \frac{1}{ \sqrt{L}} \biggr)$, where the first term
is due to Lemma \ref{s1}. We skip the details. \qed
\end{proof}

\begin{proof}[Proof of (ii)] We define $S_{l}^{m} =
\sum_{i=1}^{m}\chi_{(i-1)L+l}^{n}$ and
$T_{l}^{m} = \sum_{i=1}^{m} \bigl( \chi_{(i-1)L+l}^{n} \bigr)^{2}$
and note that
\begin{align*}
Q_{n} - U_{n} = \frac{k_{n}^{2}p^{2}}{n^{3/2}} \sum_{l=1}^{L} A_{l}^{n}
\end{align*}
where
\begin{align*}
A_{l}^{n}= \sum_{i, j=1 ~ i \neq j}^{n/L p k_{n}} \chi_{(i-1)L+l}^{n}
\chi_{(j-1)L+l}^{n} = \Bigl( S_{l}^{n/L p k_{n}} \Bigr)^{2} -
T_{l}^{n/L p k_{n}}.
\end{align*}
Since $A_{l_{1}}^{n}$ and $A_{l_{2}}^{n}$ are uncorrelated for every $l_{1}
\neq l_{2}$, we find that
\begin{align}
\label{Eqn:EstimateST}
\mathbb{E} \bigl[ | Q_{n} - U_{n} |^{2} \bigr] &= \frac{k_{n}^{4}
p^{4}}{n^{3}} \sum_{l=1}^{L} \mathbb{E} \bigl[ (A_{l}^{n})^{2} \bigr]
\leq C \frac{k_{n}^{4}p^{4}}{n^{3}} \sum_{l=1}^{L} \biggl( \mathbb{E}
\Bigl[ (S_{l}^{n/{L p k_{n}}})^{4} \Bigr] + \mathbb{E} \Bigl[
(T_{l}^{n/{L p k_{n}}})^{2} \Bigr] \biggr).
\end{align}
%To estimate the first sum, we observe that $(S_{l}^{m})_{m=1}^{n/{L p
%k_{n}}}$
%is a discrete martingale for each $l$. Then, the discrete Burkholder and
%Cauchy-Schwarz inequalities together with Eq. \eqref{m19} imply that
%\begin{align}
%\label{m21} \mathbb{E} \bigl[ (S_{l}^{n/{L p k_{n}}})^{4} \bigr] &\leq C
%\mathbb{E} \Biggl[ \Biggl( \sum_{i=1}^{n/{L pk_{n}}} \bigl(
%\chi_{(i-1)L+l}^{n} \bigr)^{2} \Biggr)^{2} \Biggr] \leq C \biggl(
%\frac{n}{L p k_{n}} \biggr)^{2} \frac{1}{p^{2}} \leq C
%\frac{n}{L^{2}p^{4}}
%\end{align}
%Another application of the Cauchy-Schwarz inequality and Eq. \eqref{m19} yield:
Proceeding as in the noiseless case, Eq. \eqref{m19}, the Cauchy-Schwarz and the
Burkholder inequalities yield:
\begin{align} \label{m21b}
\mathbb{E} \bigl[ (S_{l}^{n/{L p k_{n}}})^{4} \bigr] \leq C \frac{n}{L^{2}p^{4}}
\quad \text{and} \quad \mathbb{E} \bigl[ (T_{l}^{n/{L pk_{n}}})^{2}] \leq C
\frac{n}{L^{2} p^{4}}.
\end{align}
Then, we finish the proof using Eqs. \eqref{Eqn:EstimateST} and \eqref{m21b}.
\qed
\end{proof}
\begin{proof}[Proof of (iv)]
We draw upon the proof of Lemma 8 in \citet*{podolskij-vetter:09a}. We find
that
\begin{align}
\Sigma^{*} &= 2 \theta \int_{0}^{1} \int_{0}^{2} h \bigl( \sigma_{u}, t, f(s)
\bigr) \text{d}s \text{d}u \nonumber \\[0.25cm]
&= \frac{2}{\sqrt{n}} \frac{p k_{n}}{n} \sum_{i = 1}^{n / p k_{n}} \sum_{j =
0}^{2 k_{n} - 1} h \Biggl( \sigma_{ \frac{(i-1)pk_{n}}{n}}, t_{n}, f^{n}
\biggl( \frac{j}{k_{n}}  \biggr) \Biggr) + O_{p} \biggl( \frac{ p k_{n}}{n}
\biggr) \nonumber  \\[0.25cm]
& \equiv R_{n}' + O_{p} \biggl( \frac{ p k_{n}}{n} \biggr) \label{RE1}
\end{align}
where the function $h$ and its associated notation were introduced in Eq. (3.6)
in \citet*{podolskij-vetter:09a}.

To estimate the term $R_{n}-R_{n}'$, we recall that
\begin{equation*}
R_{n} = \frac{k_{n}^{2} p^2}{n^{3/2}} \sum_{i=1}^{n/pk_{n}} \mathbb{E}
\Bigl[( \chi_{i}^{n})^{2} \mid \mathcal{F}_{ \frac{(i-1) pk_{n}}{n}}
\Bigr].
\end{equation*}
For $m \geq l \geq i,$ we get
\begin{align*}
h \Biggl( \sigma_{ \frac{i}{n}}, t_{n}, f^{n} \biggl( \frac{m-l}{k_{n}} \biggr)
\Biggr) &= \mathbb{E} \Bigl[ | n^{1/4} \Delta \bar{Y}_{m,i}^{n}|^{q} | n^{1/4} \Delta
\bar{Y}_{m+k_{n},i}^{n} |^{r} \times |n^{1/4} \Delta \bar{Y}_{l,i}^{n}|^{q}
|n^{1/4} \Delta \bar{Y}_{l+k_{n},i}^{n}|^{r} \mid \mathcal{F}_{
\frac{i}{n}} \Bigr] \\[0.25cm]
&- \mathbb{E} \Bigl[ |n^{1/4} \Delta \bar{Y}_{m,i}^{n}|^{q} |n^{1/4} \Delta
\bar{Y}_{m+k_{n},i}^{n}|^{r} \mid \mathcal{F}_{ \frac{i}{n}} \Bigr] \times
\mathbb{E} \Bigl[ |n^{1/4} \Delta \bar{Y}_{l,i}^{n}|^{q} |n^{1/4} \Delta
\bar{Y}_{l+k_{n}, i}^{n}|^{r} \mid \mathcal{F}_{ \frac{i}{n}} \Bigr].
\end{align*}
Note that the above term vanishes for $m - l \geq 2 k_{n}$. Then, by denoting
$N = pk_{n} - 2k_{n} + 2,$ we find that
\begin{align*}
N \mathbb{E} \Bigl[( \chi_{i}^{n})^{2} \mid \mathcal{F}_{
\frac{(i-1)pk_{n}}{n}} \Bigr] &= h \bigl( \sigma_{ \frac{(i-1)pk_{n}}{n}},
t_{n}, f^{n}(0) \bigr) + \frac{2}{N} \sum_{j = 1}^{2 k_{n} - 1} (N - j)
h \Biggl( \sigma_{ \frac{(i-1)pk_{n}}{n}}, t_{n}, f^{n} \biggl( \frac{j}{k_{n}}
\biggr) \Biggr) \\[0.25cm]
&= 2 \sum_{j = 0}^{2 k_{n}-1} h \Biggl( \sigma_{ \frac{(i-1)pk_{n}}{n}}, t_{n},
f^{n} \biggl( \frac{j}{k_{n}} \biggr) \Biggr) + O_{p}(1) + O_{p} \biggl(
\frac{k_{n}}{p} \biggr).
\end{align*}
This yields that:
\begin{equation*}
\frac{pk_{n}}{\sqrt{n}} \mathbb{E} \Bigl[ \bigl( \chi_{i}^{n}
\bigr)^{2} \mid \mathcal{F}_{ \frac{(i-1) pk_{n}}{n}} \Bigr] =
\frac{2}{\sqrt{n}} \sum_{j = 0}^{2 k_{n}-1} h \Biggl( \sigma_{
\frac{(i-1)pk_{n}}{n}}, t_{n}, f^{n} \biggl( \frac{j}{k_{n}} \biggr) \Biggr) +
O_{p} \biggl( \frac{1}{\sqrt{n}} \biggr) + O_{p} \biggl( \frac{1}{p} \biggr)
\end{equation*}
uniformly in $i$. As a result,
\begin{equation}
\label{RE2}
\mathbb{E} \bigl[ | R_{n} - R_{n}' | \bigr] \leq C \biggl( \frac{1}{
\sqrt{n}}+ \frac{1}{p} \biggr).
\end{equation}
In view of Eqs. \eqref{RE1} -- \eqref{RE2}, the proof is complete. \qed
\end{proof}
As in the noiseless setting, we need to prepare a bit to show Eq. (i).
We denote by:
\begin{equation*}
\tilde{V}_{l}^{*}(q,r)^{n} = \frac{L p k_{n}}{n} \sum_{i=1}^{n/{L p
k_{n}}} \eta_{(i-l)L+l}^{n}, \qquad \hat{V}_{l}^{*}(q,r)^{n} =
\frac{L p k_{n}}{n} \sum_{i=1}^{n/{L p k_{n}}} \mathbb{E} \Bigl[
\eta_{(i-l)L+l}^{n} \mid \mathcal{F}_{t_{i-1,l-1}} \Bigr].
\end{equation*}
Then, from the decomposition
\begin{equation*}
V_{l}^{*}(q,r)^{n} - V^{*}(q,r) = \Bigl( V_{l}^{*}(q,r)^{n} -
\tilde{V}_{l}^{*}(q,r)^{n} \Bigr) + \Bigl( \tilde{V}_{l}^{*}(q,r)^{n} -
\hat{V}_{l}^{*}(q,r)^{n} \Bigr) + \Bigl( \hat{V}_{l}^{*}(q,r)^{n} -
V^{*}(q,r) \Bigr)
\end{equation*}
we find that
\begin{equation*}
\Sigma_{n}^{*} - Q_{n} = D_{n}^{(1)} + D_{n}^{(2)} + D_{n}^{(3)} + D_{n}^{(4)},
\end{equation*}
where
\begin{align*}
D_{n}^{(1)} &= \frac{2\sqrt{n}}{L^{2}} \sum_{l=1}^{L} \Bigl( V_{l}^{*}(q,r)^{n}
- \tilde{V}_{l}^{*}(q,r)^{n} \Bigr) \Bigl( \tilde{V}_{l}^{*}(q,r)^{n} -
V^{*}(q,r) \Bigr), \\[0.25cm]
D_{n}^{(2)} &= \frac{2\sqrt{n}}{L^{2}} \sum_{l=1}^{L} \Bigl(
\hat{V}_{l}^{*}(q,r)^{n} - V^{*}(q,r) \Bigr) \Bigl( \tilde{V}_{l}^{*}(q,r)^{n}
- \hat{V}_{l}^{*}(q,r)^{n} \Bigr), \\[0.25cm]
D_{n}^{(3)} &= \frac{\sqrt{n}}{L^{2}} \sum_{l=1}^{L} \Bigl( V_{l}^{*}(q,r)^{n}
- \tilde{V}_{l}^{*}(q,r)^{n} \Bigr)^{2}, \\[0.25cm]
D_{n}^{(4)} &= \frac{\sqrt{n}}{L^{2}} \sum_{l=1}^{L} \Bigl(
\hat{V}_{l}^{*}(q,r)^{n} - V^{*}(q,r) \Bigr)^{2}.
\end{align*}
To bound these terms, we exploit the following auxiliary Lemma.
\begin{lemma}
\label{s1}
Assume that the conditions of Theorem \ref{m11} are fulfilled. Then, uniformly
in $l$:
\begin{align}
\mathbb{E} \Bigl[ | \tilde{V}_{l}^{*}(q,r)^{n} - \hat{V}_{l}^{*}(q,r)^{n}
|^{2} \Bigr] &\leq C \frac{L}{\sqrt{n}}, \tag{A.6a} \\[0.25cm]
\mathbb{E} \Bigl[ | \hat{V}_{l}^{*}(q,r)^{n} - V^{*}(q,r) |^{2} \Bigr] &\leq
C \Bigl( \frac{L^{2}p^{2}}{n} + \frac{1}{\sqrt{n}} \Bigr), \tag{A.6b}
\\[0.25cm]
\mathbb{E} \Bigl[ | V_{l}^{*}(q,r)^{n} - \tilde{V}_{l}^{*}(q,r)^{n} |^{2}
\Bigr] &\leq C \frac{ L p^2}{ n} \tag{A.6c}.
\end{align}
\end{lemma}
\begin{proof}[Proof of Lemma \ref{s1}] Part (A.6a) follows by exploiting the
martingale difference property with Eq. \eqref{m19}. To prove part (A.6b), we
start with the decomposition
\begin{equation*}
\hat{V}_{l}^{*}(q,r)^{n} - V^{*}(q,r) = \Bigl( \hat{V}_{l}^{*}(q,r)^{n} -
\check{V}_{l}^{*}(q,r)^{n} \Bigr) + \Bigl( \check{V}_{l}^{*}(q,r)^{n} -
V^{*}(q,r) \Bigr),
\end{equation*}
where
\begin{align}\label{vboterm}
\check{V}_{l}^{*}(q,r)^{n} = \mu_{q} \mu_{r} \frac{L p k_{n}}{n}
\sum_{i=1}^{n/{L p k_{n}}} \biggl( \theta \psi_{2} \sigma_{t_{i-1,
l-1}}^{2} + \frac{1}{ \theta} \psi_{1} \omega^{2} \biggr)^{ \frac{q+r}{2}}.
\end{align}
To deal with the second term, we recall Lemma \ref{l1} Eq. (A.2b). Hence,
the Riemann approximation yields
\begin{equation}\label{m24}
\mathbb{E} \Bigl[ | \check{V}_{l}^{*}(q,r)^{n} - V^{*}(q,r) |^{2} \Bigr]
\leq  C\frac{L^{2}p^{2}}{n}.
\end{equation}
To estimate the first term, let $m \in B_{i}(p)$. Taking $q$ and $r$ to be
even non-negative integers, we invoke Assumption (\textbf{N}) and the
binomial expansion theorem to conclude that
\begin{equation*}
\mathbb{E} \Bigl[ | n^{1/4} \Delta \bar{Y}_{m,(i-1)pk_{n}}^{n} |^{q}
| n^{1/4} \Delta \bar{Y}_{m+k_{n},(i-1)pk_{n}}^{n} |^{r}
\mid \mathcal{F}_{ \frac{(i-1)pk_{n}}{n}} \Bigr] =
\mu_{q} \mu_{r} \biggl( \theta \psi_{2} \sigma^{2}_{
\frac{(i-1)pk_{n}}{n}} + \frac{1}{ \theta} \psi_{1} \omega^{2} \biggr)^{
\frac{q+r}{2}} + o_{p} \bigl(n^{-1/4} \bigr),
\end{equation*}
uniformly in $i$ and $m.$ Hence, we find that
\begin{equation*}
\mathbb{E} \bigl[ \eta_{i}^{n} \mid \mathcal{F}_{t_{i-1,l-1}} \bigr] = \mu_{q}
\mu_{r} \biggl( \theta \psi_{2} \sigma^{2}_{ \frac{(i-1)pk_{n}}{n}} +
\frac{1}{ \theta} \psi_{1} \omega^{2} \biggr)^{ \frac{q+r}{2}} + o_{p}
\bigl(n^{-1/4} \bigr),
\end{equation*}
uniformly in $i$ and $m$. Using these insights, we can finish part (A.6b) by
deducing that
\begin{equation}
\label{m26}
\mathbb{E} \Bigl[ | \hat{V}_{l}^{*}(q,r)^{n} - \check{V}_{l}^{*}(q,r)^{n}
|^{2} \Bigr] \leq \frac{C}{ \sqrt{n}}.
\end{equation}
As for the proof of Eq. (A.6c), we proceed as in the proof of Eq. (A.2c);
see also page 2818 in \citet*{podolskij-vetter:09a}. Thus, we just provide
a sketch of the main steps for $r = 0$. For any $m \geq i$, we define
\begin{align*}
\xi_{m,i}^{n}(1) \equiv & \sum_{j=1}^{k_{n}} w \Big( \frac{j}{k_{n}} \Big)
\Biggl(
\frac{1}{n} a_{ \frac{i}{n}} + \int_{ \frac{m+j-1}{n}}^{ \frac{m+j}{n}} \Bigl[
\tilde{ \sigma}_{ \frac{i}{n}} \bigl( W_{s} - W_{ \frac{i}{n}} \bigr) +
\tilde{v}_{ \frac{i}{n}} \bigl( B_{s} - B_{ \frac{i}{n}} \bigr)
\Bigr] \text{d}
W_{s} \Biggr) \\[0.25cm]
\xi_{m,i}^{n}(2) \equiv & \sum_{j=1}^{k_{n}} w \Big( \frac{j}{k_{n}} \Big)
\Biggl( \int_{ \frac{m+j-1}{n}}^{ \frac{m+j}{n}} \bigl( a_{s} - a_{
\frac{i}{n}} \bigr) \text{d}s + \int_{ \frac{m+j-1}{n}}^{ \frac{m+j}{n}}
\int_{ \frac{i}{n}}^{s} \tilde{a}_{u} \text{d}u \text{d}W_{s} \\[0.25cm]
&+ \int_{ \frac{m+j-1}{n}}^{ \frac{m+j}{n}} \biggl( \int_{ \frac{i}{n}}^{s}
\bigl( \tilde{ \sigma}_{u} - \tilde{ \sigma}_{ \frac{i}{n}} \bigr) \text{d}
W_{u} + \int_{ \frac{i}{n}}^{s} \bigl( \tilde{v}_{u} - \tilde{v}_{
\frac{i}{n}} \bigr) \text{d}B_{u}
\biggr) \text{d}W_{s} \Biggr).
\end{align*}
We note that $\Delta \bar{Y}^{n}_{m} - \Delta \bar{Y}^{n}_{m,i} =
\xi_{m,i}^{n}(1) + \xi_{m,i}^{n}(2)$.
Assumption (\textbf{H}), the H{\"o}lder and Burkholder inequalities imply
\begin{align}
\label{s4}
\mathbb{E} \bigl[ | \xi_{m,i}^{n}(1)|^{4} \bigr] & \leq C
\frac{p^{2}}{n^{2}}, \\[0.25cm]
\label{s5}
\mathbb{E} \bigl[ | \xi_{m,i}^{n}(2)|^{4} \bigr] & \leq C
\frac{p^{4}}{n^{3}}.
\end{align}
Now, we let $f(x) = |x|^{q}$. Taylor's theorem then yields that
\begin{equation*}
V_{l}^{*}(q,r) - \tilde{V}_{l}^{*}(q,r)^{n} = S_{l}^{n}(1) +
S_{l}^{n}(2) + O_{p} \biggl( \frac{p}{\sqrt{n}} \biggr),
\end{equation*}
where
\begin{align*}
S_{l}^{n}(1) &= \frac{Lpk_{n}}{n} \frac{n^{q/4}}{pk_{n}-2k_{n}+2}
\sum_{i=1}^{n/L p k_{n}} \sum_{m \in B_{(i-1)L+l}(p)} f'\Bigl( \Delta
\bar{Y}^{n}_{m,((i-1)L+(l-1))p k_{n}} \Bigr)
\xi_{m,(i-1)L+(l-1))p k_{n}}^{n}(1), \\[0.25cm]
S_{l}^{n}(2) &= \frac{Lpk_{n}}{n} \frac{n^{q/4}}{pk_{n}-2k_{n}+2}
\sum_{i=1}^{n/L p k_{n}} \sum_{m \in B_{(i-1)L+l}(p)} f' \Bigl( \Delta
\bar {Y}^{n}_{m,((i-1)L+(l-1)) p k_{n}} \Bigr) \xi_{m,(i-1)L+(l-1))
pk_{n}}^{n}(2).
\end{align*}
In order to bound these terms, we note that Assumption (\textbf{N})
implies that $(W,B,\epsilon) \overset{d}{=} -(W,B,\epsilon)$. Also, since
$f' \bigl( \Delta \bar{Y}^{n}_{m,i} \bigr)$ is an odd function and
$\xi_{m,i}^{n}(1)$ is an even function of $(W,B,\epsilon)$, it follows that
\begin{equation*}
\mathbb{E} \Big[ f' \bigl( \Delta \bar{Y}^{n}_{m, ipk_{n}} \bigr) \xi_{m,
ipk_{n}}^{n}(1) \mid \mathcal{F}_{ \frac{ipk_{n}}{n}} \Big] = 0.
\end{equation*}
This property---together with the Cauchy-Schwarz inequality, Eqs.
\eqref{m15} and \eqref{s4}---means that
\begin{align}
\label{s7}
\mathbb{E} \bigl[ | S_{l}^{n}(1) |^{2} \bigr] \leq C \frac{Lp^{2}}{n}.
\end{align}
Applying the Cauchy-Schwarz inequality again, combined with Eqs. \eqref{m15}
and \eqref{s5}, we also find that
\begin{align} \label{s8}
\mathbb{E} \bigl[ | S_{l}^{n}(2) |^{2} \bigr] \leq C \frac{p^2}{n},
\end{align}
and with Eqs. \eqref{s7} -- \eqref{s8} at hand, the proof is complete. \qed
\end{proof}
The following results are then sufficient to complete the proof of
Theorem \ref{m11}.
\begin{lemma}
\label{s3}
Assume that the conditions of Theorem \ref{m11} are
fulfilled. Then, it holds that
\begin{align}
\mathbb{E} \bigl[ | D_{n}^{(4)} | \bigr] &\leq C \biggl(
\frac{L p^{2}}{\sqrt{n}} + \frac{1}{L} \biggr), \tag{A.7a} \\[0.25cm]
\mathbb{E} \bigl[ | D_{n}^{(3)} | \bigr] &\leq C \frac{p^2}{\sqrt{n}}, \tag{A.7b}
\\[0.25cm]
\mathbb{E} \bigl[ | D_{n}^{(1)} | \bigr] &\leq C \Biggl( \frac{p}{n^{1/4}} +
\frac{\sqrt{L} p^{2}}{\sqrt{n}} \Biggr). \tag{A.7c}
\end{align}
\end{lemma}

\begin{lemma}\label{s3b}
Assume that the conditions of Theorem \ref{m11} are fulfilled.
Then, it holds that
\begin{equation}
\mathbb{E} \bigl[ | D_{n}^{(2)} | \bigr] \leq C \biggl( \frac{1}{
\sqrt{L}} + \frac{L p^{3/2}}{ \sqrt{n}} \biggr). \tag{A.8a}
\end{equation}
\end{lemma}

\begin{proof}[Proof of Lemma \ref{s3}]
Again, part (A.7a) follows from Eq. (A.6b). Concerning parts (A.7b) --
(A.7c), we apply Lemma \ref{s1} and find that
\begin{equation}
\label{SomeEstimate2}
\mathbb{E} \biggl[ \Bigl( \tilde{V}_{l}^{*}(q,r)^{n} - V^{*}(q,r) \Bigr)^{2}
\biggr] \leq C \biggl( \frac{L}{\sqrt{n}} + \frac{L^{2}p^{2}}{n} \biggr).
\end{equation}
Then, the Cauchy-Schwarz inequality and Eq. \eqref{SomeEstimate2} yield
\begin{equation*}
\bigl( \mathbb{E} \bigl[ | D_{n}^{(1)} | \bigr] \bigr)^{2} \leq C
\biggl( 1 + \frac{L p^{2}}{\sqrt{n}} \biggr) \frac{\sqrt{n}}{L^{2}}
\sum_{l=1}^{L} \mathbb{E} \bigl[ | V_{l}^{*}(q,r)^{n} -
\tilde{V}_{l}^{*}(q,r)^{n} |^{2} \bigr] = C \biggl( 1 +
\frac{L p^2}{\sqrt{n}} \biggr) \mathbb{E} \bigl[ | D_{n}^{(3)} | \bigr].
\end{equation*}
Hence, it is enough to show part (A.7b), which follows from Eq. (A.6c). \qed
\end{proof}

\begin{proof}[Proof of Lemma \ref{s3b}]
We start by denoting $\displaystyle \phi(x) = \mu_{q} \mu_{r} \Bigl(
\theta \psi_{2} x^{2} + \frac{1}{ \theta} \psi_{1}
\omega^{2} \Bigr)^{ \frac{q+r}{2}}$. Note that $\phi(x)$ is a smooth
function of $x$, because both $q$ and $r$ are even. After recalling Eq.
\eqref{vboterm}, an application of Taylor's theorem, the
Cauchy-Schwarz inequality and Eq. \eqref{Eqn:Burkholder} for $\sigma$
then implies that
\begin{equation*}
D_{n}^{(2)} = E_{n} + F_{n} + G_{n} + O_{p} \biggl( \frac{L p^{3/2}}{
\sqrt{n}} \biggr)+O_{p} \biggl(\frac{p}{n^{1/4}} \biggr),
\end{equation*}
where the last error term comes from the boundary integral around 0 and 1, and
\begin{align*}
E_{n} &= \frac{2\sqrt{n}}{L^{2}} \sum_{l=1}^{L} \Biggl(
\sum_{i=1}^{n/{L p k_{n}}} \phi'( \sigma_{{t_{i-1,l-1}}})
\int_{t_{i-1,l-1}}^{t_{i,l-1}} \bigl( \sigma_{{t_{i-1,l-1}}} - \sigma_{s}
\bigr) \text{d}s \Biggr) \times \Bigl( \tilde{V}_{l}^{*}(q,r)^{n} -
\hat{V}_{l}^{*}(q,r)^{n} \Bigr), \\[0.25cm]
F_{n} &= -\frac{\sqrt{n}}{L^{2}} \sum_{l=1}^{L} \Biggl(
\sum_{i=1}^{n/{L p k_{n}}} \phi''( \sigma_{{t_{i-1,l-1}}})
\int_{t_{i-1,l-1}}^{t_{i,l-1}} \bigl( \sigma_{{t_{i-1,l-1}}} -
\sigma_{s} \bigr)^{2} \text{d}s \Biggr) \times \Bigl(
\tilde{V}_{l}^{*}(q,r)^{n} - \hat{V}_{l}^{*}(q,r)^{n} \Bigr), \\[0.25cm]
G_n &= \frac{2\sqrt{n}}{L^{2}} \sum_{l=1}^{L} \Bigl( \hat{V}_{l}^{*}(q,r)^{n} -
\check{V}_{l}^{*}(q,r)^{n} \Bigr) \Bigl( \tilde{V}_{l}^{*}(q,r)^{n} -
\hat{V}_{l}^{*}(q,r)^{n} \Bigr).
\end{align*}
From the Cauchy-Schwarz inequality, Lemma \ref{s1}, Eq. (A.6a) and Eq.
\eqref{m26}, we get
\begin{equation*}
\mathbb{E} \bigl[ | G_{n} | \bigr] \leq \frac{C}{ \sqrt{L}}.
\end{equation*}
We recall Assumption (\textbf{H}), apply Eq.
\eqref{Eqn:Burkholder} to $\tilde{a}, \tilde{ \sigma}$, and
subsequently use Taylor's theorem to conclude that
\begin{equation*}
E_{n} = - E_{n}(1) - E_{n}(2) + O_{p} \biggl( \frac{L p^{3/2}}{
\sqrt{n}} \biggr),
\end{equation*}
where
\begin{align*}
E_{n}(1) &= \frac{2\sqrt{n}}{L^{2}} \sum_{l=1}^{L} \Biggl(
\sum_{i=1}^{n/{L p k_{n}}} \phi'( \sigma_{{t_{i-1,l-1}}})
\frac{L^{2}k_{n}^{2}p^{2}}{2n^{2}} \tilde{a}_{t_{i-1,l-1}} \Biggr)
\Bigl( \tilde{V}_{l}^{*}(q,r)^{n} - \hat{V}_{l}^{*}(q,r)^{n} \Bigr), \\[0.25cm]
E_{n}(2) &= \frac{2\sqrt{n}}{L^{2}} \sum_{l=1}^{L} \Biggl(
\sum_{i=1}^{n/{L p k_{n}}} \phi'( \sigma_{{t_{i-1,l-1}}})
\int_{t_{i-1, l-1}}^{t_{i, l-1}} \Bigl[ \tilde{\sigma}_{t_{i-1, l-1}}
\bigl(W_{s} - W_{ t_{i-1, l-1}} \bigr)
+ \tilde{v}_{t_{i-1, l-1}} \bigl( B_{s} - B_{t_{i-1, l-1}}\bigr)  \Bigr] \text{d}s  \Biggr)
 \\[0.25cm]
& \times \Bigl( \tilde{V}_{l}^{*}(q,r)^{n} - \hat{V}_{l}^{*}(q,r)^{n} \Bigr).
\end{align*}
For the term $E_n(1)$, we proceed as in the noiseless case. After recalling
Eqs. \eqref{m19} and (A.6a), we find that
\begin{equation*}
\mathbb{E} \bigl[ | E_{n}(1) | \bigr] \leq C \Biggl( \frac{p}{n^{1/4}} +
\frac{L p^{3/2}}{\sqrt{n}}  \Biggr).
\end{equation*}
Next, the $F_{n}$-term can be handled in a similar fashion. Thus, we get
the estimate
\begin{equation*}
\mathbb{E} \bigl[ | F_{n} | \bigr] \leq C \Biggl( \frac{p}{n^{1/4}} +
\frac{L p^{3/2}}{\sqrt{n}} \Biggr).
\end{equation*}
So we will be done, if we can show that
\begin{equation}
\label{m56}
\mathbb{E} \bigl[ | E_{n}(2) | \bigr] \leq C \frac{p}{n^{1/4}}.
\end{equation}
Throughout the remainder of the proof, we assume that $r = 0$, so that
$q$ is an even integer. Again, this is without loss of generality. We
then appeal to the binomial theorem in order to find an expansion of
\begin{equation*}
| \Delta \bar{Y}_{m,i}^{n} |^{q} = ( \Delta \bar{Y}_{m,i}^{n} )^{q} =
\Bigl( \sigma_{ \frac{i}{n}} \Delta \bar{W}_{m}^{n} + \Delta \bar{
\epsilon}_{m}^{n} \Bigr)^{q},
\end{equation*}
whereby we can separate $\tilde{V}_{l}^{*}(q,r)^{n} -
\hat{V}_{l}^{*}(q,r)^{n}$, and hence $E_{n}(2)$, into $q + 1$ terms
of the form
\begin{equation*}
E_{n}(2) = \sum_{s=0}^{q} E_{n}^{(s)}(2),
\end{equation*}
where
\begin{align*}
E_{n}^{(s)}(2) &= \frac{2\sqrt{n}}{L^{2}} \sum_{l=1}^{L} \Biggl(
\sum_{i=1}^{n/{L p k_{n}}} \phi'( \sigma_{{t_{i-1,l-1}}})
\int_{t_{i-1, l-1}}^{t_{i, l-1}} \Bigl[ \tilde{\sigma}_{t_{i-1, l-1}}
\bigl(W_{s} - W_{ t_{i-1, l-1}} \bigr)
+ \tilde{v}_{t_{i-1, l-1}} \bigl( B_{s} - B_{t_{i-1, l-1}}\bigr) \Bigr] \text{d}s  \Biggr)
 \\[0.25cm]
& \times
\Biggl( \frac{L p k_{n}}{n} \sum_{i=1}^{n/{L p k_{n}}}
\chi_{(i-1)L+l}^{n}(s) \Biggr), \\[0.25cm]
\chi_{i}^{n}(s) &= \frac{q!}{s!(q-s)!} \frac{n^{q/4}}{pk_{n}-k_{n}+2} \sum_{
m, m + k_{n} - 1
\in B_{i}(p)} \Bigl( \sigma_{ \frac{(i-1)pk_{n}}{n}} \Bigr)^{q-s}
\Bigl( (\Delta \bar{W}_{m}^{n})^{q-s}(\Delta \bar{ \epsilon}_{m}^{n})^{s} -
\mathbb{E} \bigl[ (\Delta \bar{W}_{m}^{n})^{q-s}(\Delta \bar{
\epsilon}_{m}^{n})^{s} \bigr] \Bigr).
\end{align*}
As a result, it is sufficient to show that
\begin{equation}
\label{m58b}
\mathbb{E} \bigl[ | E_{n}^{(s)}(2) | \bigr] \leq C \frac{p}{n^{1/4}},
\end{equation}
where $s$ is an arbitrary integer chosen from $0 \leq s \leq q$. Note the
equality:
\begin{align*}
(\Delta \bar{W}_{m}^{n})^{q-s}(\Delta \bar{ \epsilon}_{m}^{n})^{s} -
\mathbb{E} \bigl[ (\Delta \bar{W}_{m}^{n})^{q-s}(\Delta \bar{
\epsilon}_{m}^{n})^{s} \bigr] &=
\mathbb{E} \bigl[ (\Delta \bar{\epsilon}_{m}^{n})^{s} \bigr]
\Bigl( ( \Delta \bar{W}_{m}^{n})^{q-s} - \mathbb{E} \bigl[ (\Delta
\bar{W}_{m}^{n})^{q-s} \bigr] \Bigr) \\[0.25cm]
&+ (\Delta \bar{W}_{m}^{n})^{q-s} \Bigl( ( \Delta \bar{ \epsilon}_{m}^{n})^{s}
- \mathbb{E} \bigl[ (\Delta \bar{ \epsilon}_{m}^{n})^{s} \bigr] \Bigr).
\end{align*}
We then divide $\chi_{i}^{n}(s)$, and hence $E_{n}^{(s)}(2)$, into two parts and
denote (by preserving the above order)
\begin{equation}
E_{n}^{(s)}(2) = \bar{E}_{n}^{(s)}(2) + \tilde{E}_{n}^{(s)}(2).
\end{equation}
The term $\bar{E}_n^{(s)}(2)$ can be handled using a decomposition as
in Eq. \eqref{m39} in the no noise proof:
\begin{equation*}
\bigl( \bar{E}_{n}^{(s)}(2) \bigr)^{2} = \bar{E}_{n}^{(s)}(2.1) +
\bar{E}_{n}^{(s)}(2.2),
\end{equation*}
where $\bar{E}_{n}^{(s)}(2.1)$ and $\bar{E}_{n}^{(s)}(2.2)$ are, respectively,
composed of squared and mixed terms. We recall that the sequence $n^{s/4}
\mathbb{E} \bigl[ (\Delta \bar{ \epsilon}_{m}^{n} )^{s} \bigr]$ is uniformly
bounded in $m$ and $n$. Then, proceeding as in Eq. \eqref{m55}, we find that
\begin{equation}
\label{m69}
\mathbb{E} \Bigl[ \bar{E}_{n}^{(s)}(2.1) \Bigr] \leq C \frac{p^{2}}{
\sqrt{n}},
\end{equation}
and
\begin{equation}
\label{m70}
\mathbb{E} \Bigl[ \bar{E}_{n}^{(s)}(2.2) \Bigr] \leq C \frac{p^2}{\sqrt{n}}.
\end{equation}
For the last term $\tilde{E}_{n}^{(s)}(2)$, we recall that $X \Perp \epsilon$
and $( \Delta \bar{ \epsilon}_{m}^{n})^{s} - \mathbb{E} \bigl[ (\Delta \bar{
\epsilon}_{m}^{n})^{s} \bigr]$ has mean zero. Then,
we decompose $\bigl( \tilde{E}_{n}^{(s)}(2) \bigr)^{2}$ as in Eq. \eqref{m39},
and since the mixed terms for different $l$'s are mean zero, we find that
\begin{equation}
\label{m71}
\mathbb{E} \Bigl[ \bigl( \tilde{E}_{n}^{(s)}(2) \bigr)^{2} \Bigr] \leq C
\frac{p^{2}}{ \sqrt{n}}.
\end{equation}
Hence, Eqs. \eqref{m69} -- \eqref{m71} lead to Eq. \eqref{m58b}. \qed
\end{proof}

\subsection{Proof of Theorem \ref{NoiseAnyPowerThm}}
The proof is reminiscent to the proof of Theorem \ref{c0prime}. A careful
inspection of the proof of Theorem \ref{m11} implies
that the following steps also hold under the weaker assumptions of
Theorem \ref{NoiseAnyPowerThm}:
\begin{align*}
\mathbb{E} \bigl[ | Q_{n} - U_{n} | \bigr] \to 0, \qquad
\mathbb{E} \bigl[ | U_{n} - R_{n} | \bigr] \to 0, \qquad
\mathbb{E} \bigl[ | R_{n} - \Sigma^{*} | \bigr] \to 0, \qquad
\mathbb{E} \bigl[ |\hat{ \Sigma}_{n}^{*} - \Sigma_{n}^{*} | \bigr] \to 0.
\end{align*}
Hence, given the rates on $p$ and $L$, it suffices to show that $\mathbb{E}
\bigl[ | \Sigma_{n}^{*} - Q_{n} | \bigr] \to 0$. Following Lemma 4 -- 5 of
\citet*{podolskij-vetter:09a}, we can prove Lemma \ref{s1} under the assumptions
of Theorem \ref{NoiseAnyPowerThm}. However, we note that the right-hand side
estimate of Lemma \ref{s1}(c) changes from $L p^2/n$ to $p/\sqrt{n}$, because
the estimate in Eq. \eqref{s5} is $p^{2}/n^{2}$ instead of $p^{4}/n^{3}$.
Then, we finish the proof by an additional condition $L/p \to \infty$, which
implies that the right-hand side estimate of Lemma
\ref{s1}(a) dominates that of Lemma \ref{s1}(c). \qed

\subsection{Proofs of Theorems \ref{c1} and \ref{bipowerconsistency}}
We can prove Theorem \ref{c1} by employing the techniques that are used in the proof
of Theorem \ref{c0} and \ref{m11}. Hence, here we only sketch the main parts
that enable us to find the convergence rate. We define:
\begin{align*}
\Sigma_{n} &= \frac{1}{L} \sum_{l=1}^{L} \Biggl( \sqrt{ \frac{n}{L}} \left(
V_l(f,g) - V(f,g) \right) \Biggr)^{2}, \quad
Q_{n}= \frac{p^{2}}{n} \sum_{l=1}^{L} \Biggl( \sum_{i=1}^{n/L p}
\chi_{(i-1)L+l}^{n} \Biggr)^{2}, \\[0.25cm]
U_{n} &= \frac{p^{2}}{n} \sum_{l=1}^{L} \sum_{i=1}^{n/L p} \bigl(
\chi_{(i-1)L+l}^{n} \bigr)^{2}, \quad
R_{n}= \frac{p^{2}}{n} \sum_{l=1}^{L} \sum_{i=1}^{n/L p} \mathbb{E}
\Bigl[ \bigl( \chi_{(i-1)L+l}^{n} \bigr)^{2} \mid \mathcal{F}_{
\frac{(i-1)L+l-1}{n}} \Bigr],
\end{align*}
with
\begin{equation*}
\eta_{i}^{n} = \frac{1}{p-1} \sum_{m \in B_{i}(p)}f \Bigl( \sqrt{n}
\sigma_{ \frac{(i-1)p}{n}} \Delta_{m}^{n} W \Bigr) g \Bigl( \sqrt{n}
\sigma_{ \frac{(i-1)p}{n}} \Delta_{m+1}^{n} W \Bigr) ~\text{   and   }~
\chi_{i}^{n} = \eta_{i}^{n} - \mathbb{E} \Bigl[ \eta_{i}^{n} \mid
\mathcal{F}_{ \frac{(i-1)p}{n}} \Bigr].
\end{equation*}
There exists a $C > 0$, independent of $i$, such that
\begin{equation}
\mathbb{E} \bigl[( \eta_{i}^{n})^{4} \bigr] \leq C \qquad \text{ and }
\qquad \mathbb{E} \bigl[ (\chi_{i}^{n})^{4} \bigr] \leq \frac{C}{p^{2}},
\end{equation}
where the last inequality holds, because $\chi_{i}^{n}$ is a sum of
$1$-dependent random variates. Then, we are done due to the relationship
$p/n \ll \sqrt{p/n} \ll p/\sqrt{n} \ll 1/\sqrt{L}$ (the latter follows
from $n/Lp^2 \to \infty$) and the following Lemma. We omit the proof
for the sake of brevity.
\begin{lemma}
\label{m31}
Assume that the conditions of Theorem \ref{c1} are fulfilled. Then, we get that
\begin{align}
\mathbb{E} \bigl[ | \Sigma_{n} - Q_{n} | \bigr] &\leq C \biggl(
\frac{L p^{2}}{n} + \frac{ p}{ \sqrt{n}} \biggr), \tag{A.9a}
\\[0.25cm]
\mathbb{E} \bigl[ | Q_{n} - U_{n} | \bigr] &\leq C \frac{1}{ \sqrt{L}},
\tag{A.9b} \\[0.25cm]
\mathbb{E} \bigl[ | U_{n} - R_{n} | \bigr] &\leq C \sqrt{\frac{p}{n}},
\tag{A.9c} \\[0.25cm]
\mathbb{E} \bigl[ | R_{n} - \Sigma | \bigr] &\leq C \biggl( \frac{p}{n}
+ \frac{1}{p} \tag{A.9d} \biggr), \\[0.25cm]
\mathbb{E} \bigl[ | \hat \Sigma_{n} - \Sigma_n | \bigr] &\leq C \frac{1}{\sqrt{L}}.
\tag{A.9e}
\end{align}
\end{lemma}
Lastly, the proof of Theorem \ref{bipowerconsistency} is equivalent to the proof of
Theorem \ref{c0prime} and \ref{NoiseAnyPowerThm}, and we omit it. \qed

\subsection{Proof of Theorem \ref{theorem:truncation}}
We denote with $X'$ the continuous part of $X$
and introduce the following approximation of $\hat{ \Sigma}_{n}$:
\begin{equation*}
\hat{ \Sigma}_{n}' = \frac{1}{L} \sum_{l=1}^{L} \Biggl( \sqrt{ \frac{n}{L}}
\Bigl( V_l'(q,r)^n-V'(q,r)^n \Bigr) \Biggr)^{2},
\end{equation*}
where
\begin{align*}
\begin{split}
V'(q,r)^{n} &=\frac{1}{n}\sum_{i =
1}^{n -1} | \sqrt{n} \Delta_{i}^{n} \check{X}'|^{q} | \sqrt{n} \Delta_{i+1}^{n}
\check{X}'|^{r}, \\[0.25cm]
V_{l}'(q,r)^{n} &= \frac{Lp}{n} \sum_{i = 1}^{n/L p} v'_{(i-1)L+l}(q,r)^{n},
\\[0.25cm]
v'_{i}(q,r)^{n} &= \frac{1}{p-1} \sum_{j,j+1 \in B_{i}(p)} | \sqrt{n}
\Delta_{j}^{n} \check{X}' |^{q} | \sqrt{n} \Delta_{j+1}^{n} \check{X}' |^{r}.
\end{split}
\end{align*}
The proof of Theorem \ref{bipowerconsistency} implies that $\hat{ \Sigma}_{n}'
\overset{p}{ \to} \Sigma$. So, it suffices to show that $\hat{ \Sigma}_{n} -
\hat{ \Sigma}_{n}' \overset{p}{ \to} 0$.
Note that
\begin{align*}
\hat{ \Sigma}_{n} - \hat{ \Sigma}_{n}' = \frac{1}{L} & \sum_{l=1}^{L} \Bigg(
\sqrt{ \frac{n}{L}} \Big( V_{l}(q,r)^{n} - V_{l}'(q,r)^{n} + V'(q,r)^{n} -
V(q,r)^{n} \Big) \Bigg) \\[0.25cm]
& \ \ \, \times \Bigg( \sqrt{ \frac{n}{L}} \Big( V_{l}(q,r)^{n} - V(q,r)^{n} +
V_{l}'(q,r)^{n} - V'(q,r)^{n} \Big) \Bigg).
\end{align*}
For any $j \geq 1$, we set:
\begin{equation*}
\bar{ \eta}_{j}^{n} = | \sqrt{n} \Delta_{j}^{n} \check{X}|^{q} | \sqrt{n}
\Delta_{j+1}^{n} \check{X}|^{r} - | \sqrt{n} \Delta_{j}^{n} \check{X}'|^{q}
| \sqrt{n} \Delta_{j+1}^{n} \check{X}'|^{r}.
\end{equation*}
Applying Eq. (13.2.21) from \citet*{jacod-protter:12a} with $m = 1 + \epsilon$
and $\theta=0$, we find that:
\begin{equation*}
\mathbb{E} \big[| \bar{ \eta}_{j}^{n}|^{1 + \epsilon}| \big] \leq
\frac{1}{n^{(1 + \epsilon)/2}} \phi_{n},
\end{equation*}
uniformly in $j$, for some sequence $\phi_{n}$ going to $0$ and $\epsilon \in
(0,1 - \beta]$, and under $\check{ \omega} \geq (m s' + \epsilon -
1)/2(ms' - \beta)$. Then, the discrete H\"{o}lder inequality implies that
\begin{align*}
\sup_{1 \leq l \leq L} \mathbb{E} \Biggl[ \Bigl| \sqrt{ \frac{n}{L}} \Big(
V_{l}(q,r)^{n} - V_{l}'(q,r)^{n} \Big) \Bigr|^{1+\epsilon} \Biggr] \to 0 \quad
\text{and} \quad \sup_{1 \leq l \leq L} \mathbb{E} \Biggl[ \Bigl| \sqrt{
\frac{n}{L}} \Big( V(q,r)^{n} - V'(q,r)^{n} \Big) \Bigr|^{1+\epsilon} \Biggr]
\to 0.
\end{align*}
Applying the arguments of Lemma \ref{l1}, we also have that
\begin{equation*}
\sup_{1 \leq l \leq L} \mathbb{E} \Biggl[ \Big| \sqrt{ \frac{n}{L}} \Big(
V_{l}'(q,r)^{n} - V'(q,r)^{n} \Big) \Big|^{1+\frac{1}{ \epsilon}} \Biggr] +
\sup_{1 \leq l \leq L} \mathbb{E} \Biggl[ \Big| \sqrt{ \frac{n}{L}} \Big(
V_{l}(q,r)^{n} - V(q,r)^{n} \Big) \Big|^{1+\frac{1}{ \epsilon}} \Biggr] \leq C.
\end{equation*}
Therefore, again by the H\"older inequality,
\begin{equation*}
\mathbb{E} \Big[ | \hat{ \Sigma}_{n} - \hat{ \Sigma}_{n}'| \Big] \to 0.
\end{equation*}
As $\epsilon > 0$ can be chosen as small as possible, the proof is complete. \qed

\subsection{Selected elements of Malliavin calculus}
\label{secA.5}
Here, we introduce some concepts from Malliavin calculus. The interested
readers are referred to \citet*{ikeda-watanabe:89a} and \citet*{nualart:06a}
for more thorough textbooks on this subject. Set $\mathbb{H} =
\mathbb{L}^2([0,1], \text{d}x)$ and let $\langle \cdot, \cdot
\rangle_{\mathbb H}$ denote the scalar product on $\mathbb{H}$ and $W =
\{ W(h) : h \in \mathbb{H} \}$ be an isonormal Gaussian family indexed by
$\mathbb{H}$, i.e. the random variables $W(h)$ are centered Gaussian with a
covariance structure determined via
\begin{equation*}
\mathbb{E} \bigl[ W(g) W(h) \bigr] = \langle g,h \rangle_{\mathbb{H}}.
\end{equation*}
In our setting, $W(h) = \int_{0}^{1} h_{s} \text{d}W_{s}$, where $W$ is a
Brownian motion. The set of smooth random variables is introduced with
\begin{equation*}
\mathcal{S} = \Bigl\{ F = f \bigl( W(h_{1}), \ldots, W(h_{n}) \bigr), n \geq
1, h_{i} \in \mathbb{H} \Bigr\},
\end{equation*}
where $f \in C_{p}^{ \infty} (\mathbb{R}^{n})$ (i.e., the space of infinitely
often differentiable functions such that all derivatives exhibit polynomial
growth). The $k$th order Malliavin derivative of $F \in \mathcal S$, denoted by
$D^{k} F$, is defined by
\begin{equation} \label{Derivative}
D^{k} F = \sum_{i_{1}, \ldots, i_{k} = 1}^{n} \frac{ \partial^{k}}{
\partial x_{i_{1}} \cdots \partial x_{i_{k}}} f \bigl(W(h_{1}), \ldots,
W(h_{n}) \bigr) h_{i_{1}} \otimes \cdots \otimes h_{i_{k}}.
\end{equation}
The space $\mathbb D_{k,q}$ is the completion of the set $\mathcal{S}$ with
respect to the norm
\begin{equation*}
\| F \|_{k, q} \equiv \left( \mathbb{E} \bigl[ |F|^{q} \bigr] + \sum_{m =
1}^{k} \mathbb{E} \bigl[ \| D^{m} F \|_{ \mathbb{H}^{ \otimes m}}^{q}
\bigr] \right)^{1/q}.
\end{equation*}
If $(X_{t})_{t \in [0,1]}$ is a solution of a stochastic differential equation
(SDE)
\begin{equation*}
\text{d}X_t = a(X_{t})\text{d}t + \sigma(X_{t}) \text{d}W_{t},
\end{equation*}
and $a,\sigma \in C^{1}(\mathbb{R})$, then $DX_{t}$ is given as the solution of
the SDE
\begin{equation}
\label{MalSDE}
D_{s}X_{t} = \sigma(X_{s}) + \int_{s}^{t} a'(X_{u})D_{s}(X_{u}) \text{d}u +
\int_{s}^{t} \sigma'(X_{u}) D_{s}(X_{u}) \text{d}W_{u},
\end{equation}
for $s \leq t$, and $D_{s}X_{t} = 0$, if $s>t$.

The Malliavin derivative satisfies a chain rule: If $F \in \mathbb{D}_{1,2}$
and $g \in C^{1}( \mathbb{R})$, then it holds that
\begin{equation}
\label{ChainRule}
D(g(F))= g'(F)DF.
\end{equation}
Another application of Malliavin calculus is a refinement of the Clark-Ocone
formula. Let $F \in \mathbb D_{1,2}$, then
\begin{equation} \label{clark}
F= \mathbb{E} [F] + \int_{0}^{1}  \mathbb{E}[ D_{t} F \mid \mathcal{F}_{t} ]
\text{d}W_{t}.
\end{equation}
The operator $D^{k}$ possesses an unbounded adjoint denoted by $\delta^{k}$
(also called a multiple Skorokhod integral). The following integration by parts
formula holds: if $u \in \text{Dom}( \delta^{k})$ and $F \in \mathbb D_{k,2}$,
then
\begin{equation}\label{IntegrationByParts}
\mathbb{E} \bigl[ F \delta^{k}(u) \bigr] = \mathbb{E} \bigl[ \langle D^{k} F,
u \rangle_{ \mathbb{H}^{ \otimes k}} \bigr].
\end{equation}

\end{document}